\documentclass{ectj}

\usepackage{bbm, dsfont}

\usepackage{amsfonts,amssymb,graphics,epsfig,verbatim,bm,latexsym,amsmath,url,amsbsy,mathtools}
\RequirePackage{mathspec}
\usepackage{natbib}
\usepackage{booktabs}
\usepackage{array}
\usepackage{tabularx}
\usepackage{threeparttable}
\usepackage{enumitem}

\newtheorem{theorem}{Theorem}
\newtheorem{assumption}{Assumption}
\newtheorem{proposition}{Proposition}
\newtheorem{corollary}{Corollary}
\newtheorem{lemma}{Lemma}

\newtheorem{remark}{Remark}

\renewcommand{\thesection}{\arabic{section}}
\renewcommand{\theequation}{\arabic{section}.\arabic{equation}}
\renewcommand{\thetheorem}{\arabic{section}.\arabic{theorem}}
\renewcommand{\theassumption}{\arabic{assumption}}
\renewcommand{\theproposition}{\arabic{section}.\arabic{proposition}}
\renewcommand{\thecorollary}{\arabic{section}.\arabic{corollary}}
\renewcommand{\thelemma}{\arabic{section}.\arabic{lemma}}
\renewcommand{\theexample}{\arabic{section}.\arabic{example}}
\renewcommand{\theremark}{\arabic{section}.\arabic{remark}}

\usepackage{xcolor}
\definecolor{ssvs}{HTML}{0072B2}
\definecolor{gets}{HTML}{D55E00}
\definecolor{alasso}{HTML}{009E73}

\received{...}
\accepted{...}
\volume{21}

\title[Bayesian Indicator-Saturated Regression]{Bayesian Indicator-Saturated Regression % \\ for Climate Policy Evaluation
}

\author[Konrad, Vashold, Crespo Cuaresma]{Lucas D. Konrad$^{\dagger}$,
 Lukas Vashold$^{\ddagger}$ and Jesus Crespo Cuaresma$^{*}$}

\address{$^{\dagger}$WU Vienna University of Economics and Business,\\ Welthandelsplatz 1, 1020 Vienna (Austria)}
\email{lucas.konrad@wu.ac.at}

\address{$^{\ddagger}$WU Vienna University of Economics and Business}

\address{$^{*}$WU Vienna University of Economics and Business,\\ Wittgenstein Centre for Demography and Global Human Capital, \\ Austrian Institute of Economic Research}

\def\AmSTeX{$\cal A$\kern-.1667em\lower.5ex\hbox{$\cal M$}\kern-.125em
 $\cal S$-\TeX}
\def\BibTeX{{\rm B\kern-.05em{\sc i\kern-.025em b}\kern-.08em
 T\kern-.1667em\lower.7ex\hbox{E}\kern-.125emX}}

\begin{document}

 \begin{abstract}

Structural break detection has emerged as an important tool for assessing the effects of policies in settings where conventional policy evaluation methods might not be applicable.
% for evaluating the effectiveness of climate change mitigation policies. 
In this paper, we introduce a unified Bayesian framework for detecting structural breaks with unknown timing and arbitrary sequence in longitudinal data. 
The proposed setup builds on a indicator-saturated regression design and uses a spike-and-slab prior for selection among indicators. 
We establish that the particular choice of an inverse-moment density as the slab component provides model selection consistency in this model class. 
% We further establish that this setup provides model selection consistency. 
Simulation results show that the method outperforms comparable frequentist approaches, particularly in environments with a high probability of structural breaks. 
We illustrate the proposed framework by analysing climate policies in the European road transport sector.

 %of on the other hand has a long tradition in the field of econometrics.
 % Breaks in time series or longitudinal data can be indicative for shifts in the structural process generating these data. There exist many approaches for the search and estimation of the magnitude of such shifts but they often require a large temporal dimension or lack a coherent approach to assess break uncertainty.
 % The search for such shifts can be cast in the form of indicator-saturated regressions.
 % Recent approaches in the frequentist paradigm rely on block-search algorithms that lack a proper framework of break uncertainty quantification.


 \end{abstract}

 % main body

 \section{Introduction}

Conventional program evaluations are typically based on the ex-post analysis of interventions with known implementation dates. Inference is based on the comparison of a subset of units that is affected by the intervention (the treated group) to units that are not affected (the control group). Under the assumption of conditional ignorability of treatment assignment, this framework delivers consistent estimates of causal treatment effects \citep[see, for example,][]{imbens2009}. 
In many cases, policies are not implemented (quasi-)randomly or in isolation, and maintaining this assumption may not be viable. For example, climate change mitigation policies are typically enacted at the national or supranational level, target multiple emissions-generating sectors simultaneously, and interact with other regulatory and fiscal instruments. 
This institutional environment limits the availability of quasi-experimental variation in policy exposure and thereby complicates the application of conventional causal inference methods based on cross-sectional or temporal contrasts.

Given these identification challenges, an indirect strategy for uncovering effective policies (or mixes thereof) is to rely on structural break detection methods for time series or panel data.\footnote{
    The literature on dating structural breaks at unknown points in time has evolved substantially over the last decades and includes general tests for parameter instability \citep{Andrews1993,AndrewsPloberger1994}, methods for estimating and testing multiple breaks \citep{Bai1997,BaiPerron1998,BaiPerron2003}, and Bayesian approaches that use posterior simulation and model comparison to infer break locations and regime-specific parameters \citep{BarryHartigan1993,Chib1998,WangZivot2000,KoopPotter2007}. \citet{CasiniPerron2019} provides a comprehensive review of this line of research.
}
An emerging strand of the literature operationalizes policy effects as statistically significant level shifts detected using indicator-saturated regression (ISR) methods \citep{castle2015, pretis_discovering_2026}.
The breaks detected are subsequently attributed to specific policy interventions or bundles of measures using institutional and historical information about the implementation of policies.
ISR provides a general-to-specific modelling strategy that enables the detection of multiple structural breaks without prior knowledge of their timing.
This approach has been applied to global climate policy assessments \citep{stechemesser2024climate}, to the study of mitigation policy effectiveness at the sectoral level \citep{koch2022} and at the subnational level \citep{pretis2022b}, as well as to the construction of harmonized datasets of breakpoints in greenhouse gas emissions  \citep{tebecis2025dataset}.
% Recent contributions to the empirical identification of climate change mitigation policies exploit systematic variation in carbon emissions data to detect structural changes associated with policy interventions \citep{koch2022, stechemesser2024climate, tebecis2025dataset}. 
% Within this framework, policy effects are operationalized as statistically significant level shifts detected using indicator-saturated regression (ISR) methods \citep{castle2015, pretis_discovering_2026}. 

In this paper, we develop a Bayesian extension of the ISR framework tailored to high-dimensional structural break detection in time series or panel data. While existing contributions rely on frequentist selection procedures within ISR \citep{castle2015, pretis_discovering_2026}, we embed break detection in a fully probabilistic model selection framework. As in \citet{castle2015} and \citet{pretis_discovering_2026}, our approach considers an unknown number of structural breaks of unknown timing%within a two-way fixed effects specification
, but departs from these contributions by introducing a Bayesian setup with a spike-and-slab prior that combines a Dirac spike and a non-local slab component. We show that this prior structure is necessary to obtain consistent model selection within the ISR framework. Our setting penalises weak signals more aggressively than conventional local priors, and provides posterior uncertainty quantification within a unified probabilistic framework. Our particular choice of non-local prior, the inverse moment (iMOM) prior \citep[see][]{rossell2012nlp_modsel}, accumulates evidence against a spurious break at a near exponential speed. This keeps the large amount of potential excess breaks in check, while its Cauchy-like tails impose minimal shrinkage on large signals \citep[in line with][]{jeffreys1998, bayarri2007extending}. In addition, the model explicitly accommodates the presence of outlying observations, improving robustness in applications with noisy data while remaining computationally tractable in high-dimensional settings.

The proposed methodological framework is related to the literature on Bayesian change-point models.\footnote{
    Examples of this literature include product-partition approaches \citep{BarryHartigan1993}, latent-state and simulation-based methods for multiple change-point models \citep{Chib1998,GiordaniKohn2008}, or flexible Bayesian econometric models with multiple structural breaks that employ hierarchical priors \citep{KoopPotter2007}. 
} These methods aim at placing a posterior distribution on the number and timing of regime changes and are particularly well suited to settings in which structural change can be represented as a relatively low-dimensional segmentation problem of the time series into distinct regimes.
Our approach employs a different parametrization of structural change. We represent potential shifts through a saturated set of unit-specific step-shift indicators in a regression model and treat their detection as a high-dimensional model selection problem. Such a flexible formulation is particularly useful in settings where structural changes might follow an individual break process for each observation --- implying potentially many breaks of different size in any direction --- or where posterior evidence for a structural change is spread across (neighbouring) candidate break dates indicating gradual breaks. 

We evaluate the finite-sample performance of the proposed estimator in a comprehensive simulation study and benchmark it against leading frequentist procedures for indicator saturated regression. The results indicate that our approach performs competitively in sparse designs and exhibits clear advantages in environments characterised by multiple or closely spaced breaks.
Applying our method to European road transport emissions, we replicate the analysis of \citet{koch2022}. While our results broadly confirm the principal conclusions of \citet{koch2022}, we also detect additional policy interventions associated with reductions in transport emissions.

The remainder of the paper is organised as follows. Section~\ref{sec:methodology} presents the proposed framework, which we term Bayesian Indicator Saturated Modelling (BISAM). Section~\ref{sec:simulation} evaluates its finite-sample performance using simulated data and benchmarks it against existing frequentist approaches to structural break detection. Section~\ref{sec:applications} reports the empirical application to European road transport emissions and discusses the implications for the detection of effective climate mitigation policies. Section~\ref{sec:conclusion} concludes.

\section{The Bayesian Indicator Saturated Model}\label{sec:methodology}

In this section, we present our methodological framework for the detection of structural breaks in panel data settings. We formalise the problem of detecting an unknown number of structural breaks at unknown timing within an ISR framework. %, an extremely flexible model class. % widely used in applications to climate policy evaluation due to its interpretability \citep{koch2022, stechemesser2024climate}.

\subsection{Indicator Saturated Regressions: Specification}\label{subsec:bisam}

The literature on structural break detection in time series is extensive, yet the majority of existing approaches either require long time horizons % ---a luxury rarely afforded by macro-environmental panel data---
or impose strong restrictions on the break process. Our modelling framework uses a break function that takes the form of indicator saturation, accommodating practically every combination of transitions while performing well on relatively short time horizons.

Let $i\in \{1, \dots, N\}$ index cross-sectional units and $t\in\{1, \dots, T\}$ index time periods. 
For the observed panel $\{(y_{i,t},\, x_{i,t})\}_{i=1,\dots,N;\; t=1,\dots,T}$ with $y_{i,t} \in \mathbb{R}$ and $x_{i,t} \in \mathbb{R}^{p}$, we define the \emph{general break model}
\begin{equation}
 y_{i,t}
 = f(x_{i,t})
 + \sum_{s \in \mathcal{S}} \gamma_{s}\, d_{s}(i,t)
 + \varepsilon_{i,t},
 \label{eq:mod_general}
\end{equation}
where $f \colon \mathbb{R}^p \to \mathbb{R}$ is a mean function modelling the response under the absence of breaks. $\mathcal{S}$ is a (potentially very large) candidate set of indicator indices, with $q\coloneqq\lvert\mathcal{S}\rvert$ denoting the size of this set, i.e., the number of all possible breaks. The associated break coefficients are denoted $\gamma_s \in \mathbb{R}$, and $d_s \colon \{1, \dots, N\} \times \{1, \dots, T\} \to \{0, 1\}$ are indicator functions selecting the units and periods affected by break~$s$.
The disturbance term $\varepsilon_{i,t}$ captures idiosyncratic shocks. 
% For notational convenience, let $q$ denote the number of all possible breaks in the model.

By choosing $\mathcal{S}$ and the form of $d_s$, this formulation nests a variety of saturation approaches, including impulse indicator saturation (IIS), step-shift indicator saturation (SIS), and trend indicator saturation (TIS). Following \citet{pretis_discovering_2026}, we focus on a particular specification of the model in equation \eqref{eq:mod_general} and impose a linear mean function $f(x_{i,t}) = x_{i,t}' \beta$. The break function is defined making use of step-shift indicators
$d_{j,s}(i,t) = \mathbbm{1}_{\{i = j\}} \, \mathbbm{1}_{\{t \geq s\}}$, and indexing $\mathcal{S}$ by pairs $(j, s)$ with $j \in \{1, \dots, N\}$ and $s \in \{3, \dots, T-1\}$.\footnote{
 Note that we have $s \notin \{1,2,T\}$ as we generally consider fixed effects as well as the presence of outliers. For $s=1$, including the step-shift indicator $d_{j,1}(i,t) = \mathbbm{1}_{\{i = j\}} \, \mathbbm{1}_{\{t \geq 1\}}$ introduces a term that is perfectly collinear with the cross-sectional fixed effect for unit $i$, preventing separate identifiability. Breaks at $s=2,T$ on the other hand are equivalent to outlying observations at $t=1,T$ that are in our framework separately identified by and treated within the framework introduced in Section~\ref{subsec:outlier} below.
 % Note that we have $s \notin \{1,2,T\}$ as we generally consider cross-sectional and time fixed effects. In this case, for $j=i$, including $d_{1}(i,t) = \mathbbm{1}_{t \geq 1} = \mathbbm{1} \, \forall \, t$ or $d_{T}(i,t) =  \mathbbm{1}_{t \geq T} = \mathbbm{1}_{t = T}$ introduces terms that are perfectly collinear with the unit or time fixed effect, preventing separate identifiability. 
 % For $s=2$, since the indicator $d_{2}(i,t)$ can be expressed as $d_{2}(i,t) = \mathbbm{1}_{\{t \geq 2\}} = \mathbbm{1}_{t \geq 1} - \mathbbm{1}_{\{t = 1\}} = d_{1}(i,t) - \mathbbm{1}_{t = 1}$, which is a linear combination of the unit and a time fixed effect, it is also not separately identifiable.
} This yields the SIS linear model
%Following \citet{pretis_discovering_2026}, we cast the search for structural breaks in terms of a step-shift indicator-saturated (SIS) linear model,
\begin{equation}
 y_{i,t} = x_{i,t}' \beta + \sum_{j = 1}^N \sum_{s = 3}^{T-1} \gamma_{j,s} \mathbbm{1}_{\{j = i\}} \mathbbm{1}_{\{t \geq s\}} + \varepsilon_{i,t}.
 \label{eq:mod_simple}
\end{equation}
%where $y_{i,t}$ denotes the outcome of interest for unit $i \in \{1, \dots, N\}$ at time $t \in \{1, \dots, T\}$ \citep[for example, the log of carbon dioxide emissions as in][]{koch2022, stechemesser2024climate}. 
The vector $x_{i,t}$ %is $p$-dimensional and 
may include cross-sectional fixed effects, time fixed effects, and additional covariates associated with the outcome $y_{i,t}$. 
The indicator functions $\mathbbm{1}_{\{j = i\}}$ and $\mathbbm{1}_{\{t \geq s\}}$ define unit-specific step-shift indicators that take the value one for unit $i=j$ and periods $t \geq s$, respectively, and zero otherwise.

In its unrestricted form, the model in equation~\eqref{eq:mod_simple} contains $N(T-3)+p + 1$ free parameters to be estimated from $NT$ observations, assuming homoskedasticity and uncorrelated errors across units. This high-dimensional setting gives rise to a challenging model selection problem in practice. In a classical framework, strategies to address this curse of dimensionality include general-to-specific (GETS) block-search algorithms \citep{santos2008automatic}, as employed by \citet{castle2012model, castle2015} and \citet{pretis_discovering_2026}, as well as regularisation techniques such as the (adaptive) Least Absolute Shrinkage and Selection Operator (LASSO) \citep{tibshirani1996}, used by \citet{li2016}, \citet{qian2016} and \citet{pretis_discovering_2026}.

We propose a fully Bayesian approach designed to address the high-dimensional model selection problem inherent in such indicator-saturated regressions. Bayesian Indicator Saturated Modelling (BISAM) builds on recent advances in Bayesian variable selection. %, achieving %a faster concentration of the Bayes factors through non-local priors \citep{rossell2010nlp_test, rossell2012nlp_modsel}, facilitating 
%model consistency for $q=O(NT)$ as is necessary in an ISR framework \citep[see][]{rossell2010nlp_test, rossell2012nlp_modsel}. 
The coherent probabilistic framework avoids sequential estimation and allows for efficient structural break detection as well as exact posterior uncertainty quantification.

% %---------------- Figure 1 ------------------
% \begin{figure}[t]
% \centering
% \includegraphics[width = \linewidth, trim = {0 0 0 0}, clip]{img/plot2.pdf}
% \vspace*{-10mm}
% \caption{Illustration of step-shift indicator saturation for a single unit to detect the true structural break in $t = 14$ (red line). Blue lines in the left panel depict all potential breaks for which indicators are included in the model. Blue lines in the left panel illustrate the estimated break magnitude when picking the wrong break date, whereas the purple line depicts the bias in estimated break magnitude when ignoring that the observation in period $t = 15$ is an outlier.}
% \label{fig:step_illustration}
% \end{figure}
%--------------------------------------------

% Using the model as specified in \autoref{eq:mod_simple}, the detection of boils down to a selection over them and subsequent estimation of its corresponding magnitude. % to retain only those breaks that correspond to the subset of treated units and time periods

%\subsection{Bayesian Indicator Saturated Modeling with Spike-and-slab Priors}\label{subsec:bisam}

%\subsection{Spike-and-Slab with Inverse Moment Priors}\label{subsec:bisam1}
\subsection{Priors on the Break Function}\label{subsec:bisam1}

The break-size coefficients are modelled using a variant of the Stochastic Search Variable Selection (SSVS) prior. Every potential step-shift is included in the regression, and the estimation procedure enacts the selection of these. For this purpose, we consider a family of spike-and-slab mixture priors on $\gamma_{i,s}$ \citep{mitchell1988, george1993, george1997, ishwaran2005}. In the context of equation~\eqref{eq:mod_simple}, this can be written as
\begin{equation}
 p(\gamma_{i,s} \!\mid\! \delta_{i,s}^{\gamma}) 
 = (1 - \delta_{i,s}^{\gamma}) p_{spike}(\gamma_{i,s}) 
 + \delta_{i,s}^{\gamma} p_{slab}(\gamma_{i,s}),
 \label{eq:ssvs_pri}
\end{equation}
where the binary selection indicator $\delta_{i,s}^{\gamma} \in \{0,1\}$ governs whether the coefficient $\gamma_{i,s}$ stems from $p_{spike}$, the spike density that is typically concentrated at zero, or from $p_{slab}$, the slab density that allows for deviations from zero. Under this specification, structural break detection is naturally understood as variable selection, or equivalently, a posterior inclusion problem: a step-shift is selected whenever $\mathbb{E}(\delta_{i,s}^{\gamma}\!\mid\! y) > P$, for some given probability threshold $P \in (0,1)$.

On the selection indicator $\delta_{i,s}^{\gamma}$ we impose a Bernoulli prior,
\begin{equation}
 p(\delta_{i,s}^{\gamma}) 
 = \omega_{i,s}^{\delta_{i,s}^{\gamma}} 
 (1 - \omega_{i,s})^{1 - \delta_{i,s}^{\gamma}},
\end{equation}
where $\omega_{i,s}$ denotes the prior inclusion probability of the corresponding step-shift indicator. This parameter may be fixed ex ante or treated hierarchically. % and learned from the data.
For example, in the absence of strong prior information on the proportion of structural breaks in the data, it can be specified as constant $\omega_{i,s}=0.5 \, \forall \, i,s$. As detailed further below, this yields the median probability model \citep{barbieri2004optimal, BarbieriBerger2021} for ex-post break selection.
Alternatively, one can learn the sparsity level from the data by assigning a hyperprior such as a Beta distribution on $\omega_{i,s}$ \citep[see, for example,][]{scott2010bayes}. Yet, such a specification might suffer from poor mixing and convergence if the data is only weakly informative.

The flexibility of the break function as given by equation~\eqref{eq:mod_simple} places particular requirements on the prior design of $\gamma_{i,s}$. 
%The search for structural breaks in a specification such as that given by equation~\eqref{eq:mod_simple} places specific requirements on the prior design for $\gamma_{i,s}$, due to the flexibility of the break function. 
First, the prior should induce a clear separation between inclusion and exclusion in the posterior, thereby sharpening break detection while shrinking weak signals aggressively. Second, conditional on inclusion, the slab component should allow for flexible estimation of the break magnitude without imposing excessive shrinkage on large signals.
Third, the overall prior specification should preserve model selection consistency, which in the high-dimensional ISR framework requires consistency in a regime where the number of candidate break parameters grows linearly in $N$ and $T$. 
% Since in such a setting neighbouring step-shift indicators are almost perfectly collinear, consistency is naturally formulated at the level of break \emph{windows}. We establish this window-level consistency argument in Appendix~\ref{app:theory}. 

% \begin{definition}[Non-local priors]
% A prior $\pi(\gamma_{j,s})$ on the break coefficients in \eqref{eq:mod_simple} is a
% \emph{non-local prior} (NLP) if it satisfies
% %
% \begin{equation}
% \lim_{\lvert\gamma_{j,s}\rvert \to 0} \pi(\gamma_{j,s}) = 0,
% \end{equation}
% %
% that is, the prior density vanishes at the null value $\gamma_{j,s} = 0$
% \citep[see][for details]{rossell2017nonlocal}.
% \end{definition}

To satisfy these three requirements, we specify a prior with a Dirac point mass at zero for the spike component and adopt a non-local prior density for the slab component. The Dirac prior for $p_{spike}$ sets shifts to \textit{exactly} zero if $\delta_{i,t}=0$, virtually excluding the break from estimation and thereby adding to the specificity of break detection.\footnote{
 Such a choice, however, increases computational complexity when comparing marginal likelihoods across models that differ in their step-shift indicators, even when Laplace approximations are used. Alternatively, a highly concentrated Gaussian prior for the spike component can be used \citep[see, for example,][]{george1997, ishwaran2005}. Such an approach has the drawback of not fully excluding step-shift indicators, thus decreasing precision of separation for included and excluded breaks.
}
The choice of a non-local prior as slab component further ensures separation from the spike by concentrating probability mass away from zero. Formally, a prior on the break coefficients in equation~\eqref{eq:mod_simple} is non-local if its density $\pi(\gamma_{j,s})$ satisfies
\begin{equation} \label{nlp}
\lim_{\lvert\gamma_{j,s}\rvert \to 0} \pi(\gamma_{j,s}) = 0,
\end{equation}
that is, the prior density vanishes at $\gamma_{j,s}=0$ \citep[see e.g.][for details]{rossell2017nonlocal}. %In contrast, under a local prior density $\pi_L$ with $\pi_L(0)>0$, the Bayes factor in favour of including a break decays only at rate $\mathrm{BF}_T(1\!\mid\!0) = O_p(T^{-1/2})$ when the true coefficient is zero \citep{rossell2010nlp_test}. This slow learning rate implies weak penalization of spurious breaks, particularly in panels with a short time horizon. Non-local priors address this limitation by assigning zero density at the origin, which accelerates the decay of evidence for negligible coefficients. 
% The second requirement of the prior setting, low shrinkage of genuinely large breaks, is governed by the tail behaviour of the slab component. Heavy tails mitigate the bias that proper priors introduce in finite samples by reducing shrinkage of large coefficients \citep[see][]{jeffreys1998}. 
Among non-local prior families, the inverse-moment (iMOM) prior is particularly well suited in our setting, as it combines a rapidly vanishing density at zero with heavy tails that preserve large signals. Its density is given by
\begin{equation}
 \pi_{iMOM}(\gamma_{i,s} \!\mid\! \gamma_0, k, \nu, \tau \sigma^2) 
 = \frac{k (\tau \sigma^2)^{\nu/2}}{\Gamma(\nu/2k)} 
 \lvert \gamma_{i,s} - \gamma_0 \rvert^{-(\nu+1)} 
 \exp \left\{ -\left(\frac{\gamma_{i,s} - \gamma_0}{\sqrt{\tau} \sigma}\right)^{-2k} \right\},
 \label{eq_pimom}
\end{equation}
where $\gamma_0$ denotes the centering value, $\nu > 0$ is a shape parameter, $k > 0$ controls the order of the prior, $\tau$ is a scale hyperparameter, $\sigma^2$ is the error variance and $\Gamma$ denotes the gamma function. The parameters $\nu$ and $k$ jointly determine the tail behaviour and the rate at which density vanishes near $\gamma_0$. In our baseline specification, we set $\gamma_0 = 0$ and $\nu = k = 1$, yielding Cauchy-like tails. See Figure~\ref{fig:tau_calibration} for an illustration of the iMOM density with different scale hyperameters $\tau$.

The suitability of non-local priors in the ISR design can be justified by studying the behaviour of Bayes factors in this setting. Under a local prior density $\pi_L$ with $\pi_L(0)>0$, the Bayes factor in favour of a spurious break decays only at rate $\mathrm{BF}_T(H_A\!\mid\!H_0\!=\!true) = O_p(T^{-1/2})$ when the true coefficient is zero, irrespective of its calibration \citep{rossell2010nlp_test}. This implies a relatively weak penalty for negligible breaks.%, particularly in panels with a short time horizon. 
% Non-local priors accelerate this decay by assigning zero density at the origin. 
% The choice of a non-local slab sharpens this further. 
% Under a pMOM slab of order $k$ the spurious-break Bayes factor decays polynomially, $\mathrm{BF}_T(H_A\!\mid\!H_0) = O_p(T^{-k-1/2})$, whereas under the iMOM (piMOM) slab the decay is near-exponential,
Under an iMOM slab, the decay is near-exponential,
\begin{equation}
 \frac{\log \mathrm{BF}_T(H_A\!\mid\!H_0)}{T^{k/(k+1)}}
 \overset{p}{\to} c < 0,
 \label{eq:pimom_rate}
\end{equation}
so that evidence against a negligible break accumulates far faster than under any prior featuring a local slab. %, whose rate stays $O_p(T^{-1/2})$ irrespective of its calibration \citep[eq.~11]{rossell2010nlp_test}. 

 For $|\gamma_{j,s}| \gg 0$ on the other hand, the exponential term becomes negligible and the density decays polynomially at rate $|\gamma_{j,s}|^{-2}$ \citep[see][]{rossell2010nlp_test}, so that the induced bias on large break estimates is of order $O(1/(T\,\lvert\gamma_{j,s}\rvert^3))$ \citep[see Proposition~2 in][]{rossell2017nonlocal}, which is negligible relative to the $O_p(T^{-1/2})$ sampling variability of the least-squares estimator. Thinner-tailed slab priors, by contrast, may impose excessive shrinkage on large breaks, particularly in short samples.
Thus, this structure combines strong penalization of coefficients near zero with minimal shrinkage of large coefficients, satisfying the first two requirements mentioned about.

Moreover, the characteristic that the density of non-local priors vanishes rapidly at $\gamma_{j,s}=0$ is crucial to meet the third requirement, model selection consistency. 
Since in the ISR setting neighbouring step-shift indicators are almost perfectly collinear, consistency is naturally formulated at the level of break \emph{windows}. 
We establish this window-level consistency argument in Appendix~\ref{app:theory} with a more detailed exposition and associated proofs in Sections~\ref{app:nlp_rates}~and~\ref{app:gauge} of the Online Appendix, respectively. 
In particular, we show that non-locality is not merely a refinement of the slab in our setting but a requirement for false-discovery control once the number of candidate breaks grows at $O(T)$. 
When a break is genuinely present, non-local slabs accumulate evidence at the same exponential rate in $T$, so the additional stringency of them is costless on true breaks. Figure~\ref{fig:prior_rates} confirms both statements directly on the step-shift design. This means that, while exact date recovery is not achievable at fixed break magnitudes owing to the near-collinearity of adjacent step indicators (see Proposition~\ref{prop:resolution} of the Online Appendix), non-locality remains necessary for window-level model selection consistency within the ISR framework.\footnote{
    This argument concerns candidate dates which are separated from every true break, so that collinearity with the true break structure is asymptotically negligible. The complementary problem of \emph{localising} a detected break among its near-collinear neighbours is handled separately by the window rule (see Theorem~\ref{thm:window_consistency} in Appendix~\ref{app:theory}).
}
This feature provides a clear (asymptotic) advantage of BISAM over its frequentist competitors, which is documented in Section~\ref{sec:simulation}.

% In the indicator saturation framework, the number of candidate break parameters scales as $O(NT)$, so that even conditioning on a single cross-sectional unit implies an effective model dimension that grows linearly with $T$. \citet{rossell2012nlp_modsel} show that when the number of coefficients increases linearly in the number of observations, non-local priors %(so-called product moment, pMOM, and product-inverse, piMOM, densities)
% ensure posterior consistency, meaning the posterior model probability of the true break model converges to one as the sample size increases. In contrast, under any continuous local prior density $\pi_L$ with $\pi_L(0) > \varepsilon > 0$, the posterior probability of the true model converges to zero whenever the number of candidate parameters exceeds $O(T^{1/2})$ \citep[see Theorem~2 in][]{rossell2012nlp_modsel}.

% Finally, the iMOM prior assigns zero density at its centre value, $\pi_{iMOM}(\gamma_0)=0$, ensuring a clear separation between exclusion and inclusion. Combined with the Dirac spike at zero, this yields a coherent specification in which excluded coefficients are set exactly to zero, while included coefficients are estimated without undue shrinkage toward the null.
The resulting specification is therefore coherent across all of the three desired properties: the Dirac spike sets excluded coefficients exactly to zero; the heavy tails of the iMOM density allow included breaks to be estimated without undue shrinkage toward the null; and the non-local slab ensures rapid discrimination against spurious breaks, in particular for the iMOM prior, and delivers window-level model selection consistency in the high-dimensional ISR setting. 

% \paragraph{Prior Settings}
This leaves the question for prior specification. It is well known that model selection using SSVS is sensitive to the choice of the slab scale parameter. This also applies in our setting. However, the scale parameter $\tau$ of the iMOM prior has a natural and interpretable role, as it directly governs the prior probability that a break magnitude lies within a neighbourhood of zero. This allows explicit control over the prior mass (and thus shrinkage) imposed on small breaks.

%---------------- Figure 2 ------------------
\begin{figure}[t]
 \centering
 \includegraphics[scale = 0.75]{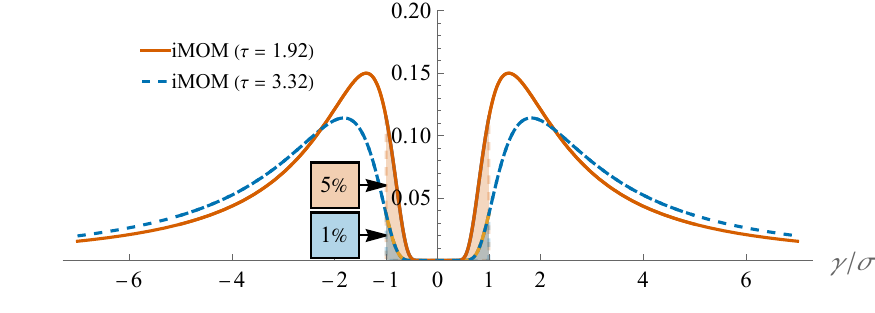}
 \caption{Probability density function of $iMOM(0,1,1,\tau)$ prior for the cases $P(|\gamma_{i,s}|\leq \sigma|\tau)=0.05$ and $P(|\gamma_{i,s}|\leq \sigma|\tau)=0.01$ corresponding to $\tau=1.92$ (-----) and $\tau=3.32$ (- - -), respectively.}
 \label{fig:tau_calibration}
\end{figure}
%--------------------------------------------

This interpretation facilitates calibration based on the prior probability $P(|\gamma_j| \leq \bar{\gamma} \!\mid\! \tau)$ for a chosen threshold $\bar{\gamma} > 0$. For example, setting $\bar{\gamma} = \sigma$ and choosing $\tau = 1.92$ ($\tau = 3.32$) implies that the prior probability of a break lying within $(-\sigma,\sigma)$ is below $0.05$ ($0.01$). Equivalently, this limits the prior probability of a break with signal-to-noise ratio $|\gamma_j|/\sigma \leq 1$ to $0.05$ and $0.01$, respectively. Figure~\ref{fig:tau_calibration} illustrates the implied prior densities for these values of $\tau$.\footnote{
    As for the case of $\omega_{i,s}$, hierarchical prior specifications $\tau$ are possible but mixing and convergence is poor when the data are weakly informative. We thus abstain from such a specification.
}

% The model specification also allows for a hierarchical prior structure, which introduces additional hyperpriors, thus increasing computational complexity. For example, in the absence of strong prior information on the proportion of structural breaks in the data, the prior inclusion probability $\omega_{i,s}$ can be specified as constant $\omega_{i,s}=\omega$. Alternatively, $\omega_{i,s}$ can be treated as unknown and assigned a hyperprior, such as a Beta distribution \citep[see, for example,][]{scott2010bayes}, allowing the sparsity level to be learned from the data. As for the case of $\tau$, hierarchical specifications of the prior for $\omega_{i,s}$ may affect mixing and convergence when the data are weakly informative.

%%% MAKE OWN SUBECTION/PARAGRAPH for uncertainty quantification here? and talk about window stuff in a bit more detail?
\paragraph{Break Selection}
Framing the break detection problem as a model selection exercise allows us to use posterior inclusion probabilities (PIPs) to quantify the evidence for individual break indicators. In our framework, the PIP of a candidate break is the posterior probability that the corresponding break indicator is included in the model, computed as the Monte Carlo average of the posterior draws of the selection indicators $\delta^{\gamma}_{i,s}$. This provides a direct and continuous measure of uncertainty about the presence of a break at a given, individual point in time. 
This probabilistic break uncertainty measure is in itself a major benefit of BISAM, which can be flexibly utilized in ways that best fit specific research problems.
% Such a continuously interpretable break uncertainty measure is in itself a major benefit of BISAM compared to frequentist procedures.
% However, while having a continuously interpretable break uncertainty measure is a major benefit of BISAM compared to frequentist procedures, for practical purposes a heuristic for selecting breaks is useful.

Moreover, structural changes can also be gradual or weak yet persistent, implying that posterior evidence may not be concentrated on a single candidate date but rather spread over adjacent candidate break dates. 
In such cases, a probabilistic measure of break timing that aggregates PIPs across neighbouring time points may be more informative. Consider a time window $W = \{s, s+1, \dots, s + \bar{t} - 1\}$ of length $\bar{t}$ for unit $i$. A natural object of interest may then be the posterior probability that \emph{at least} one break has occurred within the window,
\begin{equation}
 \mathrm{PIP}^{W}_{i} \;\equiv\; P\!\left(A^{W}_{i} \!\mid\! y\right), \qquad
 A^{W}_{i} = \left\{ \exists\, s \in W : \delta^{\gamma}_{i,s} = 1 \right\}.
 \label{eq:window_pip}
\end{equation}
which is again conveniently obtained as the Monte Carlo average using the retained draws of $\delta^{\gamma}_{i,s \in W}$.\footnote{
    Note that this aggregation does not require independence approximation across candidate dates. Where the direction of the break matters for the interpretation of the results, as in our empirical application, the event in \eqref{eq:window_pip} is signed by additionally requiring $\gamma_{i,s \in W} < 0$ (or $\gamma_{i,s \in W} > 0$), using the joint draws of $(\delta^{\gamma}_{i,s \in W}, \gamma_{i,s \in W})$.
}
However, the formulation of \eqref{eq:window_pip} also implies that $\mathrm{PIP}^{W}_{i}$ mechanically increases as $\bar{t}$ grows.
This makes the raw measure of $\mathrm{PIP}^{W}_{i}$ relatively impractical for the selection of time windows in which a break may have occurred.

We propose a simple selection rule for break windows based on $\mathrm{PIP}^{W}_{i}$ that augments a standard classification loss argument with a cost of imprecision that is increasing in their width.
Specifically, let $c_1$ denote the cost of reporting a window that contains no true break, $c_0$ the cost of failing to report a true break, and $\kappa$ the cost of each period of window width in excess of a point detection. The Bayes rule under this loss selects window $W$ whenever
\begin{equation}
 \mathrm{PIP}^{W}_{i} \;>\; P^{\bar{t}} \;=\; \frac{c_1 + \kappa (\bar{t} - 1)}{c_0 + c_1},
 \label{eq:window_rule}
\end{equation}
as shown in Appendix~\ref{app:window_rule}. This rule has three key properties. First, for $\bar{t}=1$ and symmetric classification losses ($c_0 = c_1$), the threshold reduces to $P^{1} = 0.5$. Second, the threshold is monotonically increasing in $\bar{t}$. Third, it implies a maximal reportable window length, since $P^{\bar{t}} \geq 1$ whenever $\bar{t} \geq 1 + c_0/\kappa$.

The first property implies that it nests the median probability model proposed by \citet{barbieri2004optimal} and extended to correlated regressors in \citet{BarbieriBerger2021}. When specifying a uniform prior over all possible model specifications, which is equivalent to setting $\omega_{i,s}=0.5 \forall i,s$, the proposed rule $PIP>0.5$ for an individual break date has a simple posterior-odds interpretation. This means that a break is selected whenever the data shift the posterior support in favour of inclusion rather than exclusion \citep[see][]{salaimartin2004}, constituting the Bayes decision rule for an individual break indicator. The other two properties extend the median probability model to the case of window selection.

Moreover, we show in Appendix~\ref{app:theory} that the rule proposed in equation \eqref{eq:window_rule} is consistent. This implies that, as $T$ grows, $\mathrm{PIP}^{W}_{i}$ tends to one for suitably chosen windows around each true break and to zero for windows further away than $\bar{t}_{\max}$ from every true break. Thus, the rule asymptotically recovers true break windows with probability one (Theorem~\ref{thm:window_consistency}). The result is stated along an asymptotic regime in which the cost ratio $\kappa/c_0$ vanishes slowly, so that the cap $\bar{t}_{\max}\;=\; \left\lceil 1 + c_0/\kappa \right\rceil - 1$ in \eqref{eq:tmax} grows with $T$ while remaining negligible relative to the spacing between true breaks.

As a default, we set $c_0 = c_1 = 1$ and $\kappa = 1/3$, which yields thresholds of $P^{1} = 0.5$, $P^{2} \approx 0.67$ and $P^{3} \approx 0.83$, and a window cap of $\bar{t}_{\max} = \lceil 1 + c_0/\kappa \rceil - 1 = 3$, i.e.\ windows of four periods or more are unattainable. By varying the parameters $c_0$, $c_1$ and $\kappa$, these quantities can be adjusted if one wants to assign different costs to different potential sources of misclassification.

\subsection{Outlier Robustness}\label{subsec:outlier}

The reliable detection of step-shifts may be compromised by outliers, as extreme observations can generate spurious break signals or distort estimated break magnitudes. To enhance robustness, the baseline model is extended by specifying a mixture distribution for the disturbance term $\varepsilon_{i,t}$ and introducing a data augmentation step that probabilistically rescales observations during estimation.\footnote{
Frequentist approaches achieve robustness by introducing impulse indicators and selecting among them via impulse indicator saturation \citep{castle2015, pretis_discovering_2026}, typically in a separate pre-estimation step. In contrast, our approach integrates outlier detection directly within the Bayesian estimation procedure.}

We assume the disturbance term follows the mixture distribution
\begin{align}
 p(\varepsilon_{i,t}) &= (1 - \delta^\varepsilon_{i,t}) \, \pi_N(\varepsilon_{i,t}; 0, \sigma^2_i) 
 + \delta^\varepsilon_{i,t} \, \pi_{iMOM}(\varepsilon_{i,t}; 0, 1, 3, \sigma^2 \tau^{\varepsilon}), \\
 p(\delta^\varepsilon_{i,t}) &= \eta^{\delta^\varepsilon_{i,t}} (1 - \eta)^{1 - \delta^\varepsilon_{i,t}},
\end{align}
where $\pi_N(\cdot; 0, \sigma^2_i)$ denotes the Gaussian density with mean zero and variance $\sigma^2_i$, $\delta^\varepsilon_{i,t} \in \{0,1\}$ indicates whether observation $(i,t)$ is classified as an outlier, and $\eta \in (0,1)$ is the prior inclusion probability. Note that we need $\nu=3$ here for the variance to exist. Under this specification, non-outlying observations follow the Gaussian distribution, while outliers are governed by the heavy-tailed iMOM distribution. Posterior inference on $\delta^\varepsilon_{i,t}$ therefore provides a probabilistic classification of outliers.

To mitigate their influence, observations classified as outliers are downweighted in subsequent estimation steps by the inverse standard deviation of the iMOM distribution, $1/\sqrt{2\tau^{\varepsilon}} \ll 1$. Conditional on $\delta^\varepsilon_{i,t}$ and $\tau^{\varepsilon}$, the rescaled residual at each Gibbs iteration is
\begin{align}
 \tilde{\varepsilon}_{i,t} 
 = (1 - \delta^\varepsilon_{i,t}) \hat{\varepsilon}_{i,t} 
 + \delta^\varepsilon_{i,t} \frac{\hat{\varepsilon}_{i,t}}{\sqrt{2\tau^{\varepsilon}}},
\end{align}
where $\hat{\varepsilon}_{i,t}$ denotes the model residual implied by equation~\eqref{eq:mod_simple}. This rescaling attenuates the influence of extreme observations while preserving the information content of the remaining data.

The resulting specification induces a flexible error structure with observation-specific variance inflation driven by posterior outlier probabilities that remains computationally tractable within the indicator-saturated regression framework. Conceptually, this approach is related to Bayesian models with observation-specific volatility, such as stochastic volatility (SV) models \citep[e.g.,][]{jacquier2002, kastner2014}. Modelling the error term with a SV specification could be useful for data where changes in volatility are likely to be more persistent, although the estimation of such specifications requires a rather long time dimension for reliable estimation.
% , while remaining computationally tractable within the indicator-saturated regression framework.

\subsection{Remaining Prior Setup}

%While in general the priors on $\beta$ and $\sigma$ are free to chose, computational convenience suggests the standard setup for linear regression.
The prior distribution on $\beta$, the coefficients associated with the control variables, can be specified using standard formulations such as the independent Normal prior, Zellner’s $g$-prior \citep{zellner1986}, or, as is our default choice, the fractional ($f$-)prior \citep{ohagan1995}. In settings where strong variable selection of covariates is required, continuous shrinkage priors such as the horseshoe prior \citep{carvalho2010} can provide a flexible alternative to these settings, allowing for more aggressive shrinkage of irrelevant coefficients while preserving signals of substantive magnitude.

We specify the prior distribution for $\varepsilon_{i,t}$ to be Gaussian with zero mean and unit-specific variance $\sigma^2_i$. Note that this relates to the distribution of the error term after reweighting in the framework introduced in the previous section or, in its absence, to its form as directly implied by equation~\eqref{eq:mod_simple}.
On $\sigma^2_i$, we employ an inverse gamma prior, 
\begin{equation}
 \sigma^2_i \sim \mathcal{G}^{-1}(m_0, n_0),
\end{equation}
where the hyperparameters $m_0$ and $n_0$ govern the prior scale and dispersion. 
The flexibility of this setup allows for virtually uninformative prior settings (e.g., $m_0 = n_0 = 0.01$), as well as for more informative ones, e.g., by calibrating $m_0$ and $n_0$ in an empirical Bayes fashion based on residuals from a break-free benchmark \citep[see, e.g.][for details]{chipman2010}.  %We opt for the latter as our default choice.
% \footnote{While setting $m_0$ and $n_0$ to values close to zero has been common in the literature, we opt for setting them in a data-driven manner as a default. In particular, we compute OLS residual variance estimates from a benchmark specification that excludes structural breaks. We then derive values for $m_0$ and $n_0$ that imply a prior distribution covering the 90\textsuperscript{th} percentile of these residual variance estimates. This is in line with other approaches using empirical Bayes procedures to inform hyperparameter choice \citep[see, e.g.][]{chipman2010}.} 
% We relax the assumption of homoskedastic errors in Section~\ref{subsec:outlier} to flexibly control for possible outliers.

This completes the prior setup of BISAM. We summarize its various parts together with the default choices and parametrizations for hyperparameters in Section~\ref{app:prior_settings} of the Online Appendix. Posterior inference relies on a Markov Chain Monte Carlo (MCMC) algorithm with several blocks that we outline in Section~\ref{app:posterior_sampler} of the Online Appendix.

\section{Detecting Step-Shifts: A Simulation Study}\label{sec:simulation}

We evaluate the performance of BISAM using a simulation study. We compare its break detection accuracy to two established approaches: the general-to-specific (GETS) testing procedure \citep{castle2015, pretis_discovering_2026} and the adaptive Least Absolute Shrinkage and Selection Operator (ALASSO) regularisation method \citep{tibshirani1996}.

The simulation design reflects a basic setting commonly encountered in applied work on structural break detection. Specifically, we consider a panel with ten cross-sectional units ($N = 10$) and thirty time periods ($T = 30$). Data are generated under the assumption of homoskedastic disturbances with variance $\sigma^2 = 1$. We simulate structural breaks of varying magnitude, with break sizes given by $\gamma \in \{0.5, 1, 1.5, 2, 3, 5, 10\}$. 
We consider two simulation environments that differ in the frequency of structural breaks and correspond to a sparse and a dense break design. In the sparse setting, four out of ten units exhibit exactly one break, occurring at a randomly selected time period. In the dense setting, eight out of ten units experience structural breaks, of which four units exhibit two breaks each.\footnote{Representative simulated datasets for the sparse and dense designs, together with their corresponding estimated PIPs, are presented in Section~\ref{app:add_res_sim} of the Online Appendix.}
For each combination, we generate 100 independent datasets and use each of the competing models to detect step-shifts. This design allows us to assess the relative performance of BISAM and alternative methods across a range of signal strengths and break prevalence.

For the BISAM specification, we set the slab scale parameter to $\tau \approx 3.32$, which implies $P(\lvert\gamma\rvert \leq \sigma \!\mid\! \tau) = 0.01$. This calibration assigns 99\% of the prior mass of the slab component to values outside one standard deviation of the disturbance term, thereby favouring a priori the detection of ``sufficiently large'' breaks. For the GETS procedure, we adopt a nominal significance level of $\alpha = 0.01$,\footnote{Under suitable assumptions, this nominal level can be related to a target false positive rate under the null hypothesis of no structural breaks \citep{pretis_discovering_2026}.} reflecting standard practice in the structural break detection literature for the evaluation of climate policies. For the ALASSO estimator, we follow \citet{zou2006} and use a Ridge regression in the first stage to construct adaptive weights, which are then employed in the penalised estimation step.

We evaluate performance using the following classification metrics:
\begin{itemize}
 \item The true positive rate (\textit{TPR}) and false positive rate (\textit{FPR}); the proportion of correctly and incorrectly detected breaks compared to true breaks, respectively.
 \item \textit{Precision}: the proportion of correctly detected breaks among all detected breaks.
 \item The \textit{F1 score}; the harmonic mean of precision and TPR, which provides a summary measure of detection accuracy.
 \item \textit{Near misses}; the proportion of false positives that are adjacent ($\pm 1$ period) to a true break relative to the total number of false positives.
\end{itemize}

Figure~\ref{fig:sim_results_sparse} and Figure~\ref{fig:sim_results_dense} present the results of the simulation study across break magnitudes for the sparse and dense break environments, respectively. In the sparse design, BISAM and GETS exhibit broadly comparable performance, and both substantially outperform ALASSO across all evaluation metrics. BISAM tends to be more conservative, exhibiting a slightly lower TPR but higher precision than GETS. As a result, both methods achieve similar F1 scores across most break magnitudes.

The results for the near miss metric partly exhibit an inverse-U-shaped pattern, with the proportion of near misses increasing for break magnitudes up to approximately three times the standard deviation of the disturbance term, and declining thereafter. 
At the lower end of break sizes, a large share of the relatively higher number of false positives is detected far off from true breaks due to the low signal that those breaks provide.
For moderately sized breaks, both BISAM and GETS more frequently detect structural breaks in periods adjacent to the true break location. Even in cases where the TPR of BISAM and GETS remains relatively low to moderate (approximately 0.3 and 0.8, respectively), both methods detect signals in the vicinity of the true breakpoints.
At large break sizes, both methods detect true breaks more precisely and reduce the number of false positives in their vicinity, leading to the decline in the metric.

\begin{figure}[ht!]
 \centering
 \includegraphics[width=\linewidth]{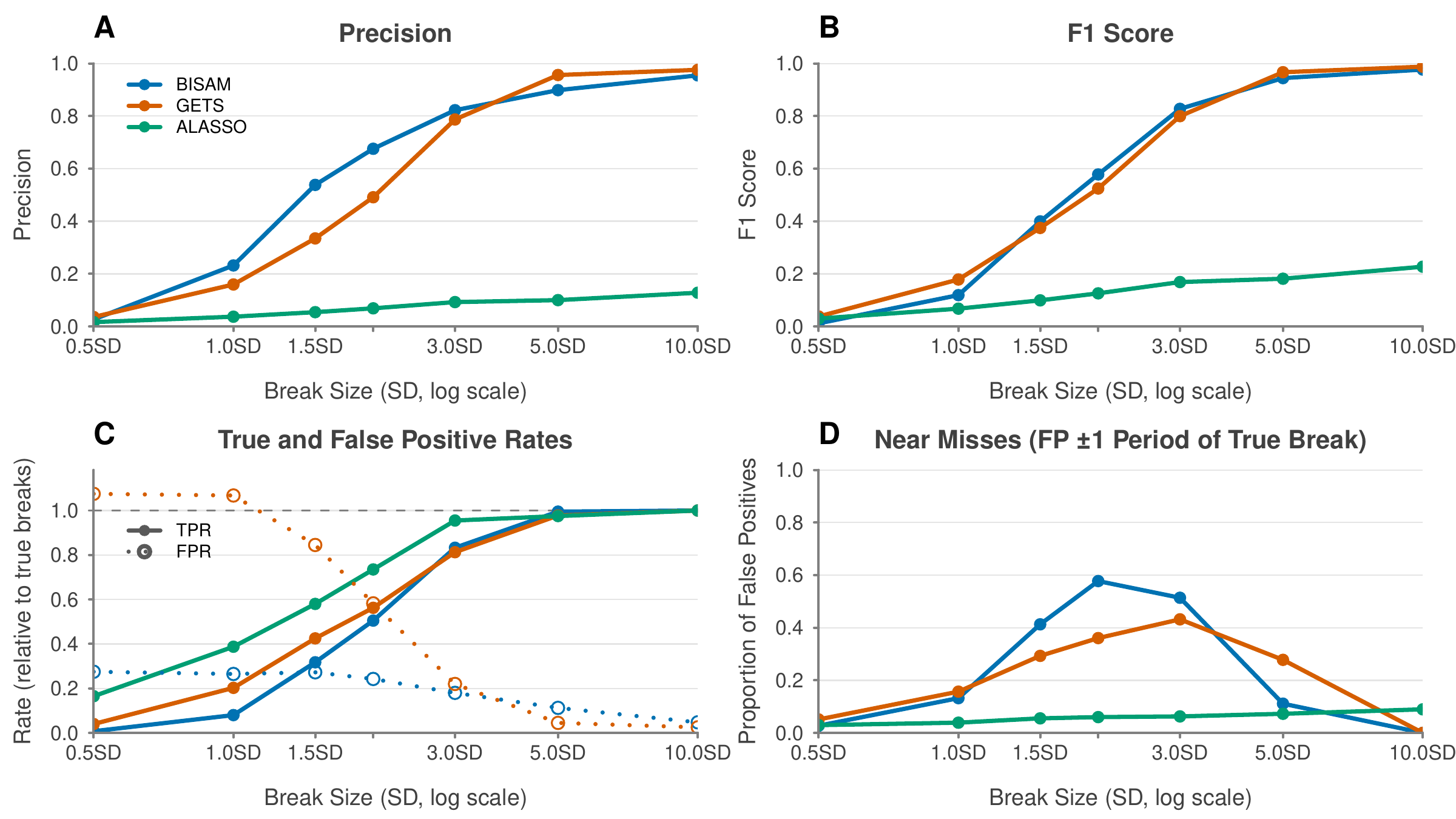}
 \caption{Performance metrics (\textcolor{ssvs}{BISAM}, \textcolor{gets}{GETS}, \textcolor{alasso}{ALASSO}) for sparse simulation setting with varying relative break sizes (measured in standard deviations of the error term). FPR of ALASSO in panel (C) not shown due to being exceedingly high.}
 \label{fig:sim_results_sparse}
\end{figure}

\begin{figure}[ht!]
 \centering
 \includegraphics[width=\linewidth]{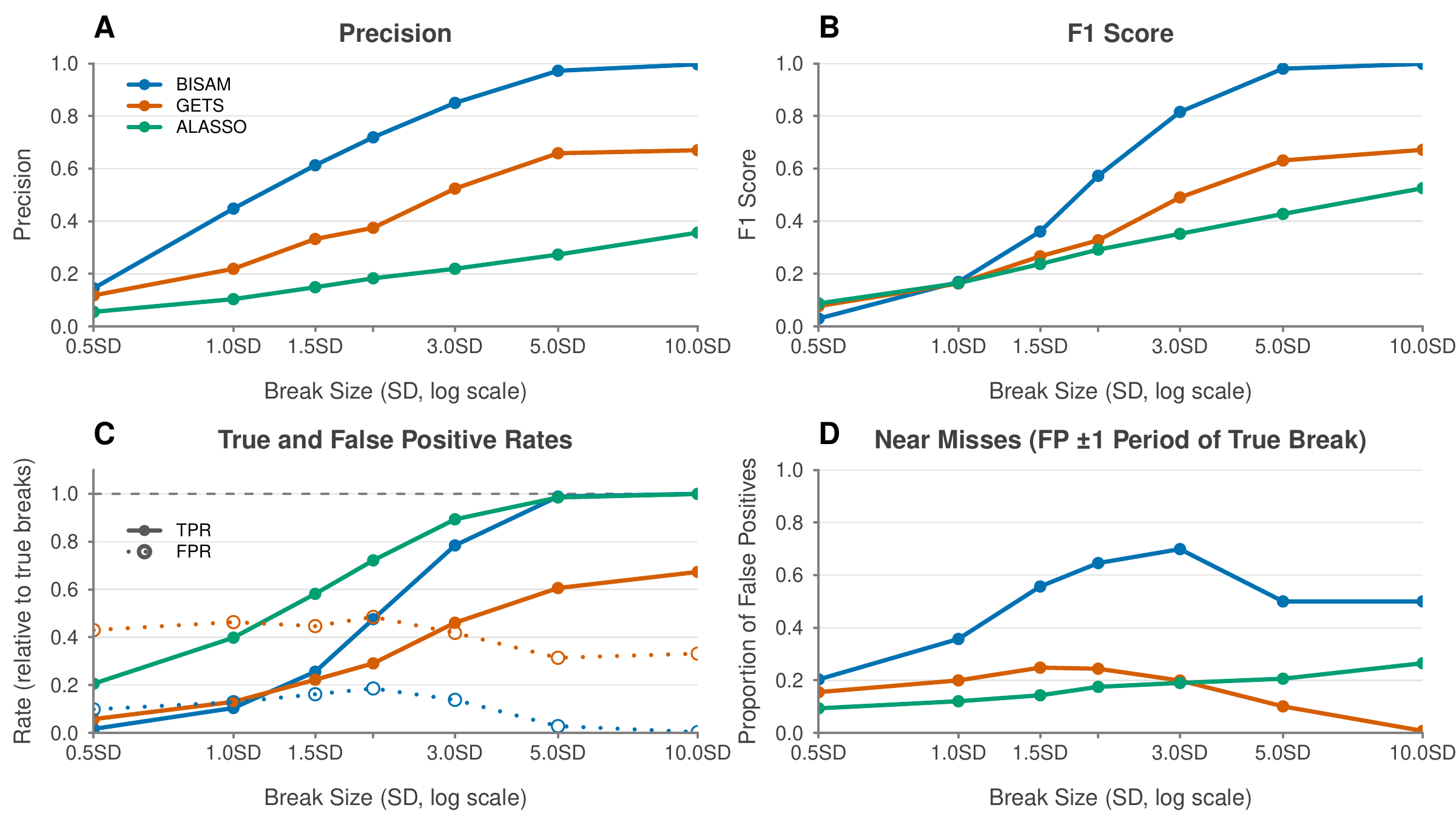}
 \caption{Performance metrics (\textcolor{ssvs}{BISAM}, \textcolor{gets}{GETS}, \textcolor{alasso}{ALASSO}) for dense simulation setting with varying relative break sizes (measured in standard deviations of the error term). FPR of ALASSO in panel (C) not shown due to being exceedingly high.}
 \label{fig:sim_results_dense}
\end{figure}

Differences across methods become more pronounced in environments with a higher frequency of structural breaks. In the dense break setting, GETS exhibits substantially weaker break detection performance. Even for large break magnitudes exceeding three standard deviations of the disturbance term, both precision and F1 scores decline markedly compared to the sparse setting. This deterioration is primarily driven by a reduction in true positive rates but also a less marked reduction in false positive rates. 
In contrast, BISAM maintains stable performance across break magnitudes and consistently outperforms both GETS and ALASSO in detection accuracy. It is also notable that the near miss metric drops markedly for GETS, whereas for BISAM it suggests that more than 50\% of the number of false positives fall in the vicinity of true breaks for most break sizes.
Figure~\ref{fig:sim_results_bn} in the Online Appendix further illustrates these patterns across varying numbers of breaks. While the performance of GETS deteriorates as the number of breaks increases, BISAM preserves its high accuracy for the detection of breaks throughout.\footnote{
    The performance of ALASSO improves as the number of breaks increases. However, this improvement largely reflects the mechanical effect of a shrinking set of non-break candidates, which reduces the scope for false positive classifications, rather than a systematic improvement in its ability to correctly detect true structural breaks.
}

We further extend the simulation study in various dimensions and report performance results of BISAM relative to GETS in Figures~\ref{fig:sim_variations_f1_ratio}~and~\ref{fig:sim_variations_f1_diff} of the Online Appendix.
These variations include different specifications of the panel dimensions, the inclusion of fixed effects and external regressors, different model parameterizations, or the inclusion of the outlier detection framework introduced in Section~\ref{subsec:outlier}. The results mirror those presented in this section, with BISAM being on par with GETS in sparse settings and outperforming strongly in dense settings.

%To summarize, the simulation study shows that the performance of BISAM is on par with that of competing alternative specifications in sparse break environments. In more break-dense environments, BISAM tends to outperform the other methods. The stable performance of BISAM across settings with a variable number of breaks makes the model particularly suited for applications where the number (or proportion) of breaks is unknown a priori.

These results demonstrate that BISAM provides a reliable and robust framework for structural break detection in high-dimensional panel settings. While BISAM and GETS perform similarly in sparse environments, BISAM exhibits substantially greater stability and accuracy as the number of breaks increases. 
% In particular, BISAM maintains consistently high precision and F1 scores across break magnitudes and break frequencies, whereas the performance of GETS deteriorates markedly in settings with higher break density due to declining true positive rates. 
Compared to ALASSO, BISAM delivers superior detection accuracy across all designs, especially for moderate break magnitudes.

\section{Reassessing the Effectiveness of Climate Policies in the European Transport Sector}\label{sec:applications}

Having established the favourable finite-sample properties of BISAM in our controlled simulation setting, we turn to an empirical application to evaluate its performance as a tool for climate policy evaluation. For that purpose, we apply BISAM to study the effectiveness of climate policies in the European road transport sector. We assess the empirical application in \citet{koch2022} and contrast our findings with those obtained using the GETS procedure.\footnote{
    The main identifying assumption imposed by the method is that negative breaks can be attributed to climate policy interventions in the absence of other major confounders. As this rather strong assumption is hard to assess empirically, we abstain from strong causal claims and view this exercise rather as a replication and expansion of the assessment in \citet{koch2022}.
}

\citet{koch2022} seek to detect structural breaks in greenhouse gas emissions from the transport sector across 15 European countries in the period 1995-–2018 and, in a subsequent step, attribute negative breaks to the implementation of transport-related climate policies. 
To this end, the authors employ a step-shift indicator-saturated (SIS) specification analogous to equation~\eqref{eq:mod_simple} with $y_{i,t}$ denoting log CO$_2$ emissions from the transport sector. Besides unit and time fixed effects, log GDP, its square and log population are included as control variables.\footnote{
    The emissions data were obtained from the Emissions Database for Global Atmospheric Research (EDGAR) v5.0 \citep{crippa2019fossil}. Data on GDP and population were sourced from the World Bank Open Data. The dataset includes information for Austria, Belgium, Germany, Denmark, Spain, Finland, France, United Kingdom, Ireland, Italy, Luxembourg, Netherlands, Greece, Portugal, and Sweden.
}

% In addition to the treated units, the empirical framework includes a broader set of countries used as a control group, comprising the remaining EU member states and the European Free Trade Association (EFTA) countries. Specifically, the control group consists of Croatia, Bulgaria, Cyprus, Czech Republic, Estonia, Hungary, Lithuania, Latvia, Malta, Poland, Romania, Slovakia, Slovenia, Switzerland, Iceland, and Norway.

In the absence of strong prior information about the magnitude of potential structural breaks in emissions, we set the slab scale parameter to $\tau \approx 1.92$. This calibration implies that 95\% of the prior mass of the slab component is assigned to values outside one standard deviation of the (country-specific) disturbance term. For the GETS procedure, we adopt a nominal significance level of $\alpha = 0.05$, following \citet{koch2022}. We allow for robust estimation using the procedure outlined in Section~\ref{subsec:outlier} for BISAM and impulse saturation for GETS \citep[see][for details]{pretis_discovering_2026}. We set the prior variance scale for included covariates to $g=10$ and fix the prior inclusion probability parameter $\omega_{i,s}=0.5$.\footnote{
    We provide results from a sensitivity check with a hierarchical treatment of $\omega_{i,s}$ in Section~\ref{app:add_res_repl} of the Online Appendix. Results are largely comparable, though somewhat more conservative in terms of detected breaks.
}

Figure~\ref{fig:eu_transport_emissions} summarises the empirical results. In the upper panels, dashed vertical lines indicate breaks detected by BISAM using a posterior inclusion probability threshold of $P(\gamma_{i,s} \neq 0 \!\mid\! y) > 0.5$ at individual dates, while crosses denote breaks detected by GETS. Shaded regions represent intervals of varying length ($\bar{t} \in {\{1,2,3\}}$) over which BISAM detects structural changes, with orange and green shading corresponding to negative and positive breaks, respectively. Windows with a longer horizon exhibit more transparent shadings.
The lower panels display the observed transport sector emissions over the period 1995–2018, together with the fitted values implied by the BISAM specification. Here, outliers are marked similarly to detected breaks for both approaches.

\begin{figure}[htbp!]
 \centering
 \includegraphics[width=\linewidth, page = 1]{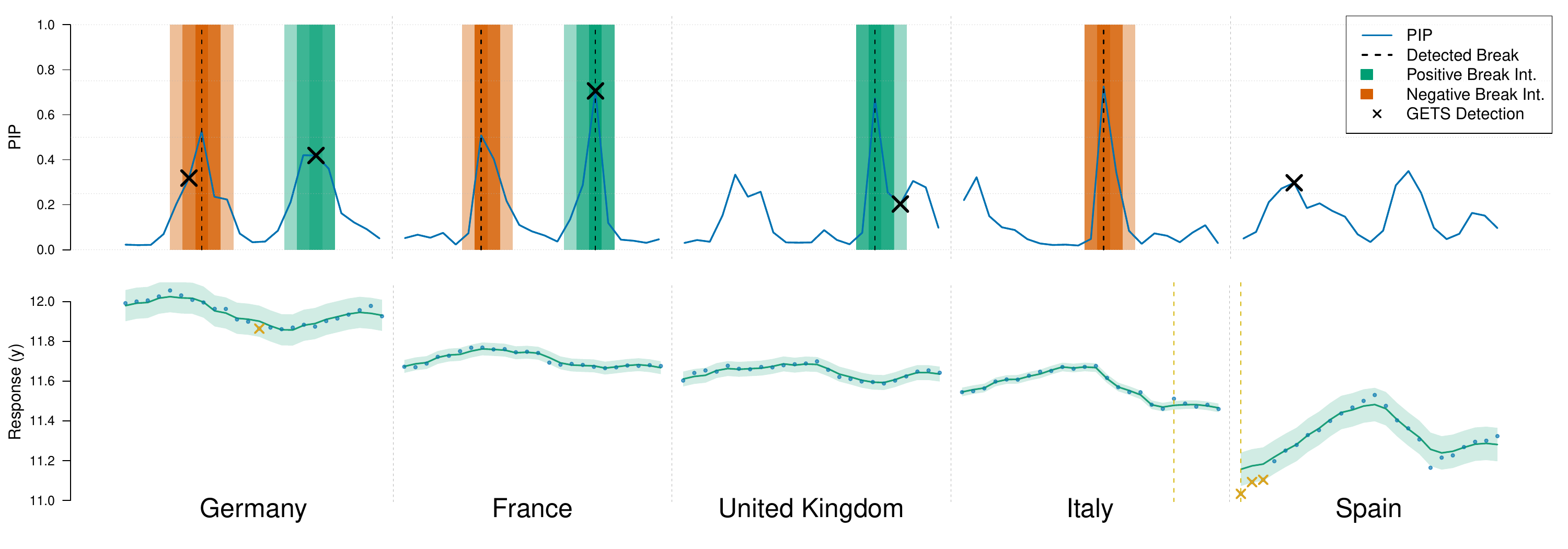}
 \includegraphics[width=\linewidth, page = 2]{img/revision/figure_5_outlier-TRUE-10_beta-f-10_sisprior-imom-1.92072941034706_modprior-bern-0.5.pdf}
 \includegraphics[width=\linewidth, page = 3]{img/revision/figure_5_outlier-TRUE-10_beta-f-10_sisprior-imom-1.92072941034706_modprior-bern-0.5.pdf}
 \caption{Break detection in EU transport emissions. Breaks for individual dates detected are marked by vertical dashed lines (for BISAM) and black crosses (for GETS) in the upper panels, detected outliers are marked similarly in the lower panels. Green (orange) shaded areas indicate periods in which BISAM detects positive (negative) breaks for window lengths $\bar{t} \in {\{1,2,3\}}$. Transparency of the shading indicates the length of windows under which a break is detected, with longer windows being more transparent.}
 \label{fig:eu_transport_emissions}
\end{figure}

Our empirical findings broadly corroborate the main conclusions in \citet{koch2022}, while also revealing systematic differences in break detection and policy attribution. \citet{koch2022} detect ten statistically significant emission reductions across EU countries over the period 1995–2018, with estimated reductions ranging between 8\% and 26\%, and attribute these breaks primarily to policy mixes involving carbon or fuel taxation combined with vehicle-related incentives. Consistent with these findings, BISAM detects a substantial overlap in the timing and location of major structural breaks, especially those associated with large and persistent reductions in transport emissions. This result provides external validation of the structural break detection framework and confirms that the largest emission reductions coincide with major policy interventions affecting the marginal cost of transport-related emissions. 

At the same time, BISAM uncovers additional negative structural breaks that are not detected by the GETS procedure. These breaks point to sustained emission declines in several countries, including France (2002--05), Italy (2007--10), Austria (1997), the Netherlands (2012--16), Sweden (2014--16), and Denmark (1999--2000, 2010--14), where the years in parentheses denote the estimated break intervals. We follow \citet{koch2022} and attribute these breaks to climate policies targeting the road transport sector. To select potentially relevant measures, we use the Policies Database of the International Energy Agency (IEA) and the IFCMA Climate Policy Database of the OECD.\footnote{These databases are available at \url{https://www.iea.org/policies} and \url{https://www.oecd.org/en/data/datasets/ifcma-climate-policy-database.html}. We consider policies at the national and subnational level, conditional on them being recorded in these databases. While we do not explicitly differentiate between national and subnational policies, we note that the contribution of the latter to the overall effects is expected to be of smaller magnitude.}

France (2002--05) implemented the \emph{Plan Climat}, signed in 2004, which introduced measures to improve urban transport efficiency, promote lower-carbon transport modes, and establish biofuel targets that were subsequently enacted into law. Decree 2002-1508 also implemented mandatory labeling schemes for fuel consumption and carbon emissions for passenger cars at the point of sale. Italy (2007--10) increased biofuel quotas in 2008 and introduced purchase incentives for cleaner vehicles between 2007 and 2009, alongside the \emph{Milan Ecopass} congestion charge launched in January 2008.

In Austria (1997), road pricing for motorways and expressways was introduced. Additionally, in the year 1995, which is not included as a potential break but detected as outlier, a tax on mineral oil (MOEST) was updated and pump prices for gasoline and diesel increased substantially.
The Netherlands (2012--16) introduced incentives for zero-emission vehicles and plug-in hybrids, together with a 2015 administrative agreement on zero-emission transport. Over the same period, the legally mandated share of renewable energy in transport fuels increased from about 4.5\% in 2012 to 6.25\% in 2015. 

Sweden (2014--16) did not implement an individual major reform in this time period, but continued increases in biofuel blending and the congestion tax introduced in Gothenburg in 2013 likely contributed with some delay.
In the period 2010--14, Denmark increased the mandatory biofuel share in transport fuels from near zero to approximately 5.75\%, supported by the \emph{Sustainable Transport Plan} adopted in 2008, which included measures to improve transport efficiency, expand public transport, and promote cleaner transport technologies. 
BISAM also finds some support for a negative break in emissions in the period 1999--2000 for Denmark. While no major policy was implemented during this period, in these years the country experienced a gradual tightening and expansion of the carbon tax introduced in the early-1990s, together with promotion of cleaner vehicles and fuels.

\section{Conclusions}\label{sec:conclusion}

In this study, we introduce Bayesian Indicator Saturated Modelling (BISAM), a probabilistic framework for detecting structural breaks in panel data. Our approach builds upon a linear regression specification saturated with step-shift indicators that exhaustively represent the full set of potential breakpoints. To select among these candidate breaks, we employ a spike-and-slab mixture prior combining a Dirac spike and a non-local (inverse-moment density) slab component. This probabilistic formulation enables coherent uncertainty quantification for individual breaks as well as combinations thereof, facilitating the detection of weaker signals that may persist over extended periods. The model also features an integrated treatment of outlier observations.

We demonstrate the accuracy and robustness of BISAM in simulation studies and benchmark its performance against established frequentist alternatives (GETS and adaptive LASSO). Standard performance metrics indicate that BISAM achieves superior detection accuracy while effectively controlling false positives, particularly in environments characterised by multiple structural breaks. Its detection performance remains stable as the number of breaks increases, making it especially well suited to applications in which the number and timing of structural changes are unknown and potentially large.

We illustrate the empirical relevance of BISAM for policy evaluation by revisiting the analysis of emission reductions in the European road transport sector by \citet{koch2022}. While BISAM largely confirms the timing of major emission reductions detected in the existing literature, it also uncovers additional negative structural breaks in countries such as France, Italy, the Netherlands, Sweden, Denmark, and Greece, indicating more persistent and gradual emission declines than previously detected. These additional breaks align closely with the timing of transport-related climate policies, including biofuel mandates, vehicle purchase incentives, and congestion pricing schemes. Taken together, these findings highlight the ability of BISAM to detect sustained structural changes associated with policy interventions and underscore its usefulness as a general and reliable tool for structural break detection in longitudinal settings.

\section*{Acknowledgements}

The authors would like to thank Felix Pretis and Moritz Schwarz for helpful methodological discussions, Jakob Goldmann for assistance in improving the \texttt{bisam} package and David Rossell for support on tailoring the \texttt{mombf} package to our needs. Support by the Vienna Science and Technology Fund under grant ESS22-040 and the WU Wien Stiftung is gratefully acknowledged. 

\section*{Code \& Data Availability}

All code and data to estimate Bayesian Indicator Saturated Models and replicate the findings presented in this paper are available at
\url{https://github.com/konradld/EctJ_climate_policy_breaks}.

\bibliography{refs}

%----------------------
\newpage

\renewcommand{\theequation}{A.\arabic{equation}}
\renewcommand{\thesection}{A}
\setcounter{equation}{0}
\renewcommand{\thefigure}{A.\arabic{figure}}
\setcounter{figure}{0}
\renewcommand{\thetable}{A.\arabic{table}}
\setcounter{table}{0}
\makeatletter
\edef\@currentlabel{\thesection}
\makeatother

\section*{Appendix A: Derivation of window selection rule}\label{app:window_rule}

% \subsection*{}
 
Consider a time window $W$ of length $\bar{t}$ for unit $i$ and let $A^{W}_{i}$ denote the event that at least one break occurs within it, as defined in equation~\eqref{eq:window_pip}. Write $\bar{\pi} = \mathrm{PIP}^{W}_{i} = P(A^{W}_{i} \!\mid\! y)$. The decision is binary. Either the window is reported as containing a break, $a = 1$, or it is not, $a = 0$. We specify the loss
\begin{equation}
 L(a, A^{W}_{i}) =
 \begin{cases}
 c_1 \,\mathbbm{1}_{\{A^{W}_{i} \text{ false}\}} + \kappa (\bar{t}-1) & \text{if } a = 1, \\[4pt]
 c_0 \,\mathbbm{1}_{\{A^{W}_{i} \text{ true}\}} & \text{if } a = 0,
 \end{cases}
 \label{eq:window_loss}
\end{equation}
with $c_0, c_1 > 0$ and $\kappa \geq 0$. The term $c_1$ prices a spurious window, $c_0$ prices a missed break, and $\kappa(\bar{t}-1)$ prices the imprecision of the reported timing, which is incurred whenever a window is reported and is linear in its width in excess of a point detection. Setting $\kappa = 0$ recovers the standard classification loss.

Taking expectations with respect to the posterior distribution of $A^{W}_{i}$ yields the posterior expected losses
\begin{align}
 \mathbb{E}\!\left[L(1, A^{W}_{i}) \!\mid\! y\right] &= c_1 (1 - \bar{\pi}) + \kappa(\bar{t}-1), \\
 \mathbb{E}\!\left[L(0, A^{W}_{i}) \!\mid\! y\right] &= c_0 \, \bar{\pi}.
\end{align}
The Bayes action is $a = 1$ whenever the former is smaller than the latter, that is, whenever
$c_1 (1 - \bar{\pi}) + \kappa(\bar{t}-1) < c_0 \bar{\pi}$. Solving for $\bar{\pi}$ gives the threshold rule
\begin{equation}
 \mathrm{PIP}^{W}_{i} \;>\; P^{\bar{t}} \;=\; \frac{c_1 + \kappa (\bar{t} - 1)}{c_0 + c_1}.
 \label{eq:window_rule_app}
\end{equation}
Three remarks are in order. First, at $\bar{t} = 1$ the window collapses to a single candidate date, $A^{W}_{i} = \{\delta^{\gamma}_{i,s} = 1\}$ and $\mathrm{PIP}^{W}_{i}$ is the marginal PIP of that indicator. The imprecision term vanishes and \eqref{eq:window_rule_app} reduces to $P^{1} = c_1 / (c_0 + c_1)$, which equals $0.5$ under symmetric classification losses. The window rule therefore nests the median probability model of \citet{barbieri2004optimal} as a special case. The optimality of the median probability model is established for the selection of individual regressors under squared-error prediction loss, and extended to correlated designs by \citet{BarbieriBerger2021}. It does not extend to disjunctive events of the form $A^{W}_{i}$, which is why the window threshold is derived from an explicit classification loss rather than obtained as a corollary. Second, $P^{\bar{t}}$ is strictly increasing in $\bar{t}$ for $\kappa > 0$, offsetting the mechanical increase of $\mathrm{PIP}^{W}_{i}$ in the window length. Third, $P^{\bar{t}} \geq 1$ whenever $\bar{t} \geq 1 + c_0 / \kappa$, since $\mathrm{PIP}^{W}_{i} \leq 1$ by construction, windows of that width or wider are never reported. This way the loss in \eqref{eq:window_loss} induces a maximal reportable window length set as
\begin{equation}
 \bar{t}_{\max} \;=\; \left\lceil 1 + c_0/\kappa \right\rceil - 1 .
 \label{eq:tmax}
\end{equation}

The cap on window width is thus implied by the loss rather than imposed separately, and is governed by the single interpretable ratio $\kappa / c_0$.

It is worth noting that the rule in \eqref{eq:window_rule_app} is invariant to the specification of the prior on the model space, as $\mathrm{PIP}^{W}_{i}$ is computed directly from the posterior draws of $\delta^{\gamma}_{i,\cdot}$. %A rule based instead on the comparison of $\mathrm{PIP}^{W}_{i}$ with its prior counterpart $P(A^{W}_{i})$ would amount to selecting windows whenever the Bayes factor in favour of $A^{W}_{i}$ exceeds one, and would require $P(A^{W}_{i}) = 1 - (1-\omega)^{\bar{t}}$. This holds only under independent Bernoulli inclusion priors with fixed $\omega$ and not, for instance, under the Beta--Bernoulli specification discussed in Section~\ref{subsec:bisam1}.

%------------------
\newpage

\renewcommand{\theequation}{B.\arabic{equation}}
\renewcommand{\thesection}{B}
\setcounter{equation}{0}
\renewcommand{\thetheorem}{\arabic{theorem}}
\makeatletter
\edef\@currentlabel{\thesection}
\makeatother

\section*{Appendix B: Consistency of the window rule}\label{app:theory}

This section provides an overview of the assumptions and the main theorem under which window-level model selection consistency is achieved in the ISR framework.
Proofs and supporting lemmas are in Section~\ref{app:gauge} of the Online Appendix, where Assumption~\ref{as:regularity} below appears as Assumptions~\ref{as:nlp}, \ref{as:design}, \ref{as:prior2} and~\ref{as:space}. Conditioning on $(\beta, \sigma^2)$ and on the outlier indicators reduces model selection to a single unit, $z := y_i - X_i\beta = \mu^0 + \varepsilon$ with $\mu^0 = \sum_{j \leq k^0} \gamma^0_j u_{s^0_j}$ and $\varepsilon \sim \mathcal{N}(0, \sigma^2 I_T)$. Write $\gamma_T = \min_j |\gamma^0_j|$, $\bar{\gamma} = \max_j |\gamma^0_j|$, $K = N(T-3)$, and $\Delta_T$ for the minimum spacing between consecutive true dates inclusive of the sample boundaries.

\begin{assumption}[Regularity]\label{as:regularity}
\leavevmode\par
\begin{enumerate}[label=(\alph*)]
    \item\emph{Signal and design.} $0 < \underline{\gamma} \leq \gamma_T \leq \bar{\gamma} < \infty$ and $\Delta_T \big/ (\log NT)^{(k+1)/k} \to \infty$, and the reporting cap of \eqref{eq:tmax} satisfies $\bar{t}_{\max} \to \infty$ with $\bar{t}_{\max} = o(\Delta_T)$.
    \item\emph{Model prior.} The prior over configurations $m$ depends on $m$ only through $|m|$, and its inclusion odds satisfy $\varrho'_T \leq \varrho \leq \bar{\varrho} < \infty$ uniformly, with $\log(1/\varrho'_T) = O(\log K)$.
    \item\emph{Model space.} Candidate dates lying within a common window of width $\ell_T$ are collapsed to a single representative, with $\ell_T \to \infty$, $\ell_T = o(\Delta_T)$ and $\ell_T^{k/(k+1)} / \log K \to \infty$.
    \item\emph{Non-local factorisation.} In the decomposition $M(z \mid m) = M^{L}(z \mid m) G_m(z)$ of \citet[Proposition~1]{rossell2017nonlocal}, the ratio $G_{m'}/G_m$ is bounded above and below by positive constants for one-coordinate additions and for comparisons in which no coordinate changes relevance status.
\end{enumerate}
\end{assumption}

\begin{theorem}[Consistency of the window rule]\label{thm:window_consistency}
Let Assumption~\ref{as:regularity} hold and let $\widehat{\mathcal{W}}$ be the set of windows reported by equation~\eqref{eq:window_rule}. 

Then, as $T \to \infty$,

\begin{enumerate}[label = (\roman*)]
\item for each true date $s^0_j$ there is a window $W_j \ni s^0_j$ of width at most $\bar{t}_{\max}$ with $\mathrm{PIP}^{W_j}_{i} \overset{p}{\to} 1$ and threshold $P^{|W_j|}$ bounded away from one, so that $W_j$ is reported with probability tending to one, simultaneously over $j = 1, \dots, k^0$, and 

\item every window of width at most $\bar{t}_{\max}$ lying further than $\bar{t}_{\max}$ from every true date has $\mathrm{PIP}^{W}_{i} \overset{p}{\to} 0$ and is not reported. Consequently $\widehat{\mathcal{W}}$ has exactly $k^0$ connected components, one containing each true break and each localised to within $2\bar{t}_{\max}$ of the corresponding true date, with probability tending to one.
\end{enumerate}
\end{theorem}

Part~(i) is the potency statement of Theorem~\ref{thm:potency} of the Online Appendix, resting on Propositions~\ref{prop:coverage} and~\ref{prop:loc}, part~(ii) is the gauge bound of Theorem~\ref{thm:gauge}, and the localisation claim follows from Corollary~\ref{cor:hausdorff}. %Assumption~\ref{as:regularity}(a) requires a growing cap, so windows reported under the default calibration $(1,1,1/3)$ do not inherit the theorem, and 
%Remark~\ref{rem:nearmiss} quantifies what survives at a fixed cap. 
Proposition~\ref{prop:resolution} shows why the statement concerns windows rather than dates.

%======================================================================
%                        O N L I N E   A P P E N D I X
%======================================================================

\clearpage

% ---- numbering conventions of the online appendix ----
\renewcommand{\thesection}{S\arabic{section}}
\renewcommand{\theequation}{S.\arabic{equation}}
\renewcommand{\thefigure}{S.\arabic{figure}}
\renewcommand{\thetable}{S.\arabic{table}}
\renewcommand{\thetheorem}{\arabic{section}.\arabic{theorem}}
\renewcommand{\theassumption}{\arabic{section}.\arabic{assumption}}
\renewcommand{\theproposition}{\arabic{section}.\arabic{proposition}}
\renewcommand{\thecorollary}{\arabic{section}.\arabic{corollary}}
\renewcommand{\thelemma}{\arabic{section}.\arabic{lemma}}
\renewcommand{\theexample}{\arabic{section}.\arabic{example}}
\renewcommand{\theremark}{\arabic{section}.\arabic{remark}}

% ---- restart counters (main text and printed appendices A/B ran them up) ----
\setcounter{section}{0}
\setcounter{equation}{0}
\setcounter{figure}{0}
\setcounter{table}{0}
\setcounter{theorem}{0}
\setcounter{assumption}{0}
\setcounter{proposition}{0}
\setcounter{corollary}{0}
\setcounter{lemma}{0}
\setcounter{example}{0}
\setcounter{remark}{0}
\setcounter{definition}{0}
\setcounter{footnote}{0}

% ---- uncomment the next two lines to restart page numbers as S1, S2, ... ----
% \renewcommand{\thepage}{S\arabic{page}}
% \setcounter{page}{1}

% ---- title block of the online appendix ----
\begin{center}
 {\large\bfseries Online Appendix\par}
 \vspace{0.4em}
 {\large\bfseries Bayesian Indicator-Saturated Regression\par}
 \vspace{0.8em}
 {\normalsize Lucas D. Konrad, Lukas Vashold and Jesus Crespo Cuaresma\par}
\end{center}

\vspace{1em}

\noindent This Online Appendix contains a discussion of model selection consistency in the indicator-saturated regression framework, associated assumptions and proofs, additional information about and results from simulations, and additional results for the replication of the study by \citet{koch2022}.

\vspace{1.5em}

\newcommand{\CN}[1]{\marginpar{\tiny\textsf{#1}}}   % \renewcommand{\CN}[1]{} to hide
\newcommand{\PP}{\mathbb P}
\newcommand{\Post}[1]{\Pi(#1\mid z)}

\section{Non-local priors and multiplicity in saturated designs}\label{app:nlp_rates}

The indicator saturation framework implies an effective model dimension that grows as $O(NT).$ \citet{rossell2012nlp_modsel} show that when the number of candidate coefficients increases linearly in the number of observations, non-local priors assign posterior probability that tends to unity to the true model, whereas any continuous local prior with $\pi_L(0) > \varepsilon > 0$ assigns the true model posterior probability that tends to zero once the number of candidates exceeds $O(T^{1/2})$ \citep[see Theorem~2 in][]{rossell2012nlp_modsel}. This result  requires the candidate design to be well conditioned, with the smallest eigenvalue of $X_n'X_n$ bounded below by $cn$. This condition fails for a step-shift--saturated design: adjacent indicators $\mathbbm{1}_{\{t\geq s\}}$ and $\mathbbm{1}_{\{t\geq s+1\}}$ differ in a single period, so their sample correlation tends to one, the smallest eigenvalue of $X_n'X_n$ stays $O(1)$ while $n$ grows, and the true break \emph{dates} are only identified up to a local neighbourhood. Exact-model consistency is therefore not available in the ISR setting. What remains identified is the \emph{window-level} configuration of breaks, which is the target of the window rule of Section~\ref{subsec:bisam1}. For the iMOM slab, the posterior detects each true break up to a time window that is negligibly small relative to the spacing between breaks (of width $O_p(\log(KT))$ rather than $O_p(1)$ by Proposition~\ref{prop:resolution}), and excludes every candidate date well separated from all true breaks, so that the windowed inclusion probabilities are consistent for the set of true break windows. We state the assumptions and prove this window-level consistency in Theorem~\ref{thm:window_consistency_oa}.

Under a pMOM slab of order $k$ the spurious-break Bayes factor decays polynomially, $\mathrm{BF}_T(H_A\!\mid\!H_0) = O_p(T^{-k-1/2})$, whereas under the iMOM (piMOM) slab the decay is near-exponential,

\begin{equation}
 \frac{\log \mathrm{BF}_T(H_A\!\mid\!H_0)}{T^{k/(k+1)}}
 \overset{p}{\to} c < 0,
 \label{eq:pimom_rate_app}
\end{equation}
so that evidence against a negligible break accumulates faster than under any local slab, whose rate stays $O_p(T^{-1/2})$ irrespective of its calibration \citep[eq.~11]{rossell2010nlp_test}. When a break is genuinely present, all of these slabs accumulate evidence at the same exponential rate in $T$, so the additional stringency of the non-local slab is costless on true breaks. Figure~\ref{fig:prior_rates} confirms both statements directly on the step-shift design.

To make this mechanism concrete on the design we use for detection, we evaluate $\mathrm{BF}_T(H_A\!\mid\!H_0)$ for a single step-shift unit as $T$ grows, isolating the slab from the sampler, the outlier component and the multiplicity correction. Figure~\ref{fig:prior_rates} reports a local (Zellner/normal) slab, a pMOM slab and the iMOM slab, all calibrated by the common rule $P(|\gamma|\leq\sigma\mid\tau)=0.01$ of Table~\ref{tab:hyperpars}. Panel~(a) recovers the three predicted null rates. Panel~(b) shows that, when a break is truly present, the results for the slabs are indistinguishable, so the iMOM buys its extra stringency at no cost in true-break evidence.

\begin{figure}[t]
 \centering
 \includegraphics[width=\linewidth]{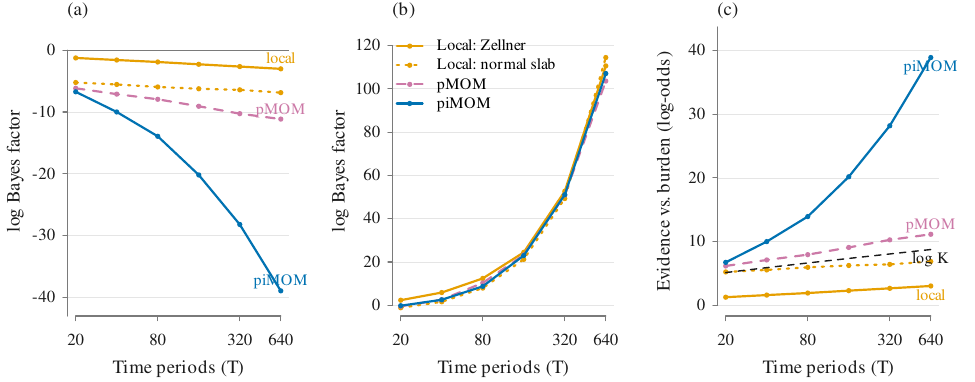}
 \caption{Evidence-accumulation rates by slab prior on the step-shift design. Each panel plots a log Bayes factor against the number of periods $T\in\{20,40,80,160,320,640\}$ for a single SIS unit, averaged over $\sim\!200$ replications, with $\mathrm{BF}_T$ evaluated by Laplace approximation via \texttt{mombf::nlpMarginal}. All slabs share the calibration $P(|\gamma|\leq\sigma\mid\tau)=0.01$. \textbf{(a)} Evidence \emph{against} a spurious break ($H_0$ true): the local slab decays at $-\tfrac12\log T$, the pMOM slab at $-\tfrac32\log T$, and the iMOM slab near-exponentially ($\propto-\sqrt T$), consistent with \eqref{eq:pimom_rate_app}. Fitted rates: local $-0.51\log T$ ($R^2=0.999$), pMOM $-1.46\log T$ ($R^2=0.997$), iMOM $-1.55\sqrt T$ ($R^2=0.999$). \textbf{(b)} Evidence \emph{for} a true $3\sigma$ break: all slabs are linear in $T$ with essentially identical slopes. The two local arms (Zellner $g=n$ and a variance-matched normal slab, $\tau\approx6370$) differ only in level, but both yield the same $-\tfrac12\log T$ null rate in (a), so that rate is a property of locality, not of calibration. \textbf{(c)} The saturation trade-off. The multiplicity burden $\log K$ with $K=N(T-3)$ grows like $\log T$. The local rate $\tfrac12\log T$ never overtakes it, whereas both non-local rates dominate it for large $T$.}
 \label{fig:prior_rates}
\end{figure}

Panel~(c) illustrates why this rate gap is decisive for indicator saturation rather than a generic refinement. In the SIS design, the candidate set grows with the sample: there are $T-3$ candidate dates per unit and $K=N(T-3)=O(NT)$ across the panel. Keeping the expected number of {spuriously} selected dates under control requires the per-candidate evidence against a null date to overcome a multiplicity burden of order $\log K = \log T + O(1)$. A local slab supplies only $\tfrac12\log T$ and is dominated by this burden. With the non-local factor $G_m$ bounded away from zero, the posterior expected number of false inclusions is $O_p(T\cdot T^{-1/2}) = O_p(T^{1/2})\to\infty$, so a local slab cannot control false discoveries in a saturated design, although it is tuned.

A pMOM slab of order $k$ supplies $(k+\tfrac12)\log T$, giving $O_p(T^{1/2-k})\to0$ for $k\geq1$, and the iMOM slab drives the per-candidate evidence down near-exponentially, dominating $\log K$ with a widening margin. This per-candidate accounting is optimistic for the pMOM. The bound must hold uniformly over the $O(NT^3)$ conditional designs, and Corollary~\ref{cor:panel} shows that the price of that uniformity pushes the pMOM requirement to $k\geq3$ even at $N=1$, while leaving the iMOM conclusion intact. This is exactly the arithmetic behind Theorem~\ref{thm:gauge}.

Non-locality is thus not a refinement of the slab in this setting but a requirement for false-discovery control once the number of candidate breaks grows with $T$, and it is what underlies the specificity advantage of BISAM over GETS and ALASSO documented in Section~\ref{sec:simulation}. This argument concerns candidate dates which are well separated from every true break, where collinearity with the truth is asymptotically negligible. The complementary problem of {localising} a detected break among its near-collinear neighbours is handled separately by the window rule and Theorem~\ref{thm:window_consistency_oa}.

\section{Gauge, potency, resolution and posterior convergence}
\label{app:gauge}

\subsection{Setup}

The break design is block diagonal across units and the model-selection block of the sampler conditions on $(\beta,\sigma^2)$. We work conditionally on $(\beta,\sigma^2)$ and analyse a single unit, writing $z:=y_i-X_i\beta$. Statements are uniform over $\sigma^2$ in a fixed compact subset of $(0,\infty)$, and counting statements aggregate over units at the end. Thus
\begin{equation}\label{eq:dgp}
 z \;=\; \mu^0+\varepsilon,\qquad
 \mu^0:=\sum_{j=1}^{k^0}\gamma^0_j u_{s^0_j},\qquad
 \varepsilon\sim\mathcal N(0,\sigma^2 I_T),
\end{equation}
with $u_s=(\mathbbm 1\{t\ge s\})_{t=1}^T$, candidate set $\mathcal S=\{3,\dots,T-1\}$, true dates $\mathcal T^0=\{s^0_1<\dots<s^0_{k^0}\}$, magnitudes $\gamma_T:=\min_j|\gamma^0_j|\ge\underline\gamma>0$ and $\bar\gamma:=\max_j|\gamma^0_j|<\infty$, and $K:=N(T-3)$ candidate coefficients in the panel. Throughout, $\PP$ denotes probability under \eqref{eq:dgp} and $\Post{\cdot}$ the posterior, while $c,C,C'$ are positive constants that may change between lines.

Three points of contact with the main text should be recorded. First, $z$ is the partialled-out outcome of a single unit, the object written $\tilde y_{i,\cdot}$ in the description of the sampler in Section~\ref{app:posterior_sampler}, and not the step-shift design matrix $ Z$. Here $u_s$ is the column of $ Z_i$ associated with candidate date $s$, and $K=N(T-3)$ is the number of candidate break coefficients, written $q$ in Section~\ref{app:posterior_sampler}. Second, when the outlier component is active the working error variance is $\sigma^2_i\lambda_{i,t}$ rather than $\sigma^2_i$. All statements below are made conditionally on $\delta^{\varepsilon}$, so that $\sigma^2$ is read as the working variance of the GLS-standardised block and \eqref{eq:dgp} is the model actually seen by the model-selection step. The symbol $\lambda_s$ below is the local information $\tau n_s$ of \eqref{eq:condquant} and is unrelated to the variance-inflation factor $\lambda_{i,t}$ or to the latent truncation variable $\lambda^\gamma_s$ of the sampler. Third, $\eta\in(0,1)$ denotes throughout the tolerance of the high-probability events, not the outlier probability of the main text.

The minimum spacing between true dates is taken \emph{boundary inclusive}:
\begin{equation}\label{eq:spacing}
 s^0_0:=1,\qquad s^0_{k^0+1}:=T+1,\qquad
 \Delta_T\;:=\;\min_{0\le j\le k^0}\bigl(s^0_{j+1}-s^0_j\bigr).
\end{equation}
This is not cosmetic: Lemma~\ref{lem:energy} is false if true breaks are permitted within $R$ of the sample ends.\footnote{With $T=100$, a single break at $s^0=4$, $m=\{60\}$ and $R=12$ one has $\lVert(I-H_m)\mu^0\rVert^2=2.85$ against $\tfrac12\gamma_T^2R=6$.} The growth conditions that $\Delta_T$ must satisfy are collected in Assumption~\ref{as:design} and are stronger than a separation condition of the form $\Delta_T/\log T\to\infty$.

A configuration is a set $m\subseteq\mathcal S$, with $m^{-s}:=m\setminus\{s\}$. We write $V_m$ for the span of $\{u_s:s\in m\}$ and the constant vector, and $H_m$ for the projection onto it. For $s\notin m^{-s}$ the \emph{flanking dates} are
\begin{equation}\label{eq:flank}
 a_s:=\max\bigl(\{s'\in m^{-s}:s'<s\}\cup\{1\}\bigr),\qquad
 b_s:=\min\bigl(\{s'\in m^{-s}:s'>s\}\cup\{T+1\}\bigr),
\end{equation}
and we set
\begin{equation}\label{eq:condquant}
 \tilde u_s:=(I-H_{m^{-s}})u_s,\quad n_s:=\lVert\tilde u_s\rVert^2,\quad
 \hat\gamma_s:=\frac{\langle\tilde u_s,z\rangle}{n_s},\quad
 t_s:=\frac{\hat\gamma_s\sqrt{n_s}}{\sigma},\quad \lambda_s:=\tau n_s .
\end{equation}
Write $D_j:=\min_{s\in m}|s-s^0_j|$ and $U_m(R):=\#\{j\le k^0:D_j>R\}$. For $s\in m'$,
\begin{equation}\label{eq:theta}
 \vartheta_s(m'):=\frac{\bigl|\langle(I-H_{m'\setminus\{s\}})u_s,\ \mu^0\rangle\bigr|}
 {\bigl\lVert(I-H_{m'\setminus\{s\}})u_s\bigr\rVert}
\end{equation}
denotes the \emph{partial signal} of $s$ in $m'$: deleting $s$ from $m'$ increases $\lVert(I-H)\mu^0\rVert^2$ by exactly $\vartheta_s(m')^2$. Two radii and one threshold are used throughout,
\begin{equation}\label{eq:twoscale}
 R_T\ \text{as in \eqref{eq:RT}},\qquad
 r_T:=\Bigl\lceil\frac{\gamma_T^2R_T}{16\,\bar\gamma^2}\Bigr\rceil,\qquad
 \theta^2:=\tfrac14\gamma_T^2\min\Bigl\{r_T,\ \frac{R_T}{8k^0}\Bigr\},
\end{equation}
and $s$ is \emph{$\vartheta$-relevant} in $m'$ if $\vartheta_s(m')>\theta$. The \emph{proxy radius} $r_T$ is the distance within which an included date explains a true break well enough that the break need not be added. It is proportional to $R_T$ precisely because $\bar\gamma/\gamma_T$ is bounded (Assumption~\ref{as:design}), and this is the one place where a \emph{ceiling} on the break magnitudes, and not merely the floor $\underline\gamma$, does work. This is a finite-sample version of the dichotomy that drives \citet[Proposition~1]{rossell2017nonlocal}: the non-local factor $G_{m'}$ of the factorisation $M(z\mid m')=M^{L}(z\mid m')G_{m'}(z)$ converges to a strictly positive limit when every included column is relevant and identified, and to zero when $m'$ contains a column asymptotically uncorrelated with $\mu^0$ given the others. Every use of that factorisation below is signposted. The reported object is $\mathrm{PIP}^W=\Post{A^W}$ with $A^W=\{\exists s\in W:\delta^\gamma_s=1\}$, compared with $P^{\bar t}=\{c_1+\kappa(\bar t-1)\}/(c_0+c_1)$ as in \eqref{eq:window_rule_app}, and the cap is $\bar t_{\max}=\lceil1+c_0/\kappa\rceil-1$ as in \eqref{eq:tmax}. A candidate is \emph{$h$-far} if $\operatorname{dist}(s,\mathcal T^0)>h$, and spurious-date statements are restricted to $\mathcal S_T:=\{s\in\mathcal S:\min(s-1,T+1-s)>h\}$. The consequence for the window rule itself, which is the statement cited in the main text, is Theorem~\ref{thm:window_consistency_oa}, which collects Theorem~\ref{thm:gauge} and Theorem~\ref{thm:potency} and is proved last.

\subsection{Geometry of the step design}

\begin{proposition}[Conditional geometry]\label{prop:geom}
Fix $m^{-s}$ and let $a=a_s$, $b=b_s$. Then
\begin{enumerate}
\item[\textup{(i)}] $\tilde u_s=\mathbbm 1_{[s,b)}-\frac{b-s}{b-a}\mathbbm 1_{[a,b)}$ and $n_s=\frac{(s-a)(b-s)}{b-a}$. In particular $\tilde u_s$ depends on $m^{-s}$ only through $(a,b)$, and $\tfrac12\min(s-a,b-s)\le n_s\le\min(s-a,b-s)$.
\item[\textup{(ii)}] $\langle\tilde u_s,u_{s'}\rangle=0$ for $s'\notin(a,b)$, and
$\langle\tilde u_s,u_{s'}\rangle=\frac{(s\wedge s'-a)(b-s\vee s')}{b-a}$ for $s'\in(a,b)$.
\item[\textup{(iii)}] $\mathbb E[\hat\gamma_s]
=n_s^{-1}\sum_{j:s^0_j\in(a,b)}\gamma^0_j\frac{(s\wedge s^0_j-a)(b-s\vee s^0_j)}{b-a}$ and $t_s-\mathbb E[t_s]\sim\mathcal N(0,1)$.
\end{enumerate}
\end{proposition}

\begin{proof}
(i) $V_{m^{-s}}$ is the space of vectors constant on each cell of the partition induced by $m^{-s}$, and the cell containing $s$ is $[a,b)$. Write $u_s=\mathbbm 1_{[b,T]}+\mathbbm 1_{[s,b)}$. The first summand is constant on every cell, hence lies in $V_{m^{-s}}$, and the second is supported inside the single cell $[a,b)$, where the projection subtracts its cell mean $(b-s)/(b-a)$. This gives the stated form, which involves $m^{-s}$ only through $(a,b)$, and
\[
 n_s=(b-s)\Bigl(1-\tfrac{b-s}{b-a}\Bigr)^2+(s-a)\Bigl(\tfrac{b-s}{b-a}\Bigr)^2
 =\frac{(b-s)(s-a)^2+(s-a)(b-s)^2}{(b-a)^2}=\frac{(s-a)(b-s)}{b-a},
\]
the harmonic mean of $s-a$ and $b-s$ up to a factor two.

(ii) $\tilde u_s$ is supported on $[a,b)$ and sums to zero there, while $u_{s'}$ is constant on $[a,b)$ when $s'\notin(a,b)$, giving zero. For $s'\in(a,b)$ put $\tilde u_{s'}:=(I-H_{m^{-s}})u_{s'}$, so that $\langle\tilde u_s,u_{s'}\rangle=\langle\tilde u_s,\tilde u_{s'}\rangle$. Taking $s\le s'$ and expanding the two expressions from (i), $\langle\tilde u_s,\tilde u_{s'}\rangle=(b-s')-\frac{(b-s)(b-s')}{b-a} =\frac{(b-s')(s-a)}{b-a}$.

(iii) Substitute \eqref{eq:dgp} and apply (ii). The noise term is $\mathcal N(0,\sigma^2/n_s)$ since $\tilde u_s$ is non-random given $(a,b)$.
\end{proof}

Three consequences are used throughout. The direction $\tilde u_s/\lVert\tilde u_s\rVert$ is indexed by the triple $(a,s,b)$ alone, so suprema of the noise over configurations are suprema over at most $T^3$ fixed directions rather than over $2^{|\mathcal S|}$ models. A candidate whose flanking dates are close carries $O(1)$ information: adjacent candidates are aliased. And $\mathbb E[\hat\gamma_s]=0$ whenever the cell $[a,b)$ contains no true break.

\begin{lemma}[Mismatch bound]\label{lem:mismatch}
Suppose $\Delta_T>2h$ and every true break lies within $h$ of a flanking date of its own cell. Then every $h$-far candidate satisfies $|\mathbb E[t_s]|\le2\bar\gamma\sqrt h/\sigma$.
\end{lemma}

\begin{proof}
Let $[a,b)$ be the cell of an $h$-far $s$. By hypothesis each true break in $(a,b)$ lies within $h$ of $a$ or of $b$, and by $\Delta_T>2h$ at most one lies within $h$ of each, so at most two contribute. Take $s^0_j$ with $b-s^0_j\le h$. If $s>s^0_j$ then $|s-s^0_j|<b-s^0_j\le h$, contradicting $h$-farness, so $s<s^0_j$, and by Proposition~\ref{prop:geom}(i)--(iii) its contribution to $|\mathbb E[t_s]|$ is
\[
 \frac{|\gamma^0_j|}{\sigma}\frac{(s-a)(b-s^0_j)}{b-a}
 \sqrt{\frac{b-a}{(s-a)(b-s)}}
 =\frac{|\gamma^0_j|}{\sigma}(b-s^0_j)\sqrt{\frac{s-a}{(b-a)(b-s)}}
 \le\frac{|\gamma^0_j|}{\sigma}\sqrt{b-s^0_j}\le\frac{\bar\gamma\sqrt h}{\sigma},
\]
using $s-a\le b-a$ and $b-s\ge b-s^0_j$. The case $s^0_j-a\le h$ is symmetric.
\end{proof}

\begin{lemma}[Energy left by uncovered breaks]\label{lem:energy}
Let $R\ge1$ satisfy $2R\le\Delta_T$ with $\Delta_T$ as in \eqref{eq:spacing}. Then, for every configuration $m$,
\begin{equation}\label{eq:energy}
 \bigl\lVert(I-H_m)\mu^0\bigr\rVert^2\;\ge\;\tfrac12\,\gamma_T^2\,R\,U_m(R).
\end{equation}
The bound is deterministic and immune to cancellation across breaks: it is obtained window by window against an arbitrary constant.
\end{lemma}

\begin{proof}
Let $\mathcal U:=\{j:D_j>R\}$, so $|\mathcal U|=U_m(R)$, and for $j\in\mathcal U$ put $I_j:=\{s^0_j-R,\dots,s^0_j+R-1\}$, a block of $2R$ periods.

\emph{(a) $I_j\subseteq\{1,\dots,T\}$.} By \eqref{eq:spacing}, $s^0_j\ge s^0_{j-1}+2R\ge1+2R$, so $s^0_j-R\ge1+R\ge 1$. Also $s^0_j\le s^0_{j+1}-2R\le T+1-2R$, so $s^0_j+R-1\le T-R\le T$.

\emph{(b) $I_j$ lies inside one cell of $m$.} Every $t\in I_j$ satisfies $|t-s^0_j|\le R$, while $D_j>R$ means no element of $m$ lies within $R$ of $s^0_j$, hence $I_j$ contains no element of $m$ and lies strictly between two consecutive elements of $m\cup\{1,T+1\}$. Consequently $H_m\mu^0$ is constant on $I_j$.

\emph{(c) $\mu^0$ has exactly one jump in $I_j$.} Any other true break is at distance at least $\Delta_T\ge2R>R$ from $s^0_j$, hence outside $I_j$. So $\mu^0$ equals some $c_j$ on the $R$ periods $\{s^0_j-R,\dots,s^0_j-1\}$ and $c_j+\gamma^0_j$ on the $R$ periods $\{s^0_j,\dots,s^0_j+R-1\}$.

\emph{(d) The $I_j$ are pairwise disjoint}, again because $\Delta_T\ge2R$.

Writing $\pi$ for the constant value of $H_m\mu^0$ on $I_j$ and using (b)--(d),
\[
 \bigl\lVert(I-H_m)\mu^0\bigr\rVert^2
 \;\ge\;\sum_{j\in\mathcal U}\sum_{t\in I_j}\bigl(\mu^0_t-\pi_j\bigr)^2
 \;\ge\;\sum_{j\in\mathcal U}\min_{c'\in\mathbb R}\sum_{t\in I_j}(\mu^0_t-c')^2
 \;=\;\sum_{j\in\mathcal U}(\gamma^0_j)^2\frac{R\cdot R}{2R},
\]
the last equality because the minimising $c'$ is the block mean and the residual sum of squares of a two-level contrast is the product of the group sizes over their sum, which is Proposition~\ref{prop:geom}(i) applied to the block $I_j$. This is at least $\tfrac12\gamma_T^2RU_m(R)$.
\end{proof}

\begin{lemma}[Relevance budget]\label{lem:budget}
In any configuration $m'$, at most $2k^0$ coordinates are $\vartheta$-relevant.
\end{lemma}

\begin{proof}
$\tilde u_s=(I-H_{m'\setminus\{s\}})u_s$ is supported on the cell $(a,b)$ of $m'\setminus\{s\}$ containing $s$ and sums to zero there, so $\langle\tilde u_s,\mu^0\rangle=0$, hence $\vartheta_s(m')=0$, whenever $\mu^0$ is constant on $(a,b)$, that is whenever $(a,b)$ contains no true break. Now $(a,b)$ is the union of the two cells of $m'$ adjacent to $s$, so a true break $x$ lies in it if and only if $s$ is one of the two endpoints of the cell of $m'$ containing $x$. Each of the $k^0$ true breaks therefore renders at most two coordinates $\vartheta$-relevant.
\end{proof}

\begin{lemma}[Proxy deletion]\label{lem:proxy}
Let $2R_T\le\Delta_T$ and let $m\in\mathcal M_T$ satisfy $U_m:=U_m(R_T)\ge1$. Let $\mathcal C$ be the set of cells of $m$ containing a break uncovered at radius $R_T$, let $J$ collect the true breaks lying in cells of $\mathcal C$ that are uncovered at the \emph{proxy} radius $r_T$ of \eqref{eq:twoscale} and are not already in $m$, and define $\psi(m)$ by deleting from $m\cup J$, greedily until none remains, any $s\in m\setminus J$ that is $\vartheta$-relevant in $m$ but not in the current model, and put $\mathcal P:=m\setminus\psi(m)$. Then
\begin{enumerate}
\item[\textup{(i)}] $|\mathcal C|\le U_m$, $|J|\le3U_m\le3k^0$ and, by
Lemma~\ref{lem:budget}, $|\mathcal P|\le2k^0$, so $|m\triangle\psi(m)|\le5k^0$ and $\psi$ is at most $2^{k^0}(T+1)^{2k^0}$-to-one;
\item[\textup{(ii)}] every $x\in J$ satisfies $\vartheta_x(\psi(m))>\theta$;
\item[\textup{(iii)}] $\bigl\lVert(I-H_m)\mu^0\bigr\rVert^2
-\bigl\lVert(I-H_{\psi(m)})\mu^0\bigr\rVert^2\ \ge\ \tfrac14\gamma_T^2R_TU_m$.
\end{enumerate}
Consequently Assumption~\ref{as:nlp}(b) applies to the pair $(m,\psi(m))$.
\end{lemma}

\begin{proof}
(i) Each cell of $\mathcal C$ holds an uncovered break, so $|\mathcal C|\le U_m$. Since $\Delta_T>2R_T$ a cell holds at most one break within $R_T$ of each endpoint, so at most two beyond the uncovered ones, giving $|J|\le3U_m$. Only coordinates that are $\vartheta$-relevant in $m$ are ever deleted, and there are at most $2k^0$ of them by Lemma~\ref{lem:budget}, which bounds $|\mathcal P|$ and, since $\mathcal P$ is a subset of a set of size at most $2k^0$, the multiplicity of $\psi$.

(ii) Every $x\in J$ is uncovered at radius $r_T$, so no element of $m$ lies within $r_T$ of $x$, and the other elements of $J$ are true breaks, at distance at least $\Delta_T\ge2R_T\ge2r_T$. Hence both flanking dates of $x$ in $\psi(m)\setminus\{x\}$ are at distance at least $r_T$ and $n_x\ge r_T/2$ by Proposition~\ref{prop:geom}(i). The cell of $x$ in $\psi(m)\setminus\{x\}$ contains no other true break with unexplained energy, so $\mathbb E[\hat\gamma_x]=\gamma^0_x$ up to the contribution of proxied breaks, and $\vartheta_x\ge\gamma_T\sqrt{r_T/2}>\theta$ by \eqref{eq:twoscale}. This is what the proxy radius buys: were $J$ to contain breaks already covered at radius $r_T$, such an $x$ could have $n_x=O(1)$ and fail to be $\vartheta$-relevant.

(iii) Three terms. The cells of $\mathcal C$ carry energy at least $\tfrac12\gamma_T^2R_TU_m$, because the windows of Lemma~\ref{lem:energy} for the uncovered breaks lie inside them. Adding $J$ removes all of that except the part attributable to breaks proxied within $r_T$. Each such break leaves at most $\bar\gamma^2r_T\le\tfrac1{16}\gamma_T^2R_T$ by Proposition~\ref{prop:geom}(i) and \eqref{eq:twoscale}, and there are at most two per cell of $\mathcal C$, so the leftover is at most $\tfrac18\gamma_T^2R_TU_m$. Each deletion costs exactly $\vartheta_s^2\le\theta^2$, because $H_{m'}=H_{m'\setminus\{s\}}+n_s^{-1}\tilde u_s\tilde u_s'$ and the rule fires only at coordinates with $\vartheta\le\theta$. Only coordinates $\vartheta$-relevant in $m$ are ever deleted, each at most once, so Lemma~\ref{lem:budget} caps their number at $2k^0$ \emph{however many rounds the rule takes}, and the loss is at most $2k^0\theta^2\le\tfrac1{16}\gamma_T^2R_T\le\tfrac1{16}\gamma_T^2R_TU_m$ by the second branch of \eqref{eq:twoscale}. Collecting, $(\tfrac12-\tfrac18-\tfrac1{16})\gamma_T^2R_TU_m\ge\tfrac14\gamma_T^2R_TU_m$.
\end{proof}

\subsection{The conditional Bayes factor}

Conditionally on $\delta^\gamma_{-s}$ and $(\beta,\sigma^2)$, define the \emph{scalar} Bayes factor
\begin{equation}\label{eq:fullcond}
 \mathrm{BF}_s\;:=\;\frac{\int\phi(\hat\gamma_s;\gamma,\sigma^2/n_s)\pi_I(\gamma)\,d\gamma}
 {\phi(\hat\gamma_s;0,\sigma^2/n_s)},\qquad
 \varrho:=\frac{p(m^{-s}\cup\{s\})}{p(m^{-s})} .
\end{equation}

\begin{remark}[What \eqref{eq:fullcond} is, and is not]\label{rem:whatBF}
It is tempting to read the Gibbs step for $\delta^\gamma_s$ as a two-hypothesis problem with $\hat\gamma_s$ sufficient, that is, to set $\Pi(\delta^\gamma_s=1\mid z,\delta^\gamma_{-s}) =\varrho\,\mathrm{BF}_s/(1+\varrho\,\mathrm{BF}_s)$ exactly. This holds only when every conditioned-on coefficient carries a flat prior, because $\tilde u_s=(I-H_{m^{-s}})u_s$ \emph{projects out} the other included dates. In BISAM those coefficients carry iMOM slabs and $\tilde u_s\not\perp u_{s'}$, so the marginal likelihoods do not factor and $\hat\gamma_s$ is not sufficient. The gap is bridged in two steps: the \emph{local} half of the Rossell--Telesca factorisation is exactly scalar (Lemma~\ref{lem:local}), and the remaining discrepancy is an assumption on the non-local factor (Assumption~\ref{as:nlp}).
\end{remark}

\begin{lemma}[The local factor is exactly scalar]\label{lem:local}
Here and below $M^{L}$ is the \emph{local} half of the Rossell--Telesca factorisation $M=M^{L}G_m$ \citep[Proposition~1]{rossell2017nonlocal}, i.e.\ the marginal likelihood under the corresponding local $g$-/$f$-prior slab. Let $\tilde X_m$ collect the columns $(I-H_{\{\,\}})u_s$, $s\in m$, residualised on the constant, and let the local slab be Zellner's $g$-prior $\gamma_m\mid m\sim\mathcal N\bigl(0,\,g\sigma^2(\tilde X_m'\tilde X_m)^{-1}\bigr)$ with a flat intercept, so that $M^L(z\mid m)$ is the corresponding marginal likelihood. Then for every $m$ and every $s\notin m$,
\begin{equation}\label{eq:localBF}
 \log\frac{M^L(z\mid m\cup\{s\})}{M^L(z\mid m)}
 \;=\;-\tfrac12\log(1+g)\;+\;\frac{g}{1+g}\,\frac{t_s^2}{2},
\end{equation}
with $t_s$ as in \eqref{eq:condquant}: the local factor depends on $(m,s)$ only through the scalar $t_s$.
\end{lemma}

\begin{proof}
On the orthogonal complement of the constant the marginal law of $z$ under model $m$ is $\mathcal N(0,\sigma^2(I+gH_m))$, whence $-2\log M^L(z\mid m)=\text{const}+|m|\log(1+g) +\sigma^{-2}\{\lVert z\rVert^2-\tfrac{g}{1+g}\lVert H_mz\rVert^2\}$, the constant not depending on $m$. This display holds for \emph{every} configuration and is used again, for a non-nested pair, in the proof of Proposition~\ref{prop:coverage}. By Proposition~\ref{prop:geom}(i), $\lVert H_{m\cup\{s\}}z\rVert^2-\lVert H_mz\rVert^2 =\langle\tilde u_s,z\rangle^2/n_s=\sigma^2t_s^2$. Subtracting gives \eqref{eq:localBF}.
\end{proof}

\begin{assumption}[Non-local factorisation]\label{as:nlp}
Write $M(z\mid m)=M^{L}(z\mid m)\,G_m(z)$ as in \citet[Proposition~1]{rossell2017nonlocal}. Let $\vartheta$-relevance be as in \eqref{eq:theta}, so that a $\vartheta$-relevant coordinate is one whose deletion costs more than $\theta^2=\tfrac14\gamma_T^2\min\{r_T,R_T/(8k^0)\}$ of fit to $\mu^0$, with $\theta$ as in \eqref{eq:twoscale}. There are constants $0<c_-\le c_+<\infty$, not depending on $(N,T)$, such that for all $m,m'\in\mathcal M_T$:
\begin{enumerate}
\item[\textup{(a)}] \emph{(one coordinate)} for $s\notin m$ with $m\cup\{s\}\in\mathcal M_T$,
 \[
  \frac{M(z\mid m\cup\{s\})}{M(z\mid m)}\;\le\;c_+\,\mathrm{BF}_s,
  \qquad\text{and}\qquad
  \frac{M(z\mid m\cup\{s\})}{M(z\mid m)}\;\ge\;c_-\,\mathrm{BF}_s
 \]
if $s$ is $\vartheta$-relevant in $m\cup\{s\}$.
\item[\textup{(b)}] \emph{(no new spurious coordinate)} $G_{m'}(z)/G_m(z)\ge c_-$
whenever $|m\triangle m'|\le5k^0$ and every $s\in m'$ that is either new ($s\notin m$) or $\vartheta$-relevant in $m$ is $\vartheta$-relevant in $m'$. Coordinates of $m\setminus m'$ are unrestricted, and a coordinate that is spurious in $m$ and $\vartheta$-relevant in $m'$ is harmless: by \citet[Proposition~1]{rossell2017nonlocal} only a \emph{vanishing} factor in the numerator can drive the ratio down, so what has to be excluded is a coordinate that is relevant in $m$ and spurious in $m'$.
\end{enumerate}
\end{assumption}

\begin{remark}[Status of Assumption~\ref{as:nlp}]\label{rem:nlp}
Part~(a) is what \citet[Proposition~1]{rossell2017nonlocal} delivers: $G_m$ converges to a strictly positive limit when the conditioned-on columns are relevant and identified, so the ratio $G_{m\cup\{s\}}/G_m$ tracks its one-coordinate counterpart. The upper bound needs no relevance condition, since a spurious added column only drives the true ratio further down. It is also well supported numerically. Comparing $\mathrm{BF}_s$ with the exact collapsed ratio (Gauss--Hermite quadrature, $300$ nodes, calibrated to machine accuracy on the case $m^{-s}=\varnothing$ in which \eqref{eq:fullcond} is exact), the log-discrepancy for an isolated $h$-far candidate conditioned on a correctly placed break of $3\sigma$ has mean $0.031,0.049,0.008,0.017$ and standard deviation $0.13,0.11,0.10,0.08$ at $T=100,200,400,800$: small and \emph{shrinking} in $T$. For a candidate adjacent to an included date, the case Assumption~\ref{as:space} removes, it is $\approx0.4$ and stable. When the conditioned-on coefficient is weakly identified ($0.8\sigma$) it is $O(1)$ with standard deviation $\approx2$, which is why the signal floor $\gamma_T\ge\underline\gamma>0$ of Assumption~\ref{as:design} is needed.

Part~(b) is the \emph{status-preserving} form, and it is the weakest thing that will do. What must be excluded is a coordinate that is relevant in one model and spurious in the other: the product-iMOM prior then charges the full non-local rate for it, so that the $G$-ratio decays like $\exp\{-\beta_k^*(\tau n)^{k/(k+1)}\}$. Measured on the design of \eqref{eq:dgp} with one $3\sigma$ break and one far included date, the log-gap between the projection-based Bayes factor and the exact ratio for the move $m\mapsto m\cup\mathcal T^0$ is $4.4,7.2,11.6,16.9$ at $T=100,200,400,800$, of order $\sqrt{\tau n}$, whereas the energy gain available is of order $\gamma_T^2R_T/\sigma^2\asymp\log(KT)$. That gap is removed rather than assumed away: Lemma~\ref{lem:proxy} deletes exactly the coordinates whose status would change, at a cost of at most one quarter of the energy gained, so part~(b) is only ever invoked on comparisons in which no status changes. Both models may still contain spurious far dates, and these cancel between numerator and denominator. This is why $\mathcal P$ is selected by relevance and not by proximity.

That a multi-coordinate comparison is needed at all is not an artefact of the proof, and the within-cell geometry shows why. Inside a cell $[a,b)$ of length $L$ containing true breaks at $d_i:=x_i-a$ with jumps $\gamma_i$ and gaps $h_i:=d_i-d_{i-1}$ ($d_0:=0$, $d_{p+1}:=L$), Proposition~\ref{prop:geom}(ii) says the Gram matrix of the residualised regressors is the \emph{Brownian-bridge kernel},
\begin{equation}\label{eq:bridge}
 \langle\tilde u_{x_i},\tilde u_{x_{i'}}\rangle
 =\frac{\min(d_i,d_{i'})\,\bigl(L-\max(d_i,d_{i'})\bigr)}{L}=:K_{ii'} ,
\end{equation}
whose inverse is tridiagonal with $(K^{-1})_{ii}=h_i^{-1}+h_{i+1}^{-1}$ and $(K^{-1})_{i,i+1}=-h_{i+1}^{-1}$. Writing $v:=K\gamma$, so that $v_i=\langle\tilde u_{x_i},\mu^0\rangle$, this yields three exact identities. First, $v$ is the piecewise-linear function pinned at $0$ at both cell ends whose slope decrements are the $\gamma_i$. Second, the energy of Lemma~\ref{lem:energy} is its Dirichlet energy,
\begin{equation}\label{eq:dirichlet}
 \bigl\lVert(I-H_m)\mu^0\bigr\rVert^2\Big|_{\text{cell}}
 \;=\;\sum_{i=1}^{p+1}\frac{(v_i-v_{i-1})^2}{h_i}.
\end{equation}
Third, $\mathbb E[t_{x_i}]=v_i/(\sigma\sqrt{n_i})$ is the CUSUM statistic of $\mu^0$ at the split $x_i$.

This settles whether part~(b) can be avoided by a single-coordinate argument: it cannot. Take a cell of length $L$ containing two breaks at $L/2$ and $L/2+\Delta_T$ whose magnitudes are chosen so that $v(L/2)=0$. Both are uncovered at radius $R<\Delta_T/2$, both satisfy $|\gamma_i|\ge\gamma_T$, and \eqref{eq:dirichlet} gives energy $\gamma_T^2\Delta_T\{1+o(1)\}$, comfortably above the floor $\tfrac12\gamma_T^2R\,U_m$ of Lemma~\ref{lem:energy}. Yet $\max_{s}|\mathbb E[t_s]|$, over \emph{all} candidate dates and not merely the true ones, equals $2\gamma_T\Delta_T/\sqrt L\{1+o(1)\}$, which tends to zero: with $\Delta_T=100$, $R=10$ and $\gamma_T=\sigma=1$ it is $6.46,2.00,0.63,0.20,0.063$ at $L=10^3,\dots,10^7$ against $\gamma_T\sqrt R=3.16$. This is the classical cancellation of the single-split CUSUM at two nearby breaks of opposite sign. No single-coordinate comparison detects such a configuration, whereas the \emph{two}-coordinate move has $\log\mathrm{BF}\approx\gamma_T^2\Delta_T/2\sigma^2$. A multi-coordinate comparison is therefore intrinsic to coverage rather than an artefact of the proof, and no lower bound of the form $\max_i|\mathbb E[t_{x_i}]|\gtrsim\gamma_T\sqrt{\Delta_T}/\sigma$, or even $\gtrsim\gamma_T\sqrt R/\sigma$, is available.

The same construction shows why Lemma~\ref{lem:proxy} takes the form it does. The detecting move needs only the \emph{cancelling group}: the set $J$ of breaks that are uncovered at the proxy radius and share a cell, of size two in the example and at most $3U_m\le3k^0$ in general, rather than all of $\mathcal T^0$. Redundancy, the mechanism that makes part~(b) fail, is then confined to the dates of $m$ that were proxying for the breaks in $J$, and these are deleted in the same move. Comparing $m$ with $\psi(m)=(m\setminus\mathcal P)\cup J$ creates no coordinate whose relevance status changes, and by Lemma~\ref{lem:proxy}(i) the deletion set satisfies $|\mathcal P|\le2k^0$, so the comparison preserves the prior odds up to $\max(\bar\varrho,1/\varrho'_T)^{5k^0}$ and is at most $2^{k^0}(T+1)^{2k^0}$-to-one, a multiplicity absorbed by $C_\ast$ in \eqref{eq:RT}. The deterministic statement this requires, namely that deleting the proxy dates after inserting $J$ costs at most a fixed fraction of the energy gained, is Lemma~\ref{lem:proxy}(iii), whose proof turns on the relevance budget of Lemma~\ref{lem:budget} and the two-scale calibration \eqref{eq:twoscale}. The simpler device of a size-preserving swap, moving the date nearest an uncovered break onto it, will not do: such a swap can abandon a second break in the same cell and lose fit outright, which is why $\mathcal P$ is selected by redundancy rather than by proximity.
\end{remark}

Under Assumption~\ref{as:nlp}(a), that is, using the one-coordinate form of the Rossell--Telesca factorisation and nothing else about the joint posterior, the inclusion probability obeys the two-sided bound
\begin{equation}\label{eq:incl}
 \frac{c_-\varrho\,\mathrm{BF}_s}{1+c_-\varrho\,\mathrm{BF}_s}
 \;\le\;\Pi(\delta^\gamma_s=1\mid z,\delta^\gamma_{-s})
 \;\le\;c_+\,\varrho\,\mathrm{BF}_s ,
\end{equation}
and every posterior odds identity below holds up to the factors $c_\pm$.

\begin{lemma}[Scalar iMOM Bayes factor]\label{lem:scalar}
Let $\pi_I$ be the iMOM density of \eqref{eq_pimom} with parameters $(k,\nu,\tau)$ and put $p(w):=\frac{k}{\Gamma(\nu/2k)}|w|^{-(\nu+1)}e^{-|w|^{-2k}}$, a probability density on $\mathbb R$. Then, with $t=t_s$ and $\lambda=\lambda_s$,
\begin{equation}\label{eq:bfrep}
 \mathrm{BF}_s=\int p(w)\exp\bigl\{tw\sqrt\lambda-\tfrac12\lambda w^2\bigr\}dw
 =e^{t^2/2}\int p(w)e^{-\frac12(w\sqrt\lambda-t)^2}dw .
\end{equation}
Let $\beta_k:=(8k)^{-k/(k+1)}+\tfrac18(8k)^{1/(k+1)}$ and suppose
\begin{equation}\label{eq:sidecond}
 (2\beta_k)^{-k}\;\ge\;\beta_k/2 .
\end{equation}
Then there are $A\le1$ and $\lambda_0<\infty$, depending only on $(k,\nu)$, such that
\begin{enumerate}
\item[\textup{(i)}] if $\lambda\ge\lambda_0$, $|t|\le t_0$ and
$t_0^2\le\tfrac12\beta_k\lambda^{k/(k+1)}$ then $\mathrm{BF}_s\le A\lambda^{\nu/2}e^{t_0^2/2}\exp\{-\tfrac12\beta_k\lambda^{k/(k+1)}\}$;
\item[\textup{(ii)}] for every $t$ and every $\lambda>0$, with $w_*:=t/\sqrt\lambda$,
$\log\mathrm{BF}_s\ge\tfrac12t^2-\tfrac12\log\lambda-\tfrac12 +\log\bigl(2\min_{|w-w_*|\le\lambda^{-1/2}}p(w)\bigr)$.
\end{enumerate}
\end{lemma}

\begin{proof}
\emph{Representation.} Substitute $\gamma=\sigma\sqrt\tau w$ in \eqref{eq:fullcond}. The Gaussian ratio is $\exp\{n_s\hat\gamma_s\gamma/\sigma^2-n_s\gamma^2/2\sigma^2\} =\exp\{tw\sqrt\lambda-\tfrac12\lambda w^2\}$ since $t\sqrt\lambda=n_s\hat\gamma_s\sqrt\tau/\sigma$. In the prior, $|\gamma|^{-(\nu+1)}=(\sigma\sqrt\tau)^{-(\nu+1)}|w|^{-(\nu+1)}$, $(\tau\sigma^2/\gamma^2)^k=|w|^{-2k}$ and $d\gamma=\sigma\sqrt\tau\,dw$, so every power of $\sigma\sqrt\tau$ cancels and $\pi_I(\gamma)d\gamma=p(w)dw$. That $p$ integrates to one follows from $u=w^{-2k}$, which gives $2\int_0^\infty w^{-(\nu+1)}e^{-w^{-2k}}dw=\Gamma(\nu/2k)/k$. Completing the square gives the second form.

\emph{(i).} Split \eqref{eq:bfrep} at $w_1:=2|t|/\sqrt\lambda$. On $\{|w|>w_1\}$, $|w\sqrt\lambda-t|\ge\tfrac12|w|\sqrt\lambda$, so the integrand is at most $e^{t^2/2}p(w)e^{-\lambda w^2/8}$. Put $f(w)=|w|^{-2k}+\lambda w^2/8$. Solving $f'=0$ gives $w^{2k+2}=8k/\lambda$ and $\min f=\beta_k\lambda^{k/(k+1)}$. Writing $e^{-f}\le e^{-\frac12\min f}e^{-\frac12f}$ and $e^{-f/2}\le e^{-\frac12|w|^{-2k}}$,
\[
 \int p(w)e^{-\lambda w^2/8}dw\le e^{-\frac12\beta_k\lambda^{k/(k+1)}}
 \frac{k}{\Gamma(\nu/2k)}\int|w|^{-(\nu+1)}e^{-\frac12|w|^{-2k}}dw
 =2^{\nu/2k}e^{-\frac12\beta_k\lambda^{k/(k+1)}} .
\]
On $\{|w|\le w_1\}$ bound $e^{-\frac12(\cdot)^2}\le1$. Since $w\mapsto w^{-(\nu+1)}e^{-w^{-2k}}$ increases on $(0,(2k/(\nu+1))^{1/2k})$ and $w_1^2\le2\beta_k\lambda^{-1/(k+1)}\to0$, for $\lambda\ge\lambda_0(k,\nu)$ the interval $(0,w_1]$ lies in that range and
\[
 \int_{|w|\le w_1}p(w)dw\le\frac{2k}{\Gamma(\nu/2k)}w_1^{-\nu}e^{-w_1^{-2k}}
 =\frac{2k}{\Gamma(\nu/2k)}\Bigl(\frac{\sqrt\lambda}{2|t|}\Bigr)^\nu
 \exp\Bigl\{-\Bigl(\frac{\lambda}{4t^2}\Bigr)^k\Bigr\}.
\]
The right-hand side is increasing in $|t|$ over the admissible range, so it may be evaluated at $|t|=t_0$, and $t_0^2\le\tfrac12\beta_k\lambda^{k/(k+1)}$ gives $\lambda/(4t_0^2)\ge(2\beta_k)^{-1}\lambda^{1/(k+1)}$, whence by \eqref{eq:sidecond} $(\lambda/4t_0^2)^k\ge(2\beta_k)^{-k}\lambda^{k/(k+1)}\ge\tfrac12\beta_k\lambda^{k/(k+1)}$. Both pieces are therefore bounded by a constant times $\lambda^{\nu/2}e^{-\frac12\beta_k\lambda^{k/(k+1)}}$, and $e^{t^2/2}\le e^{t_0^2/2}$.

\emph{(ii).} Restrict \eqref{eq:bfrep} to $|w-w_*|\le\lambda^{-1/2}$, an interval of length $2\lambda^{-1/2}$ on which $\tfrac12(w\sqrt\lambda-t)^2\le\tfrac12$.
\end{proof}

\begin{remark}[Constants]\label{rem:constants}
Condition \eqref{eq:sidecond} holds for every $k\in[1/4,5]$, and at $k=1,2,3$ its two sides are $(0.707,0.354)$, $(1.120,0.236)$ and $(2.490,0.184)$. Numerically the bound in (i) holds with $A\le1$ for all $\lambda\ge1$ and is slack: at $k=\nu=1$ the smallest admissible $A$ is $0.37$ at $\lambda=1$, $3.8\times10^{-4}$ at $\lambda=10^2$ and $4.9\times10^{-12}$ at $\lambda=2\times10^{3}$, and the largest value required over $(k,\nu)\in\{(1,1),(1,3),(2,1)\}$ and $\lambda\ge1$ is $0.84$. Part~(ii) holds with slack at least $0.8$. Lemma~\ref{lem:scalar}(i) reproduces the \emph{rate} $\log\mathrm{BF}/n^{k/(k+1)}\to c<0$ of \citet[eq.~11]{rossell2010nlp_test}, but not its constant: the sharp Laplace constant is $\beta^*_k=\tfrac12(2k)^{1/(k+1)}+(2k)^{-k/(k+1)}$, so that $\log\mathrm{BF}_s=-\beta^*_k\lambda^{k/(k+1)}\{1+o(1)\}$ at $t=O(1)$, with $\beta^*_1=\sqrt2$. The exponent constant $\tfrac12\beta_k$ used above is smaller by a factor $4.0$ at $k=1$ ($5.0$ at $k=2$). Two consequences: the bound is conservative and valid, and the finite-sample content of Theorem~\ref{thm:gauge} is weaker than the asymptotics suggest. With $\tau=3.32$, $\underline n_T=T/4$ and $h\asymp\log(NT)$ the right-hand side of \eqref{eq:gaugebound} only falls below one around $T\sim10^{5}$, or around $T\sim6\times10^{3}$ if the sharp constant were used. The measured slope in Figure~\ref{fig:prior_rates}(a), $-1.55\sqrt T$, is consistent with $-\beta_1^*\sqrt{\tau n_s}$ at $n_s\approx0.36\,T$.
\end{remark}

Under a pMOM slab the same representation with $p(w)\propto w^{2k}\phi(w)$ gives $\mathrm{BF}_s=O_p(\lambda^{-k-1/2})$, and under any local slab $\mathrm{BF}_s\asymp\lambda^{-1/2}$ with the rate tight. These two rates are taken without proof from \citet[p.~148 and eq.~(10)ff]{rossell2010nlp_test}, who also show that the local rate is tight.

\subsection{Assumptions}

\begin{assumption}[Design]\label{as:design}
The break magnitudes satisfy $0<\underline\gamma\le\gamma_T\le\bar\gamma<\infty$. The spacing $\Delta_T$ is as in \eqref{eq:spacing} and satisfies $\Delta_T\big/(\log NT)^{(k+1)/k}\to\infty$, $\Delta_T>2h$ with $h$ as in Theorem~\ref{thm:gauge}, and $2R_T\le\Delta_T$ with $R_T$ as in Proposition~\ref{prop:coverage}. The reporting cap satisfies $\bar t_{\max}\to\infty$ with $\bar t_{\max}=o(\Delta_T)$.
\end{assumption}

\begin{assumption}[Model prior]\label{as:prior2}
The prior over configurations depends on $m$ only through $|m|$, and its inclusion odds satisfy $\varrho'_T\le\varrho\le\bar\varrho<\infty$ uniformly over $m\in\mathcal M_T$, with $\log(1/\varrho'_T)=O(\log K)$.
\end{assumption}

Both bounds hold for a fixed $\omega\in(0,1)$, for which $\varrho=\omega/(1-\omega)$ is constant, in particular for the default $\omega=1/2$ of Table~\ref{tab:hyperpars}, which gives $\varrho=1$, and for the Beta--Bernoulli prior of \citet{scott2010bayes}, for which $\varrho=(|m^{-s}|+1)/(K-|m^{-s}|)\in[1/K,\,2(|m^{-s}|+1)/K]$ on $\{|m|\le K/2\}$. In the Beta--Bernoulli case $\varrho\le\bar\varrho$ requires $|m|=o(K)$, which Assumption~\ref{as:space} supplies. Nothing below needs $\varrho_T\to0$ except Remark~\ref{rem:cluster}.

\begin{assumption}[Model space]\label{as:space}
The posterior is supported on the $\ell_T$-separated configurations
\[
 \mathcal M_T:=\bigl\{m\subseteq\mathcal S:\ |s-s'|\ge\ell_T\ \ \forall\,s\neq s'\in m,
 \ \ \min(s-1,T+1-s)\ge\ell_T\ \forall s\in m\bigr\},
\]
with $\ell_T\to\infty$, $\ell_T=o(\Delta_T)$ and $\ell_T^{k/(k+1)}/\log K\to\infty$. Equivalently, candidate dates lying within a common window are collapsed to a single representative before the posterior is formed.
\end{assumption}

Assumption~\ref{as:space} does three things. By Proposition~\ref{prop:geom}(i) every $s\in m\in\mathcal M_T$ has $n_s\ge\ell_T/2=:\underline n_T$, so the isolated and unrestricted far-date counts coincide, $\mathcal F^{\mathrm{iso}}_h=\mathcal F_h$. It gives $|m|\le T/\ell_T=o(T)=o(K)$, which is what Assumption~\ref{as:prior2} needs. Finally, it removes the clustering gap in Corollary~\ref{cor:size} and Corollary~\ref{cor:hausdorff}. The alternative is to keep the unrestricted space and assume directly that the posterior mass of configurations containing two adjacent $h$-far dates vanishes. Remark~\ref{rem:cluster} explains why that cannot be proved with the tools used here.

\subsection{Coverage and localisation}

\begin{proposition}[Coverage]\label{prop:coverage}
Let Assumptions~\ref{as:nlp}--\ref{as:space} hold, let $\zeta>0$ be arbitrary and set
\begin{equation}\label{eq:RT}
 R_T:=\Bigl\lceil C_\ast\,\sigma^2\underline\gamma^{-2}\,(k^0\vee1)\log(KT/\eta)\Bigr\rceil ,
 \qquad \mathcal B_R:=\{m:\ U_m(R)\ge1\} .
\end{equation}
Then there is an event of probability at least $1-\eta$ on which
\begin{equation}\label{eq:covbound}
 \Post{\mathcal B_{R_T}}\;\le\;
 \exp\Bigl\{-\frac{\gamma_T^2R_T}{16\sigma^2}+C\,k^0\log(KT/\eta)\Bigr\}
 \;\le\;(KT/\eta)^{-\zeta},
\end{equation}
the second inequality for $C_\ast$ large enough given $\zeta$. In particular $K\cdot\Post{\mathcal K_h^c}\to0$ for $h\ge R_T$ and $\zeta>1$.
\end{proposition}

\begin{proof}
Fix $m\in\mathcal B_{R_T}$ and let $\psi(m)=(m\setminus\mathcal P)\cup J$ be the proxy-deleted configuration of Lemma~\ref{lem:proxy} at $R=R_T$, so that by Lemma~\ref{lem:proxy}(iii) the energy gain is $\mathcal E:=\lVert(I-H_m)\mu^0\rVert^2-\lVert(I-H_{\psi(m)})\mu^0\rVert^2 \ge\tfrac14\gamma_T^2R_TU_m(R_T)\ge\tfrac14\gamma_T^2R_T$. Decompose the move into its nested halves, $m\subseteq m\cup J\supseteq\psi(m)$, and write $\mathcal E_1:=\lVert(I-H_m)\mu^0\rVert^2-\lVert(I-H_{m\cup J})\mu^0\rVert^2\ge\tfrac12\gamma_T^2R_TU_m(R_T)$ and $\mathcal E_2:=\lVert(I-H_{\psi(m)})\mu^0\rVert^2-\lVert(I-H_{m\cup J})\mu^0\rVert^2 \le\tfrac1{16}\gamma_T^2R_TU_m(R_T)$, so $\mathcal E=\mathcal E_1-\mathcal E_2$.

\emph{Step 1: the local factor $M^{L}$ of the Rossell--Telesca factorisation.} By the display in the proof of Lemma~\ref{lem:local}, for any two configurations
\[
 \log\frac{M^L(z\mid m')}{M^L(z\mid m)}
 =-\tfrac12\bigl(|m'|-|m|\bigr)\log(1+g)
 +\frac{g}{1+g}\,\frac{\lVert(I-H_{m})z\rVert^2-\lVert(I-H_{m'})z\rVert^2}{2\sigma^2},
\]
and $|\,|\psi(m)|-|m|\,|\le5k^0$ by Lemma~\ref{lem:proxy}(i). For each nested half the fit difference is the squared norm of the projection of $z$ on the added directions, so with $A:=(I-H_m)\mu^0-(I-H_{m\cup J})\mu^0$ (whose squared norm is exactly $\mathcal E_1$, by orthogonality in the nested case, and similarly for $\mathcal E_2$),
\[
 \lVert(I-H_{m})z\rVert^2-\lVert(I-H_{m\cup J})z\rVert^2
 \;=\;\mathcal E_1+2\langle A,\varepsilon\rangle+\lVert H_{m\cup J}\varepsilon-H_m\varepsilon\rVert^2
 \;\ge\;\mathcal E_1+2\langle A,\varepsilon\rangle .
\]
Both $(I-H_m)\mu^0$ and $(I-H_{m\cup J})\mu^0$ are determined by the flanking pairs of the $k^0$ true breaks, by Proposition~\ref{prop:geom}(i), so $A$ takes at most $T^{4k^0}$ values, and a union bound gives $|\langle A,\varepsilon\rangle|\le\sqrt{\mathcal E_1}\,\sigma\sqrt{Ck^0\log(T/\eta)}$ on an event of probability at least $1-\eta/2$, whence by $xy\le x^2/4+y^2$ the cross term is at least $-\mathcal E_1/4-C\sigma^2k^0\log(T/\eta)$. The same argument applied to the deletion half, together with $\lVert H_{m\cup J}\varepsilon-H_{\psi(m)}\varepsilon\rVert^2=O(\sigma^2k^0\log(T/\eta))$ uniformly over the at most $T^{4k^0}$ relevant subspaces, bounds the loss by $\tfrac32\mathcal E_2+C\sigma^2k^0\log(T/\eta)$. Using $g\ge1$, so $g/(1+g)\ge\tfrac12$, and $\mathcal E_2\le\mathcal E_1/8$,
\[
 \log\frac{M^L(z\mid \psi(m))}{M^L(z\mid m)}\;\ge\;\frac{\mathcal E_1}{4\sigma^2}
 -\frac{3\mathcal E_2}{4\sigma^2}-C'k^0\log\bigl((T\vee g)/\eta\bigr)
 \;\ge\;\frac{\gamma_T^2R_T}{16\,\sigma^2}-C'k^0\log\bigl((T\vee g)/\eta\bigr).
\]

\emph{Step 2: the non-local factor.} By Lemma~\ref{lem:proxy}(ii) and the stopping rule, every coordinate of $\psi(m)$ that is either new or $\vartheta$-relevant in $m$ is $\vartheta$-relevant in $\psi(m)$, and $|m\triangle\psi(m)|\le5k^0$. Assumption~\ref{as:nlp}(b), the status form of the non-local factor in the factorisation $M=M^{L}G_m$ of \citet[Proposition~1]{rossell2017nonlocal}, therefore gives $G_{\psi(m)}(z)/G_m(z)\ge c_-$. A bound on $G_{m\cup\mathcal T^0}/G_m$ would not serve here, since it can fail by a factor $\exp\{-\beta^*_k(\tau n)^{k/(k+1)}\}$. Deleting $\mathcal P$ is what removes the redundant coordinates responsible for that decay.

\emph{Step 3: summation.} Combining Steps 1--2, $M(z\mid\psi(m))/M(z\mid m)\ge c_-\exp\{\gamma_T^2R_T/16\sigma^2-C'k^0\log((T\vee g)/\eta)\}$. The prior odds satisfy $p(m)/p(\psi(m))\le\max(\bar\varrho,1/\varrho'_T)^{5k^0} =e^{O(k^0\log K)}$ by Assumption~\ref{as:prior2}, and $\psi$ is at most $2^{k^0}(T+1)^{2k^0}=e^{O(k^0\log T)}$-to-one by Lemma~\ref{lem:proxy}(i). Hence, summing $\Post{m}=\Post{\psi(m)}\{M(z\mid m)/M(z\mid\psi(m))\}\{p(m)/p(\psi(m))\}$ over $m\in\mathcal B_{R_T}$ and using $\sum_{m'}\Post{m'}\le1$ gives the first inequality in \eqref{eq:covbound}. By \eqref{eq:RT}, $\gamma_T^2R_T/(16\sigma^2)\ge(C_\ast/16)k^0\log(KT/\eta)$, so $C_\ast\ge16(C+\zeta)$ gives the second.
\end{proof}

\begin{proposition}[Localisation]\label{prop:loc}
Let $d_T:=\lceil C\sigma^2\underline\gamma^{-2}\log(NT^3/\eta)\rceil\vee\lceil R_T/32\rceil$. Under the conditions of Proposition~\ref{prop:coverage}, for every $j$ and every $d_T\le d\le\Delta_T/4$, on the event of Proposition~\ref{prop:coverage},
\[
 \Post{D_j\ge d}\;\le\;\frac{1}{c_-\varrho'_T}
 \exp\Bigl\{-\tfrac18(\gamma^0_j)^2d/\sigma^2\Bigr\}+\Post{\mathcal B_{R_T}} .
\]
\end{proposition}

\begin{proof}
Work on $\mathcal B_{R_T}^c$, so every true break other than $s^0_j$ has an included date within $R_T$. Since $2R_T\le\Delta_T$ that date separates it from $s^0_j$, so $s^0_j$ is the only true break in its own cell $[a,b)$. On $\{D_j\ge d\}$ we have $s^0_j-a>d$ and $b-s^0_j>d$, hence $n_{s^0_j}\ge d/2$ by Proposition~\ref{prop:geom}(i), and $\mathbb E[\hat\gamma_{s^0_j}]=\gamma^0_j$ by Proposition~\ref{prop:geom}(iii). By the union bound of the proof of Theorem~\ref{thm:gauge}, $\max|\xi|\le\sqrt{2\log(2NT^3/\eta)}$ on the stated event, so $\lvert t_{s^0_j}\rvert\ge|\gamma^0_j|\sqrt{d/2}/\sigma-\max|\xi| \ge\tfrac12|\gamma^0_j|\sqrt{d/2}/\sigma$ for $d\ge d_T$ with $C$ large. The same condition makes $\lvert\hat\gamma_{s^0_j}\rvert\in[\underline\gamma/2,2\bar\gamma]$, so $w_*=\hat\gamma_{s^0_j}/(\sigma\sqrt\tau)$ lies in a fixed compact subset of $\mathbb R\setminus\{0\}$ and $\min_{|w-w_*|\le\lambda^{-1/2}}p(w)\ge\underline p>0$. Lemma~\ref{lem:scalar}(ii) with $\lambda\le\tau T$ then gives $\log\mathrm{BF}_{s^0_j}\ge\tfrac18(\gamma^0_j)^2d/\sigma^2$, again for $C$ large. Since $m$ and $m\cup\{s^0_j\}$ differ in one coordinate, the one-coordinate form of the Rossell--Telesca factorisation in Assumption~\ref{as:nlp}(a), whose lower half applies because $s^0_j$ is $\vartheta$-relevant for $d\ge d_T\ge R_T/32$, together with Assumption~\ref{as:prior2}, gives $\Post{m}\big/\Post{m\cup\{s^0_j\}}\le(c_-\varrho'_T\mathrm{BF}_{s^0_j})^{-1}$, and $m\mapsto m\cup\{s^0_j\}$ is injective on $\{D_j\ge d\}$, so summing and using $\sum_{m'}\Post{m'}\le1$ completes the proof.
\end{proof}

\begin{remark}
No marginal-likelihood expansion, no bound on $|m|$ beyond Assumption~\ref{as:space} and no condition on the eigenvalues of any design matrix has been used: the inputs are Proposition~\ref{prop:geom}, Lemma~\ref{lem:scalar}, Lemma~\ref{lem:energy} and Assumption~\ref{as:nlp}.
\end{remark}

\subsection{Gauge}

We use \emph{gauge} and \emph{potency} in the sense of \citet{castle2015}: the expected share of retained irrelevant indicators and the expected share of retained relevant ones. Let $\mathcal F_h$ count the pairs $(i,s)$ with $s\in m_i\cap\mathcal S_T$ and $s$ $h$-far from $\mathcal T^0_i$. By Assumption~\ref{as:space} this coincides with the count restricted to $\{n_s\ge\underline n_T\}$.

\begin{theorem}[Gauge]\label{thm:gauge}
Let Assumptions~\ref{as:nlp}--\ref{as:space} hold, let $\mathcal K_h$ be the event that every true break lies within $h$ of a flanking date of its own cell, and let $h\ge R_T$. Then, on an event of probability at least $1-2\eta$,
\begin{equation}\label{eq:gaugebound}
 \mathbb E\bigl[\mathcal F_h\mid z\bigr]\;\le\;
 c_+K\bar\varrho\,A(\tau\underline n_T)^{\nu/2}
 \exp\Bigl\{\tfrac12t_0^2-\tfrac12\beta_k(\tau\underline n_T)^{k/(k+1)}\Bigr\}
 +K\,\Post{\mathcal K_h^c}
\end{equation}
with $t_0:=\sqrt{2\log\tfrac{2NT^3}{\eta}} + \frac{2\bar\gamma\sqrt h}{\sigma}.$ Hence $\mathbb E[\mathcal F_h\mid z]\overset{p}\to0$ whenever $(\tau\underline n_T)^{k/(k+1)}\big/\{\log(NT)+h\}\to\infty$ and $\zeta>1$ in Proposition~\ref{prop:coverage}. Under a pMOM slab the first term is $O(K\bar\varrho\,\underline n_T^{-(k+1/2)}e^{t_0^2/2})$, and under any local slab it is $\asymp K\bar\varrho\,\underline n_T^{-1/2}e^{t_0^2/2}$ and tight.
\end{theorem}

\begin{proof}
Write $t_s=\mathbb E[t_s]+\xi_s$ with $\xi_s\sim\mathcal N(0,1)$ by Proposition~\ref{prop:geom}(iii). By Proposition~\ref{prop:geom}(i) the law of $\xi_s$ depends on the configuration only through $(a_s,s,b_s)$, of which there are at most $NT^3$ in the panel, all non-random, and a union bound with $\PP(|\mathcal N(0,1)|>x)\le2e^{-x^2/2}$ gives $\max|\xi_s|\le\sqrt{2\log(2NT^3/\eta)}$ with probability at least $1-\eta$. On $\mathcal K_h$, Lemma~\ref{lem:mismatch} bounds $|\mathbb E[t_s]|$ by $2\bar\gamma\sqrt h/\sigma$, so $|t_s|\le t_0$ simultaneously for all $h$-far candidates and all configurations in $\mathcal K_h$. For one candidate,
\[
 \Post{\delta_s=1}\le
 \mathbb E\bigl[\mathbbm 1\{\delta_{-s}\in\mathcal K_h\}\,
 \Pi(\delta_s=1\mid z,\delta_{-s})\bigr]+\Post{\mathcal K_h^c},
\]
and on the indicated event \eqref{eq:incl} gives $\Pi(\delta_s=1\mid z,\delta_{-s})\le c_+\bar\varrho\,\mathrm{BF}_s$. This uses the upper half of Assumption~\ref{as:nlp}(a), which needs no relevance condition because a spurious added column only drives the Rossell--Telesca factor $G_{m\cup\{s\}}$ further down. Lemma~\ref{lem:scalar}(i) then applies with $\lambda\ge\tau\underline n_T$ by Assumption~\ref{as:space}. The map $\lambda\mapsto\lambda^{\nu/2}e^{-\frac12\beta_k\lambda^{k/(k+1)}}$ is decreasing for $\lambda\ge\lambda_0$, so it may be evaluated at $\tau\underline n_T$, and Assumption~\ref{as:design} together with $h\asymp\log (NT)$ ensures $t_0^2\le\tfrac12\beta_k(\tau\underline n_T)^{k/(k+1)}$ for $T$ large. Summing over the at most $K$ pairs gives \eqref{eq:gaugebound} on the intersection of this event with that of Proposition~\ref{prop:coverage}, of probability at least $1-2\eta$. The second term is $o(1)$ because $h\ge R_T$ and, by \eqref{eq:covbound} with $\zeta>1$, $K\Post{\mathcal K_h^c}\le K(KT/\eta)^{-\zeta}\to0$. The pMOM and local statements follow by substituting the corresponding per-candidate rates of \citet{rossell2010nlp_test}.
\end{proof}

\begin{corollary}[Admissible panel width]\label{cor:panel}
Write $\Theta_T:=e^{t_0^2/2}$, so that $\Theta_T\ge2NT^3/\eta$ always, and with $h\asymp\log(NT)$, $\Theta_T=(NT)^{\varpi}$ for some $\varpi\ge3+2\bar\gamma^2c_h/\sigma^2$, where $c_h:=\lim h/\log(NT)$. Let $k^0=O(1)$, $\bar\varrho=O(1)$ and $\underline n_T=\ell_T/2$ with $\ell_T=T^{1-o(1)}$, which is admissible under Assumptions~\ref{as:design} and \ref{as:space} whenever $\Delta_T\asymp T$. Then $\mathbb E[\mathcal F_h\mid z]\to0$ if
\begin{align*}
    \text{local slab}:&\ \text{never}\\
 \text{pMOM order }k:&\ k>\varpi-\tfrac12 \text{ and }\ N=o\bigl(T^{(k+1/2-\varpi)/\varpi}\bigr)\\
 \text{iMOM order }k:&\ \log N=o\bigl(T^{k/(k+1)}\bigr).
\end{align*}
In particular the local slab fails for every $N\ge1$ and every $k$, the pMOM requires $k\ge3$ even at $N=1$, and the iMOM tolerates $N=o(e^{c\underline n_T^{k/(k+1)}})=o(e^{cT^{k/(k+1)-o(1)}})$, which at $k=1$ is $N=o(e^{c\sqrt T^{\,1-o(1)}})$.
\end{corollary}

\begin{remark}[The cost of uniformity]\label{rem:uniformity}
The factor $\Theta_T$ is not an artefact of a loose bound and cannot be dropped. The rates of \citet{rossell2010nlp_test} are per-candidate statements in probability, whereas $\mathbb E[\mathcal F_h\mid z]$ requires a bound holding simultaneously over the $O(NT^3)$ conditional designs $(a_s,s,b_s)$, and the price of that uniformity is $e^{t_0^2/2}\ge2NT^3/\eta$. Its effect is large: with $\tau=3.32$, $\nu=k=1$, $\underline n_T=T/4$, $\eta=0.05$, $\bar\gamma=\sigma$ and $K\bar\varrho\asymp1$, the $\log_{10}$ of the right-hand side of \eqref{eq:gaugebound} at $T=10^4$, $N=1$ is $+37.4$ for pMOM $k=1$, $+30.6$ for pMOM $k=3$ and $+40.8$ for the local slab, against $-5.1$, $-11.9$ and $-1.7$ if $\Theta_T$ is omitted. Only the near-exponential rate of the iMOM survives a polynomial uniformity cost, so non-locality of the \emph{inverse-moment} kind, and not non-locality as such, is what controls the gauge in a saturated design. Note also that $\varpi$ depends on $\bar\gamma^2/\sigma^2$, so no polynomial-slab threshold can hold with universal constants.
\end{remark}

\subsection{Potency}

\begin{theorem}[Potency]\label{thm:potency}
Let $W_j$ have width $L$ with $2d_T\vee C\underline\gamma^{-2}\sigma^2\log(k^0/\varrho'_T) \le L\le\bar t_{\max}\wedge(\Delta_T/2)$ and be centred at $s^0_j$. Then, on the event of Proposition~\ref{prop:coverage},
\[
 \mathrm{PIP}^{W_j}\;\ge\;1-\frac{1}{c_-\varrho'_T}
 \exp\Bigl\{-\tfrac1{16}(\gamma^0_j)^2L/\sigma^2\Bigr\}-\Post{\mathcal B_{R_T}}
 \;\longrightarrow\;1,
\]
and all $k^0$ windows are covered simultaneously provided $\gamma_T^2L/\sigma^2\gg\log(k^0/\varrho'_T)$. Since $\bar t_{\max}=\lceil1+c_0/\kappa\rceil-1$ implies $\kappa<c_0/(\bar t_{\max}-1)$, the threshold satisfies $P^L\le\{c_1+c_0(L-1)/(\bar t_{\max}-1)\}/(c_0+c_1)$, which stays bounded below one whenever $L=o(\bar t_{\max})$. The rule therefore reports $W_j$ with probability tending to one.
\end{theorem}

\begin{proof}
$A^{W_j}$ holds whenever $D_j\le(L-1)/2$. Apply Proposition~\ref{prop:loc} at $d=\lfloor(L-1)/2\rfloor+1\ge L/2\ge d_T$, and a union bound over $j$ for the simultaneous statement.
\end{proof}

\begin{remark}[The cap must grow, and what that costs]\label{rem:nearmiss}
Theorem~\ref{thm:potency} requires an intermediate width $L\to\infty$ with $L=o(\bar t_{\max})$, hence $\bar t_{\max}\to\infty$ and, by $\kappa<c_0/(\bar t_{\max}-1)$, a vanishing width penalty $\kappa=\kappa_T\to0$. Since $L\gtrsim\sigma^2\underline\gamma^{-2}\log(1/\varrho'_T)$, which is of order $\log K$ under a Beta--Bernoulli prior, consistency needs $\bar t_{\max}\gg\log K$. This is not satisfied by the default $(c_0,c_1,\kappa)=(1,1,1/3)$, which gives $\bar t_{\max}=3$: in the empirical application ($N=15$ countries, $T=24$ years) $\log K=5.8$, and $\log K=7.4$ at $N=30,T=60$. At a fixed cap the conclusion weakens to $\liminf_T\mathrm{PIP}^{W_j}\ge1-Ce^{-c\bar t_{\max}\gamma_T^2/\sigma^2}$. The asymptotic statement and the default calibration therefore describe different regimes: windows reported under $(1,1,1/3)$ do not inherit the consistency result of Theorem~\ref{thm:window_consistency_oa}.
\end{remark}

\subsection{Posterior convergence}

\begin{corollary}[Posterior model size]\label{cor:size}
Under the conditions of Theorem~\ref{thm:gauge}, $\mathbb E[\,|m|\mid z\,]\le2k^0h/\ell_T+\mathbb E[\mathcal F_h\mid z] =O(k^0h/\ell_T)+o_p(1)$.
\end{corollary}

\begin{proof}
$\mathbb E[|m|\mid z]=\sum_s\Post{\delta_s=1}$. By Assumption~\ref{as:space} at most $\lceil 2h/\ell_T\rceil$ included dates lie within $h$ of any one true break, each contributing at most one. The remainder is $\mathbb E[\mathcal F_h\mid z]$, which is the \emph{unrestricted} far-date count because $\mathcal F^{\mathrm{iso}}_h=\mathcal F_h$ under Assumption~\ref{as:space}. A bound on the isolated count alone would omit clustered far dates.
\end{proof}

\begin{theorem}[Contraction of the mean function]\label{thm:contraction}
Let the conditions of Theorem~\ref{thm:gauge} and Proposition~\ref{prop:loc} hold, let $\omega_T\to\infty$ arbitrarily slowly, and put
\[
 \epsilon_T^2\;:=\;\frac{C}{T}\Bigl\{\underbrace{\frac{\bar\gamma^2\sigma^2k^0\log(NT)}{\gamma_T^2}}_{\text{localisation}}
 \;+\;\underbrace{\frac{\sigma^4(k^0)^2\log(KT)\,\omega_T}{\underline\gamma^{2}\ell_T}}_{\text{estimation}}\Bigr\}.
\]
Then, for the posterior draw $\mu=\alpha\mathbf 1+\sum_{s\in m}\gamma_su_s$ including the intercept, $\Post{T^{-1}\lVert\mu-\mu^0-\bar\mu^0\mathbf 1\rVert^2>\epsilon_T^2}\overset{p}\to0$. When $k^0=O(1)$ and $\gamma_T\asymp\bar\gamma\asymp\sigma$ this is the usual rate $\epsilon_T\asymp\{k^0\log(NT)/T\}^{1/2}$ for piecewise-constant regression with $k^0$ segments.
\end{theorem}

\begin{proof}
Decompose $\lVert\mu-\mu^0\rVert^2\le2\lVert(I-H_m)\mu^0\rVert^2+2\lVert\mu-H_m\mu^0\rVert^2$, both terms measured on the orthogonal complement of $\mathbf 1$, which is why the intercept is carried in $\mu$ and $\mu^0$ is centred.

\emph{Approximation error.} On $\{\max_jD_j\le d\}$ each true break has an included date within $d$, so by Proposition~\ref{prop:geom}(i) the residual of $\gamma^0_ju_{s^0_j}$ after projection is supported in the cell containing $s^0_j$ with squared norm at most $(\gamma^0_j)^2\min(s^0_j-a,b-s^0_j)\le\bar\gamma^2d$, and the residuals for different $j$ are supported on disjoint cells once $\Delta_T>2d$, so $\lVert(I-H_m)\mu^0\rVert^2\le\bar\gamma^2k^0d$. Taking $d=C\sigma^2\gamma_T^{-2}\log(k^0NT/\varrho'_T)\ge d_T$ and applying Proposition~\ref{prop:loc} with a union bound over $j$ makes $\Post{\max_jD_j>d}=o_p(1)$.

\emph{Estimation error.} $\mu-H_m\mu^0$ lies in $V_m$, of dimension $|m|+1$, and equals the projection of $\varepsilon$ onto $V_m$ plus the deviation of the posterior draw from that projection. Conditionally on $m$ both have squared norm $O_p(\sigma^2(|m|+1))$. By Corollary~\ref{cor:size} and Markov's inequality, $\Post{|m|>\omega_Tk^0h/\ell_T}\le1/\omega_T\to0$, and $h\asymp R_T$ gives $k^0h/\ell_T\asymp\sigma^2\underline\gamma^{-2}(k^0)^2\log(KT/\eta)/\ell_T$. Adding the two terms and dividing by $T$ gives the claim.
\end{proof}

\begin{corollary}[Convergence of the reported windows]\label{cor:hausdorff}
Let $\widehat{\mathcal W}$ be the union of all windows $W$ of width at most $\bar t_{\max}$ that satisfy $\mathrm{PIP}^W>P^{|W|}$, let $\hat k$ be its number of connected components, and let $d_H(A,B):=\max\{\sup_{a\in A}\inf_{b\in B}|a-b|,\sup_{b\in B}\inf_{a\in A}|a-b|\}$. Under Assumptions~\ref{as:nlp}--\ref{as:space}, with $R_T\le h\le\bar t_{\max}$ and $\bar t_{\max}=o(\Delta_T)$,
\[
 \PP\bigl(\hat k=k^0\bigr)\longrightarrow1,\qquad
 \PP\Bigl(\tfrac1T\,d_H\bigl(\widehat{\mathcal W},\mathcal T^0\bigr)
 \le\tfrac{2\bar t_{\max}}{T}\Bigr)\longrightarrow1 ,
\]
so the reported break \emph{fractions} converge to $\{\lambda^0_j\}=\{s^0_j/T\}$ at rate $\bar t_{\max}/T$.
\end{corollary}

\begin{proof}
By Theorem~\ref{thm:gauge} and Markov's inequality, $\Post{\exists s\in m:\operatorname{dist}(s,\mathcal T^0)>h}=o_p(1)$. Since $\mathrm{PIP}^W\le\sum_{s\in W}\Post{\delta_s=1}$, every window $W$ with $\operatorname{dist}(W,\mathcal T^0)>h$ has $\mathrm{PIP}^W=o_p(1)$, which is below $P^{|W|}\ge c_1/(c_0+c_1)>0$, and is not reported. Here it matters that Assumption~\ref{as:space} makes $\mathcal F_h$, not merely its isolated part, the object bounded in \eqref{eq:gaugebound}: a pair of adjacent far dates has $n_s\le1$ for both members and so contributes nothing to $\mathcal F^{\mathrm{iso}}_h$, so a bound on the isolated count alone would not exclude a spurious component. Every reported $W$ therefore meets $\{s:\operatorname{dist}(s,\mathcal T^0)\le h\}$ and, having width at most $\bar t_{\max}$, satisfies $W\subseteq\{s:\operatorname{dist}(s,\mathcal T^0)\le h+\bar t_{\max}\le2\bar t_{\max}\}$, which gives one half of $d_H$. By Theorem~\ref{thm:potency} each $s^0_j$ lies in a reported window, which gives the other half and shows each of the $k^0$ sets $\{s:|s-s^0_j|\le2\bar t_{\max}\}$ is met. These sets are pairwise disjoint because $\bar t_{\max}=o(\Delta_T)$, so $\widehat{\mathcal W}$ has exactly $k^0$ components.
\end{proof}

\begin{theorem}[Consistency of the window rule]\label{thm:window_consistency_oa}
Let Assumptions~\ref{as:nlp}--\ref{as:space} hold, let the reporting cap satisfy $\bar t_{\max}\to\infty$ and $\bar t_{\max}=o(\Delta_T)$, equivalently let the cost ratio $\kappa/c_0$ vanish slowly, and let $R_T\le h\le\bar t_{\max}$. Then, as $T\to\infty$:
\begin{enumerate}
\item[\textup{(i)}] \emph{(Potency.)} for each true break $s^0_j$ there is a window
$W_j\ni s^0_j$ of width at most $\bar t_{\max}$ with $\mathrm{PIP}^{W_j}\overset{p}\to1$, whose threshold $P^{|W_j|}$ stays bounded away from one, so the rule \eqref{eq:window_rule_app} reports $W_j$ with probability tending to one;
\item[\textup{(ii)}] \emph{(Gauge.)} every window $W$ of width at most
$\bar t_{\max}$ with $\operatorname{dist}(W,\mathcal T^0)>\bar t_{\max}$ satisfies $\mathrm{PIP}^{W}\overset{p}\to0$ and is not reported.
\end{enumerate}
Consequently $\widehat{\mathcal W}$ has exactly $k^0$ connected components, one containing each true break and each localised to within $2\bar t_{\max}$ of the corresponding true date, with probability tending to one.
\end{theorem}

\begin{proof}
Part~(i) is Theorem~\ref{thm:potency}, applied at an intermediate width $L\to\infty$ with $L=o(\bar t_{\max})$, which exists by $\bar t_{\max}\to\infty$ and is admissible by Assumption~\ref{as:design}. The threshold statement is the last display of that theorem. Part~(ii) and the final claim are Corollary~\ref{cor:hausdorff} and its proof, at $h=\bar t_{\max}$. What the theorem does \emph{not} assert, and what Proposition~\ref{prop:resolution} shows is unavailable at fixed signal, is convergence of the reported window to a point. Remark~\ref{rem:nearmiss} records what survives when the cap is held fixed at its default value.
\end{proof}

\subsection{Resolution: why a window and not a date}

\begin{proposition}[No exact recovery at fixed signal]\label{prop:resolution}
Fix $r\ge1$ with $r=o(\Delta_T)$ and let $m_0,m_r$ be two configurations differing only in that a date sits at $s^0_j$ rather than at $s^0_j+r$, both isolated with flanking dates $a<s^0_j<b$. Then
\begin{enumerate}
\item[\textup{(i)}] the energy left by the displacement is
$(\gamma^0_j)^2\,r\,(n_1-r)/n_1=(\gamma^0_j)^2r\{1+o(1)\}$, where $n_1=s^0_j+r-a$ is the length of the segment containing the mismatch;
\item[\textup{(ii)}] consequently $\log\mathrm{BF}(m_0:m_r)=O_p(1)$: the two
configurations have the same size, hence equal prior mass under Assumption~\ref{as:prior2}, and their posterior odds are bounded in probability. No energy-based argument can therefore deliver $\Post{D_j=0}\to1$ at fixed signal;
\item[\textup{(iii)}] Proposition~\ref{prop:loc} gives $D_j=O_p(\log(KT))$, not
$O_p(1)$: the multiplicity factor $1/\varrho'_T$ is what limits the resolution.
\end{enumerate}
\end{proposition}

\begin{proof}
(i) By Proposition~\ref{prop:geom}(i) the mismatch is confined to the $r$ periods $[s^0_j,s^0_j+r)$ inside the segment $[a,s^0_j+r)$ of length $n_1$, on which the fit is the segment mean. The residual sum of squares of a two-level contrast splitting $n_1$ periods into $n_1-r$ and $r$ is $(\gamma^0_j)^2r(n_1-r)/n_1$. The relevant length is $n_1$, the segment carrying the mismatch, and not $b-a$. (ii) is immediate from (i) and Lemma~\ref{lem:scalar}, and (iii) from Proposition~\ref{prop:loc} at $d\asymp\sigma^2\underline\gamma^{-2}\log(1/\varrho'_T)$.
\end{proof}

\begin{corollary}[Exact recovery under a growing signal]\label{cor:exactrec}
Suppose $\gamma_T^2 / \{\sigma^2\log(NT^3/\eta)\}\to\infty$, so that Assumption~\ref{as:design}'s requirement $\bar\gamma<\infty$ is replaced by a growing signal. Then $d_T\le1$ and Proposition~\ref{prop:loc} at $d=1$ with a union bound over $j$ gives $\Post{\max_jD_j\ge1}\le k^0(c_-\varrho'_T)^{-1}e^{-\gamma_T^2/8\sigma^2}\to0$: $\mathcal T^0$ is recovered exactly and a fixed cap $\bar t_{\max}=1$ suffices.
\end{corollary}

Part~(ii) of Proposition~\ref{prop:resolution} is the Bayesian counterpart of a familiar frequentist fact. At fixed break magnitude the least-squares change-point date estimator is $O_p(1)$ rather than consistent \citep{bai2003critical}, because the competing configurations then differ by an amount that does not grow with the sample, whatever asymptotic regime is adopted. It is why Corollary~\ref{cor:hausdorff} is stated for the window set and in the fraction scale.

\subsection{What is assumed, and why}

\begin{remark}[The clustering limit and the width penalty]\label{rem:cluster}
A candidate adjacent to an included date has $n_s=O(1)$ by Proposition~\ref{prop:geom}(i), exactly $0.995$, $1.980$ and $2.956$ at separations $1$, $2$ and $3$ in a sample of $T=500$, so by \eqref{eq:bfrep} its Bayes factor is $O(1)$ for \emph{any} slab: no prior on $\gamma$ suppresses it at a fixed practical-significance threshold, the only escape being $\tau_T\to\infty$, which is the growing-signal regime of Corollary~\ref{cor:exactrec}. Assumption~\ref{as:space} removes such configurations from the model space, and this is what Corollary~\ref{cor:size} and Corollary~\ref{cor:hausdorff} rest on.

It is worth recording why the restriction is not merely convenient. One would like to argue instead that the model prior suppresses clusters: peeling a run of $q$ adjacent far dates, the last removal is isolated and contributes the factor of Lemma~\ref{lem:scalar}(i) while the other $q-1$ contribute $\varrho_T\cdot O(1)$ each, so that summing over the $O(T)$ run positions gives a convergent geometric series whenever $\varrho_T\to0$. That argument is not available with the bounds proved here: uniformly over configurations, Lemma~\ref{lem:scalar} only gives $\mathrm{BF}_s\le e^{t_0^2/2}$ when $\lambda_s=O(1)$, and $e^{t_0^2/2}\ge2NT^3/\eta$, so each clustered date costs $\varrho_Te^{t_0^2/2}\gg1$ and the series diverges. A multi-coordinate treatment of a whole run faces the same uniformity cost. Under a fixed $\omega$ the prior ratio is the constant $\omega/(1-\omega)$ and there is no suppression at all. Assumption~\ref{as:space}, or an explicit assumption that the posterior mass of clustered far configurations vanishes.

Two consequences survive unchanged. Spurious adjacent dates inflate the \emph{width} of a reported window rather than creating additional windows, which is what the width penalty $\kappa$ prices. The three failure modes of a saturation procedure map onto three distinct devices: isolated false positives onto the non-local slab (Theorem~\ref{thm:gauge}), clustered false positives onto $\kappa$ and the collapsing in Assumption~\ref{as:space}, and missed breaks onto signal strength (Theorem~\ref{thm:potency}).
\end{remark}

\begin{remark}[Relation to the high-dimensional non-local prior literature]\label{rem:eigenvalue}
The model-selection consistency results of \citet{rossell2012nlp_modsel} and \citet{rossell2017nonlocal} are not used, and their regularity conditions are unavailable here: the condition $\lambda_{\min}(X_n'X_n)>nc$ of \citet[Theorem~1]{rossell2012nlp_modsel} is imposed on the full candidate matrix, which contains near-duplicate columns ($\lVert u_s-u_{s+1}\rVert^2=1$), while condition~D2 of \citet{rossell2017nonlocal} is imposed per model and fails for every configuration of growing size, since Proposition~\ref{prop:geom}(i) gives $\lambda_{\min}(X_m'X_m)\le\tfrac12\min_ig_i\le T/\{2(|m|-1)\}$. For $T=400$ and $|m|=40$ equally spaced dates, $\lambda_{\min}=1.63$ against the bound $5.13$. What \emph{is} used from that literature is the factorisation $M=M^LG_m$ of \citet[Proposition~1]{rossell2017nonlocal}, in the form of Lemma~\ref{lem:local} and Assumption~\ref{as:nlp}, together with the univariate rate in the explicit form of Lemma~\ref{lem:scalar}. Because the slab of \eqref{eq:ssvs_pri} is applied independently to each $\gamma_{i,s}$, the joint slab is the product (piMOM) density and the per-coordinate argument is what the product structure licenses. The multivariate iMOM of \citet[Section~3.3]{rossell2010nlp_test} vanishes only when every coordinate does, and Theorem~\ref{thm:gauge} would fail under it.
\end{remark}

%------------------------------

\section{Prior Parametrization and Posterior Sampling Sampler}\label{app:prior_posterior}

This section presents details on the default parametrization and hyperparameter choices for BISAM as well as the sampling strategy for posterior inference.

\subsection{Prior Parametrization}\label{app:prior_settings}

%------------------------------------------------------------------------------
% TABLE -- HYPERPARAMETERS
%------------------------------------------------------------------------------

Table~\ref{tab:hyperpars} collects the default hyperparameter settings of BISAM. Where possible, defaults follow the recommendations for non-local priors in \citet{rossell2010nlp_test}, \citet{rossell2012nlp_modsel} and \citet{rossell2017nonlocal}, adapted to the indicator-saturation setting. The two hyperparameters that matter most in practice are the slab scale $\tau$ and the
prior inclusion probability $\omega$. Both are given an interpretable calibration rather than a conventional default. In particular, $\tau$ is set by inverting $P(|\gamma_{i,s}| \leq \bar{\gamma} \mid \tau) = \alpha$ for a practically relevant threshold $\bar{\gamma}$, which is the indicator-saturation analogue of the ``practical significance" calibration of \citet{rossell2010nlp_test}, who set $\tau = 0.133$ for the piMOM prior so that $P(|\theta_i|/\sigma \leq 0.2) = 0.01$. Since a break of less than one error standard deviation is of little substantive interest in the applications considered here, we use $\bar{\gamma} = \sigma$ instead of $\bar{\gamma} = 0.2\sigma$, which yields the markedly larger values reported below.

\begin{table}[htbp!]
\centering
\footnotesize
\setlength{\tabcolsep}{4pt}
\renewcommand{\arraystretch}{1.2}
\caption{Standard hyperparameter settings for BISAM.}
\label{tab:hyperpars}
\begin{tabular}{@{}>{\raggedright\arraybackslash}p{2.4cm}>{\raggedright\arraybackslash}p{3.9cm}>{\raggedright\arraybackslash}p{2.6cm}>{\raggedright\arraybackslash}p{4.4cm}@{}}
\toprule
\textbf{Parameter} & \textbf{Role} & \textbf{Default} & \textbf{Rationale / source} \\
\midrule
\multicolumn{4}{@{}l}{\textit{Panel A: Break function (slab)}} \\
\addlinespace[2pt]
$\gamma_0$ & Centring value of the iMOM slab & $0$ & Non-locality around the Null of no break. \\
$k$ & Order of the iMOM prior & $1$ & \citet{rossell2010nlp_test}, controlling the rate at which the density vanishes at $\gamma_0$. \\
$\nu$ & Shape (tail index) of the iMOM slab & $1$ & Cauchy-like tails, minimal shrinkage of large breaks. \\
$\tau$ & Slab scale & $1.92$ \newline ($3.32$) & Solves $P(|\gamma_{i,s}| \leq \sigma \mid \tau) = 0.05$ ($0.01$). Analogue of the default in \citet{rossell2010nlp_test}. \\
$\omega$ & Prior inclusion probability of a step-shift & $0.5$ & Uniform prior on the model space. Minimizes misclassification (median probability model of \citet{barbieri2004optimal}). \\
$(\zeta_0,\zeta_1)$ & Beta--Bernoulli hyperparameters (if $\omega$ is treated hierarchically) & $(1,1)$ & Agnostic sparsity-inducing alternative to fixed $\omega$, see \citet{scott2010bayes}. \\
\addlinespace[3pt]
\multicolumn{4}{@{}l}{\textit{Panel B: Mean function and error variance}} \\
\addlinespace[2pt]
$\underline{{b}}$ & Prior mean of ${\beta}$ & Break-free OLS estimate & $f$-prior \citep{ohagan1995}; estimated on observations below the $90$th percentile of the squared OLS residuals to limit outlier contamination. Set to $\bm{0}$ under the $g$-prior \citep{zellner1986}. \\
$g$ & Prior variance scale of $\beta$ & $100$ & Weakly informative to control only the nuisance coefficients. \\
$m_0$ & IG shape of $\sigma^2_i$ & $3$ & Ensures a finite prior variance. \\
$n_0$ & IG rate of $\sigma^2_i$ & Solves $P(\sigma^2_i \leq \hat{\sigma}^2_{OLS}) = 0.9$ & Empirical-Bayes calibration on the break-free benchmark residuals \citep[cf.][]{chipman2010}. Uninformative alternative: $m_0 = n_0 = 0.01$. \\
\addlinespace[3pt]
\multicolumn{4}{@{}l}{\textit{Panel C: Outlier component}} \\
\addlinespace[2pt]
$(a_{\eta}, b_{\eta})$ & Beta prior on the outlier probability $\eta$ & $(1,10)$ & Prior mean of $\approx 9\%$ outlying observations; deliberately sparse. \\
$\tau^{\varepsilon}$ & Scale of the iMOM error component & $10$ & Implies a variance inflation of $2\tau^{\varepsilon} = 20$ and a downweighting factor of $1/\sqrt{2\tau^{\varepsilon}} \approx 0.22$ for flagged observations. \\
$(\nu, k)$ & Shape/order of the iMOM error component & $(3,1)$ & $\nu = 3$ is required for the variance of the iMOM component to exist while maintaining heavy tails. \\
% \addlinespace[3pt]
% \multicolumn{4}{@{}l}{\textit{Panel D: MCMC and decision rules}} \\
% \addlinespace[2pt]
% $R$ & Total number of draws & $10{,}000$ & \\
% $R_0$ & Burn-in draws & $2{,}000$ & \\
% -- & Draws per inner NLP block & $1$ retained \newline ($1$--$2$ burn-in for thinning) & Model-space and truncation samplers are embedded in the outer Gibbs sweep. \\
% $P$ & PIP threshold for individual break selection & $0.5$ & Median probability model \citep{barbieri2004optimal}; Bayes rule under symmetric classification losses ($c_0 = c_1$, $\bar{t} = 1$). \\
% $(c_0, c_1, \kappa)$ & Costs of a missed break, a spurious window, and one excess period of window width & $(1, 1, 1/3)$ & Symmetric classification losses; $\kappa/c_0 = 1/3$ implies that three excess periods of imprecision are as costly as one missed break. \\
% $P^{\bar{t}}$ & PIP threshold for a window of length $\bar{t}$ & $\frac{c_1 + \kappa(\bar{t}-1)}{c_0 + c_1}$ & Bayes rule under the width-penalised loss in \ref{eq:window_rule_app}. \\

\bottomrule
\end{tabular}

\vspace{0.6em}
\begin{minipage}{\textwidth}
\footnotesize
\textit{Notes:} Defaults refer to the standard \texttt{BISAM} implementation, with values in brackets for optional parts. 
% In the simulation study of Section~\ref{sec:simulation} we set $\tau = 3.32$, corresponding to $P(|\gamma| \leq \sigma \mid \tau) = 0.01$, in the empirical application of Section~\ref{sec:applications} we set $\tau = 1.92$, corresponding to $P(|\gamma| \leq \sigma \mid \tau) = 0.05$. 
% Values of $\tau$ are obtained by numerically inverting the prior cdf of the iMOM density in equation~\eqref{eq_pimom}. Note that the default step scales here reflect the threshold $\bar{\gamma} = \sigma$ that is natural when the regressors are $\{0,1\}$ step-shift indicators.
\end{minipage}
\end{table}

\subsection{Posterior Sampler}\label{app:posterior_sampler}

Posterior inference in BISAM relies on a Markov chain Monte Carlo (MCMC) algorithm which combines standard Gibbs updates for the mean function and the error variance with two non-local-prior (NLP) blocks: a model-space sampler for the break indicators $\bm{\delta}^{\gamma}$ \citep[in the spirit of][]{rossell2012nlp_modsel} and a latent-truncation Gibbs sampler for the break magnitudes $\bm{\gamma}$ \citep[following][]{rossell2017nonlocal}. Throughout, $\mid \bullet$ denotes conditioning on all remaining parameters and on the data, and $\hat{\varepsilon}_{i,t} = y_{i,t} - x_{i,t}'\bm{\beta} - \sum_{s} \gamma_{i,s} \mathbbm{1}_{\{t \ge s\}}$ denotes the current residual. Let
\begin{equation}
 \lambda_{i,t} = (1-\delta^{\varepsilon}_{i,t}) + \delta^{\varepsilon}_{i,t}\, 2\tau^{\varepsilon}
 \label{eq:lambda_inflation}
\end{equation}
denote the observation-specific variance inflation factor implied by the outlier component of Section~\ref{subsec:outlier}, so that, conditional on $\delta^{\varepsilon}_{i,t}$, the working error variance is $\sigma^2_i \lambda_{i,t}$ and the model can be treated as a heteroskedastic Gaussian regression. All model-selection and break-magnitude blocks operate on the GLS-transformed (i.e.\ variance-standardised) data $\tilde{y}_{i,t} = (y_{i,t} - x_{i,t}'\bm{\beta}) / \sqrt{\sigma^2_i \lambda_{i,t}}$ and $\tilde{z}_{i,t} = z_{i,t} / \sqrt{\sigma^2_i \lambda_{i,t}}$, so that the working error variance in those blocks equals unity.

Because step-shift indicators are unit-specific, the design matrix $\bm{Z}$ is block-diagonal across cross-sectional units. The model-selection and break-magnitude blocks can therefore be run unit by unit ($i = 1, \dots, N$) on $T-3$ candidate indicators each, rather than jointly on all $q = N(T-3)$ candidates, which reduces the model space explored per sweep from $2^{q}$ to $N$ spaces of size $2^{T-3}$ without altering the target posterior.

Table~\ref{tab:sampler} summarises one sweep of the sampler. Starting values are obtained from a (break-free) benchmark regression: $\bm{\beta}^{(0)} = \hat{\bm{\beta}}_{OLS}$, 
$\bm{\gamma}^{(0)} = \bm{0}$, $\bm{\delta}^{\gamma,(0)} = \bm{0}$,  $\bm{\delta}^{\varepsilon,(0)} = \bm{0}$ and $\sigma^{2,(0)}_i \sim \mathcal{G}^{-1}(m_0, n_0)$ (trimmed, see Table~\ref{tab:hyperpars}).

%------------------------------------------------------------------------------
% TABLE A.1 -- POSTERIOR SAMPLING ALGORITHM
%------------------------------------------------------------------------------
\begin{table}[htbp!]
\centering
\footnotesize
\setlength{\tabcolsep}{4pt}
\renewcommand{\arraystretch}{1.25}
\caption{Blocks of the BISAM posterior sampler. One sweep of the Gibbs algorithm.}
\label{tab:sampler}
\begin{tabular}{@{}c>{\raggedright\arraybackslash}p{3.1cm}>{\raggedright\arraybackslash}p{8.5cm}@{}}
\toprule
\textbf{Step} & \textbf{Quantity} & \textbf{Full conditional and sampling scheme} \\
\midrule

1 & Outlier probability $\eta$ &
$\eta \mid \bullet \sim \mathcal{B}\!\left(a_{\eta} + \textstyle\sum_{i,t}\delta^{\varepsilon}_{i,t}, b_{\eta} + NT - \textstyle\sum_{i,t}\delta^{\varepsilon}_{i,t}\right)$.
Conjugate Beta--Bernoulli update. \\

2 & Outlier indicators $\delta^{\varepsilon}_{i,t}$ &
$\delta^{\varepsilon}_{i,t} \mid \bullet \sim \mathcal{B}er(\bar{\eta}_{i,t})$ with
$\bar{\eta}_{i,t} = \dfrac{\eta\, \pi_{iMOM}(\hat{\varepsilon}_{i,t}; 0,1,3,\tau^{\varepsilon}\sigma^2_i)}
{\eta\, \pi_{iMOM}(\hat{\varepsilon}_{i,t}; 0,1,3,\tau^{\varepsilon}\sigma^2_i) + (1-\eta)\,\pi_{N}(\hat{\varepsilon}_{i,t}; 0,\sigma^2_i)}$,
evaluated on the log scale. Update $\lambda_{i,t}$ according to \eqref{eq:lambda_inflation}. \\

3 & Error variance $\sigma^2_i$ &
$\sigma^2_i \mid \bullet \sim \mathcal{G}^{-1}\!\left(m_0 + \tfrac{T}{2},
 n_0 + \tfrac{1}{2}\sum_{t=1}^{T} \dfrac{\hat{\varepsilon}^2_{i,t}}{\lambda_{i,t}}\right)$
(unit-specific), or with $T$ replaced by $NT$ and the sum taken over all $(i,t)$ under a pooled variance. Weighting by $\lambda_{i,t}^{-1}$ removes the outlier inflation. \\

4 & Mean-function coefficients $\bm{\beta}$ &
$\bm{\beta} \mid \bullet \sim \mathcal{N}(\bar{\bm{b}}, \bar{\bm{B}})$ with
$\bar{\bm{B}} = \left(\bm{X}'\bm{W}\bm{X} + \underline{\bm{B}}^{-1}\right)^{-1}$ and
$\bar{\bm{b}} = \bar{\bm{B}}\left(\bm{X}'\bm{W}\bm{y}^{\ast} + \underline{\bm{B}}^{-1}\underline{\bm{b}}\right)$,
where $\bm{y}^{\ast} = \bm{y} - \bm{Z}\bm{\gamma}$ and
$\bm{W} = \mathrm{diag}\{(\sigma^2_i \lambda_{i,t})^{-1}\}$.
Under the $f$-prior, $\underline{\bm{B}}^{-1} = \bm{X}'\bm{W}\bm{X}/g$ and $\underline{\bm{b}}$ is the break-free (trimmed) OLS estimate. Under the $g$-prior, $\underline{\bm{b}} = \bm{0}$. For stability, drawn via a Cholesky factor of $\bar{\bm{B}}$. \\

5 & Break indicators $\bm{\delta}^{\gamma}_{i,\cdot}$ \newline (for each unit $i$) & One sweep of the model-space Gibbs sampler of \citet{rossell2012nlp_modsel} on the $T-3$ candidate indicators of unit $i$, applied to the standardised data $(\tilde{y}_{i,\cdot}, \tilde{\bm{Z}}_{i,\cdot})$. For each candidate $s$, the current model $\mathcal{M}_{curr}$ is compared to the model $\mathcal{M}_{cand}$ that flips $\delta^{\gamma}_{i,s}$. The flip is accepted with probability 
$r = \dfrac{
m_{cand}(\tilde{y})\,p(\mathcal{M}_{cand})
} {
m_{cand}(\tilde{y})\,p(\mathcal{M}_{cand}) + m_{curr}(\tilde{y})\,p(\mathcal{M}_{curr})
}$,
where $p(\mathcal{M}) \propto \omega^{|\mathcal{M}|}(1-\omega)^{T-3-|\mathcal{M}|}$. Marginal likelihoods $m(\cdot)$ under the iMOM slab have no closed form and are evaluated by (approximate) Laplace approximation \citep[Eq.~9--11 in][]{rossell2012nlp_modsel}. The block returns both the sampled model and the marginal inclusion probabilities (PIPs). \\

6 & Break magnitudes $\gamma_{i,s}$, $s: \delta^{\gamma}_{i,s}=1$ &
Conditional on the selected model, draw from the iMOM posterior using the latent-truncation representation of \citet{rossell2017nonlocal}, which writes
$\pi_{iMOM}(\bm{\gamma} \mid \sigma^2) = \mathcal{N}(\bm{\gamma}; \bm{0}, \tau_N \sigma^2 \bm{I}) \prod_{s} d(\gamma_{s}, \sigma^2)$ with
$d(\gamma_{s}, \sigma^2) = \pi_{iMOM}(\gamma_{s} \mid \sigma^2) / \mathcal{N}(\gamma_{s}; 0, \tau_N \sigma^2)$.
Setting $\tau_N = 2\tau$ renders $d(\cdot)$ monotone increasing in $\gamma_{s}^2$, so that $d^{-1}(\cdot)$ exists. With $\bm{S} = \tilde{\bm{Z}}'\tilde{\bm{Z}}$ and $\bm{m} = \bm{S}^{-1}\tilde{\bm{Z}}'\tilde{\bm{y}}$:
\newline
\hspace*{0.4em}(a) $\lambda^{\gamma}_{s} \sim \mathcal{U}(0,\, d(\gamma^{(k-1)}_{s}, \sigma^2))$ for every included $s$;
\newline
\hspace*{0.4em}(b) $\bm{\gamma} \sim \mathcal{N}(\bm{m}, \bm{S}^{-1}) \prod_{s} \mathbbm{1}\{ |\gamma_{s}| > d^{-1}(\lambda^{\gamma}_{s}; \sigma^2)\}$,
i.e. a multivariate Normal truncated away from a neighbourhood of the origin. The region is non-convex and draws are obtained after orthogonalisation, with $d^{-1}(\cdot)$ evaluated by an asymptotic approximation combined with a linear interpolation search. Excluded coefficients are set exactly to zero (Dirac spike).\\

% 7 & Storage &
% Retain $\bm{\beta}, \bm{\gamma}, \bm{\delta}^{\gamma}, \bm{\delta}^{\varepsilon}, \sigma^2$ and the PIPs after burn-in. Posterior statistics such as means and equal-tailed credible intervals are computed from the retained draws. \\
\bottomrule
\end{tabular}

\vspace{0.6em}
\begin{minipage}{\textwidth}
\footnotesize
\textit{Notes:} $\mathcal{B}$, $\mathcal{B}er$, $\mathcal{N}$, $\mathcal{U}$ and $\mathcal{G}^{-1}$ denote the Beta, Bernoulli, Normal, Uniform and inverse Gamma distribution, respectively, while $\pi_{N}$ and $\pi_{iMOM}$ denote the Gaussian and iMOM densities as defined in equation~\eqref{eq_pimom}. Steps 5 and 6 are executed unit by unit, exploiting the block-diagonal structure of $\bm{Z}$. Steps 1--2 are omitted when the model is estimated without the outlier component, in which case $\lambda_{i,t} \equiv 1$. Steps 5 and 6 are implemented through a tailored version of the \texttt{mombf} package \citep{rossell2017nonlocal} where within each sweep, the model-space sampler and the latent-truncation sampler are run for a single retained draw, which preserves the invariant distribution of the overall chain.
\end{minipage}
\end{table}

\clearpage

\section{Additional simulation results}
\label{app:add_res_sim}

\begin{figure}[ht!]
 \centering
 \includegraphics[width=.9\linewidth]{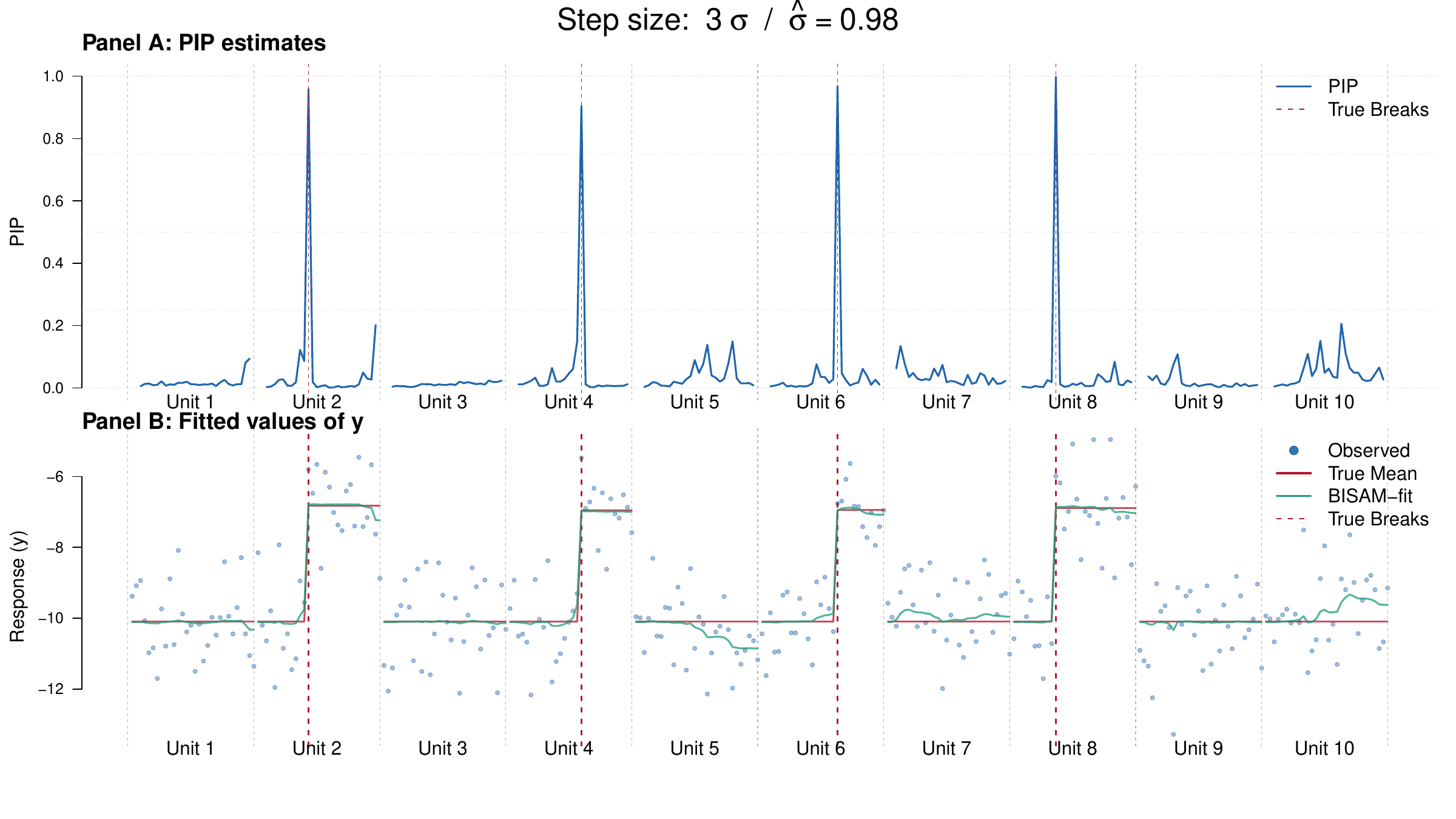}
 \caption{Sparse simulation setup. Posterior inclusion probabilities in the upper panel. Lower panel shows the simulated data points (blue dots), the true mean (red line) and the fit of BISAM. Red dashed lines indicate true breaks.}
 \label{fig:sim_setup_sparse}
\end{figure}

\begin{figure}[ht!]
 \centering
 \includegraphics[width=.9\linewidth]{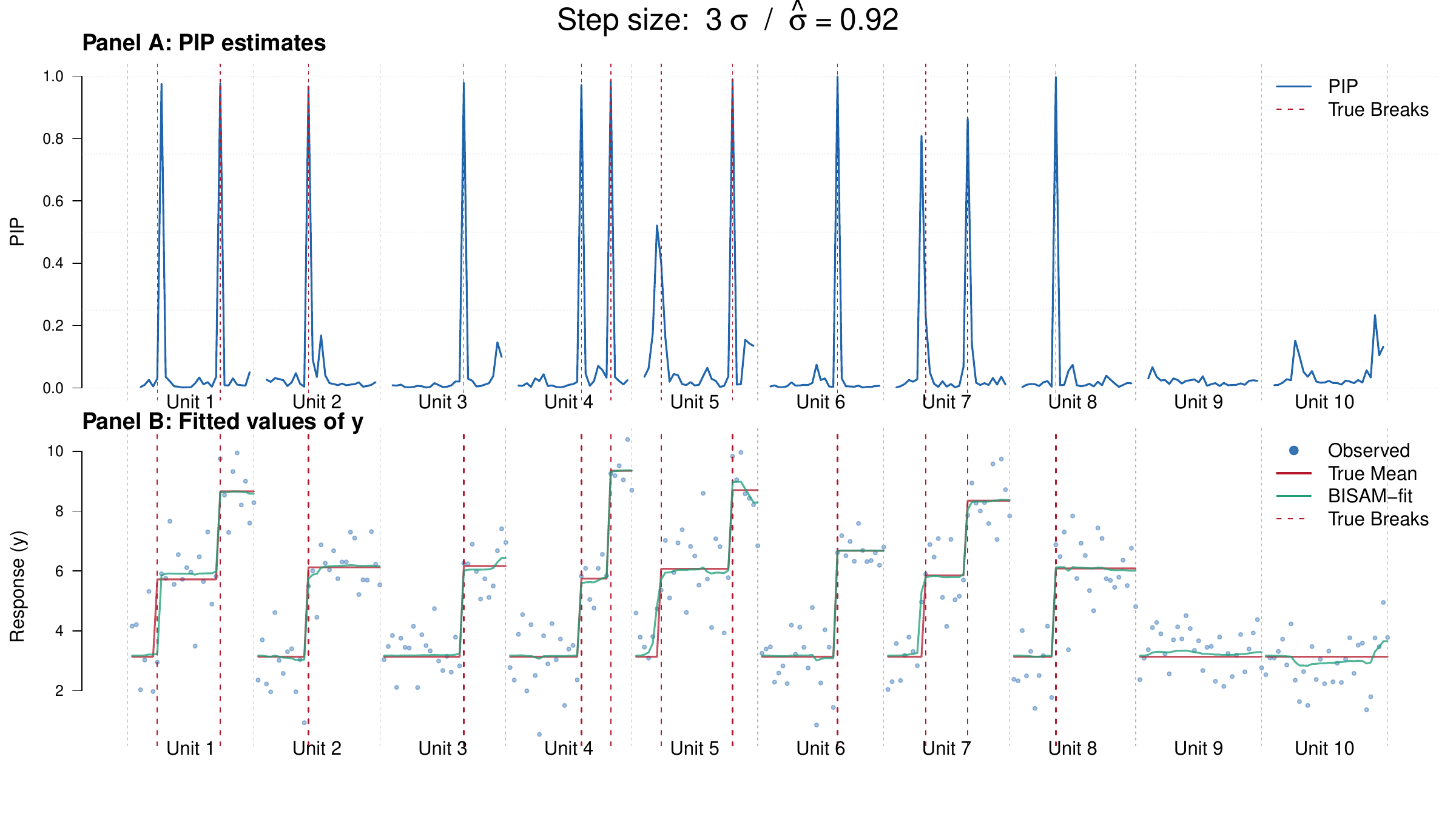}
 \caption{Dense simulation setup. Posterior inclusion probabilities in the upper panel. Lower panel shows the simulated data points (blue dots), the true mean (red line) and the fit of BISAM. Red dashed lines indicate true breaks.}
 \label{fig:sim_setup_dense}
\end{figure}

\begin{figure}[ht!]
 \centering
 \includegraphics[width=.9\linewidth]{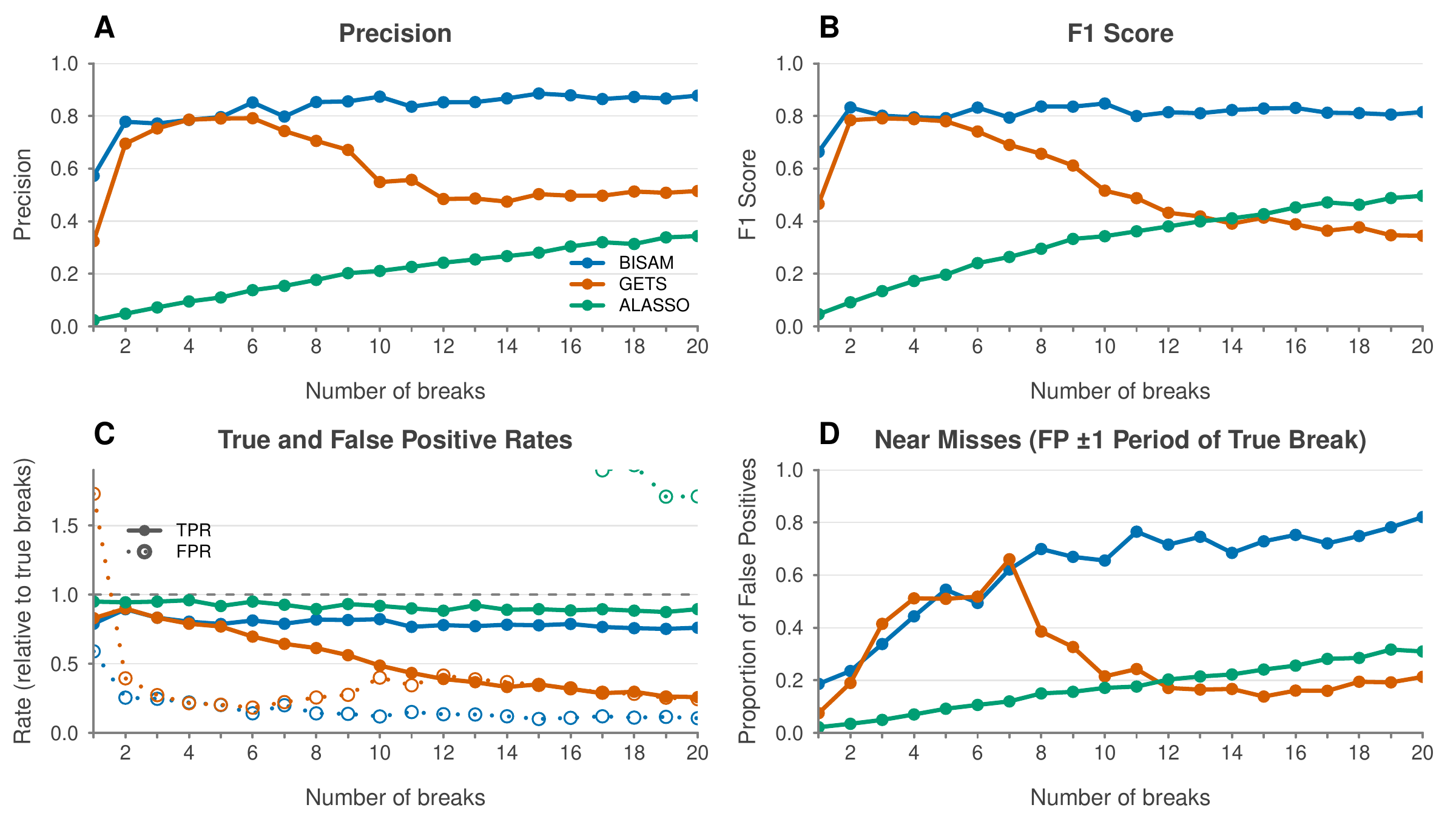}
 \caption{Performance metrics (\textcolor{ssvs}{BISAM}, \textcolor{gets}{GETS}, \textcolor{alasso}{ALASSO}) for simulations assuming a relative break size of three standard deviations of the error term and varying number of breaks. FPR of ALASSO in panel (C) not shown due to being exceeedingly high.}
 \label{fig:sim_results_bn}
\end{figure}

\begin{figure}[ht!]
 \centering
 \includegraphics[width=\linewidth, page = 1]{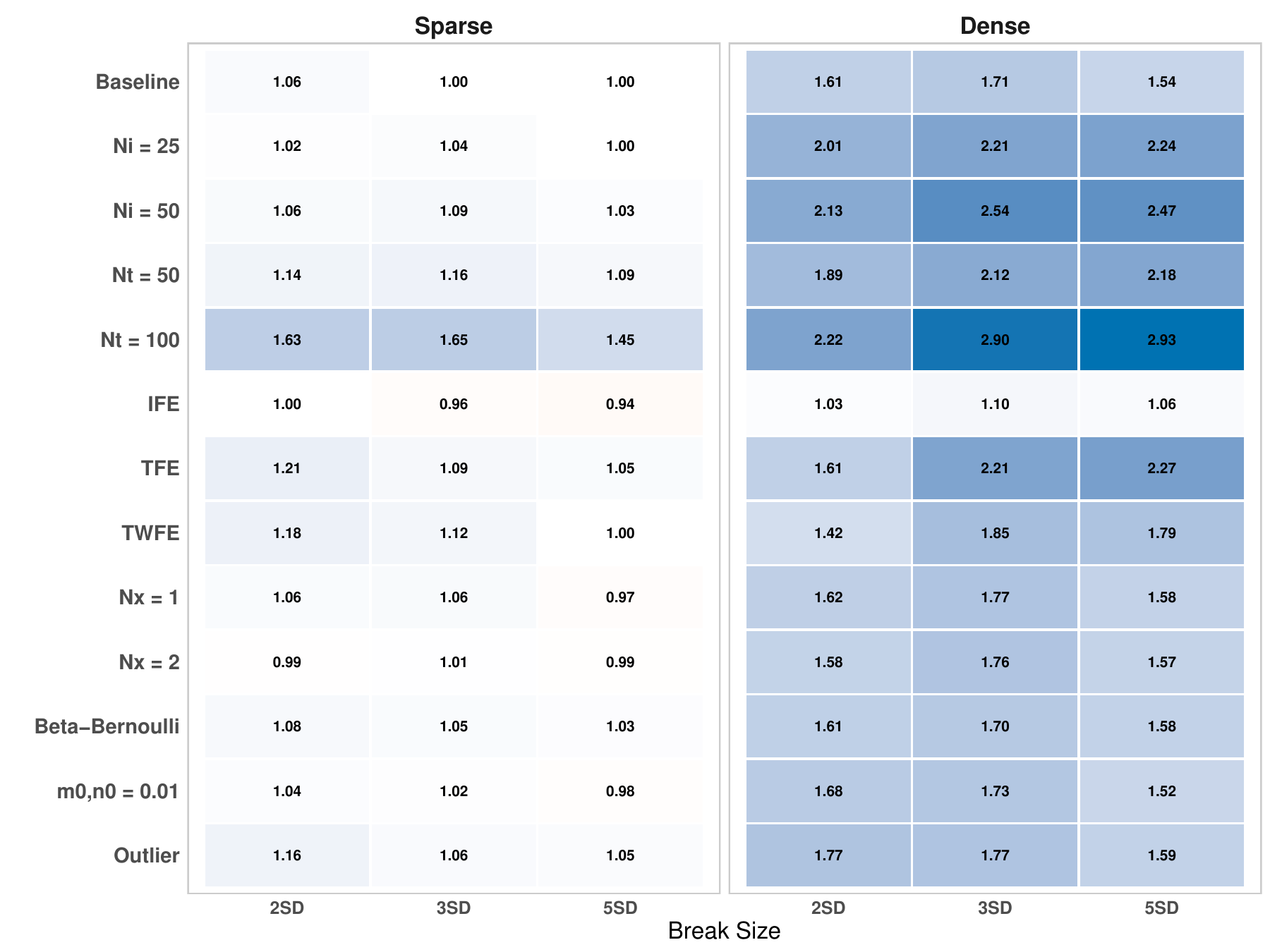}
 \caption{Simulation variations. Relative F1 scores of BISAM compared to GETS for different specification variations. 
 Numbers refer to ratio between F1 scores achieved by BISAM and GETS for a given specification. Blue (orange) cell shadings indicate relatively higher (lower) performance of BISAM compared to GETS. 
 Variations include baseline (1\textsuperscript{st} row), using an increased cross-section of 25 or 50 units (2\textsuperscript{nd} \& 3\textsuperscript{rd} rows), using an increased time dimension of 50 and 100 time periods (4\textsuperscript{th} \& 5\textsuperscript{th} rows), 
 inclusion of unit fixed effects (6\textsuperscript{th} row), time fixed effects (7\textsuperscript{th} row), two-way fixed effects (8\textsuperscript{th} row), inclusion of one or two external regressors (9\textsuperscript{th} \& 10\textsuperscript{th} rows), using a hierarchical Beta-Bernoulli prior setup for $\omega_{i,s}$ (11\textsuperscript{th} row), using uninformative prior values $m_0 = n_0 = 0.01$ for the error term (12\textsuperscript{th} row), using the outlier detection framework as detailed in Section~\ref{subsec:outlier} (13\textsuperscript{th} row).}
 \label{fig:sim_variations_f1_ratio}
\end{figure}

\begin{figure}[ht!]
 \centering
 \includegraphics[width=\linewidth, page = 2]{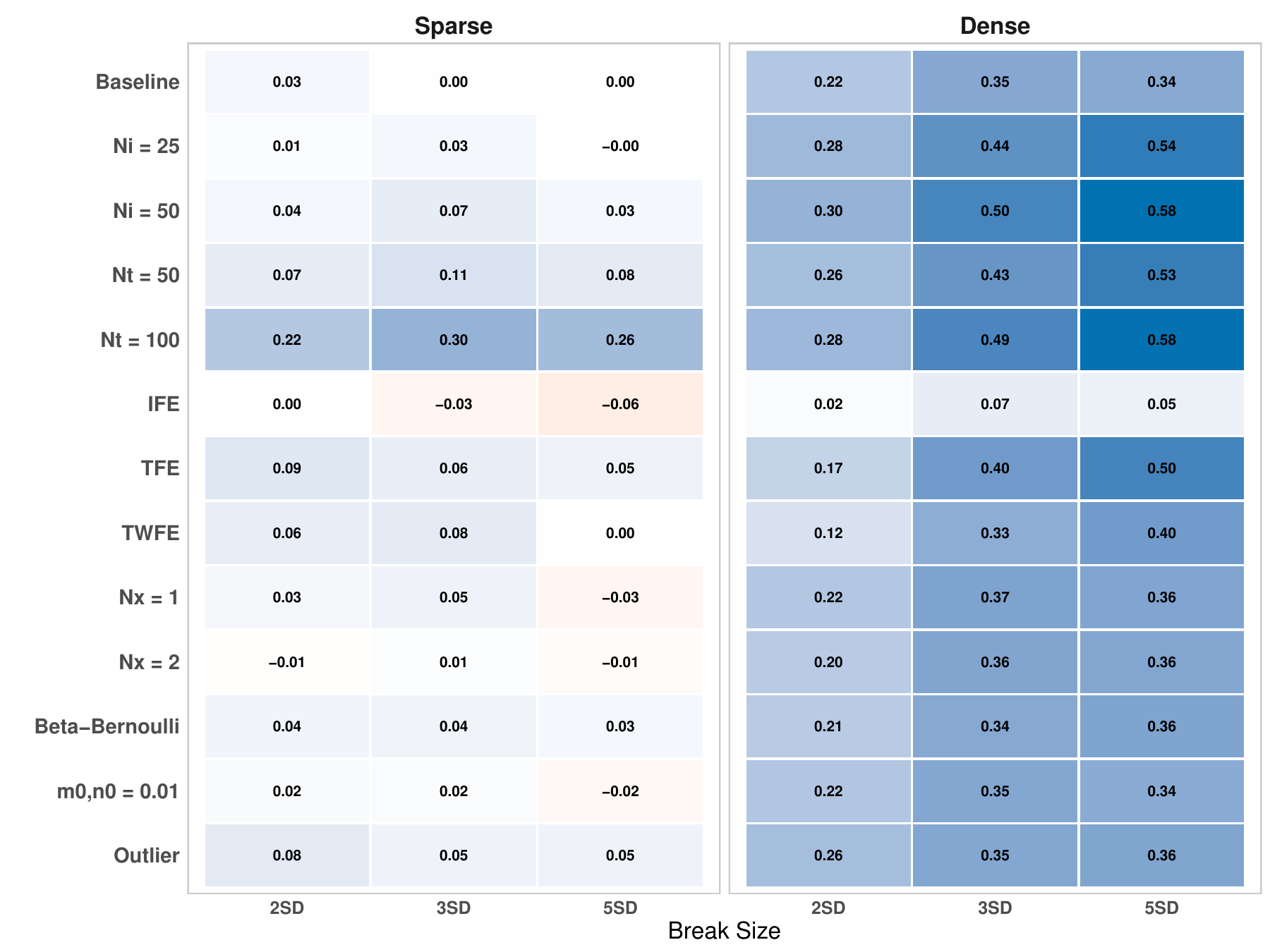}
 \caption{Simulation variations. Difference of F1 scores between BISAM and GETS for different specification variations. 
 Numbers refer to difference between F1 scores achieved by BISAM and GETS for a given specification. Blue (orange) cell shadings indicate relatively higher (lower) performance of BISAM compared to GETS.
 Variations include baseline (1\textsuperscript{st} row), using an increased cross-section of 25 or 50 units (2\textsuperscript{nd} \& 3\textsuperscript{rd} rows), using an increased time dimension of 50 and 100 time periods (4\textsuperscript{th} \& 5\textsuperscript{th} rows), 
 inclusion of unit fixed effects (6\textsuperscript{th} row), time fixed effects (7\textsuperscript{th} row), two-way fixed effects (8\textsuperscript{th} row), inclusion of one or two external regressors (9\textsuperscript{th} \& 10\textsuperscript{th} rows), using a hierarchical Beta-Bernoulli prior setup for $\omega_{i,s}$ (11\textsuperscript{th} row), using uninformative prior values $m_0 = n_0 = 0.01$ for the error term (12\textsuperscript{th} row), using the outlier detection framework as detailed in Section~\ref{subsec:outlier} (13\textsuperscript{th} row).}
 \label{fig:sim_variations_f1_diff}
\end{figure}

\clearpage

\section{Additional replication results}
\label{app:add_res_repl}

\begin{figure}[htbp!]
 \centering
 \includegraphics[width=\linewidth, height=0.25\textheight, keepaspectratio=false, page = 1]{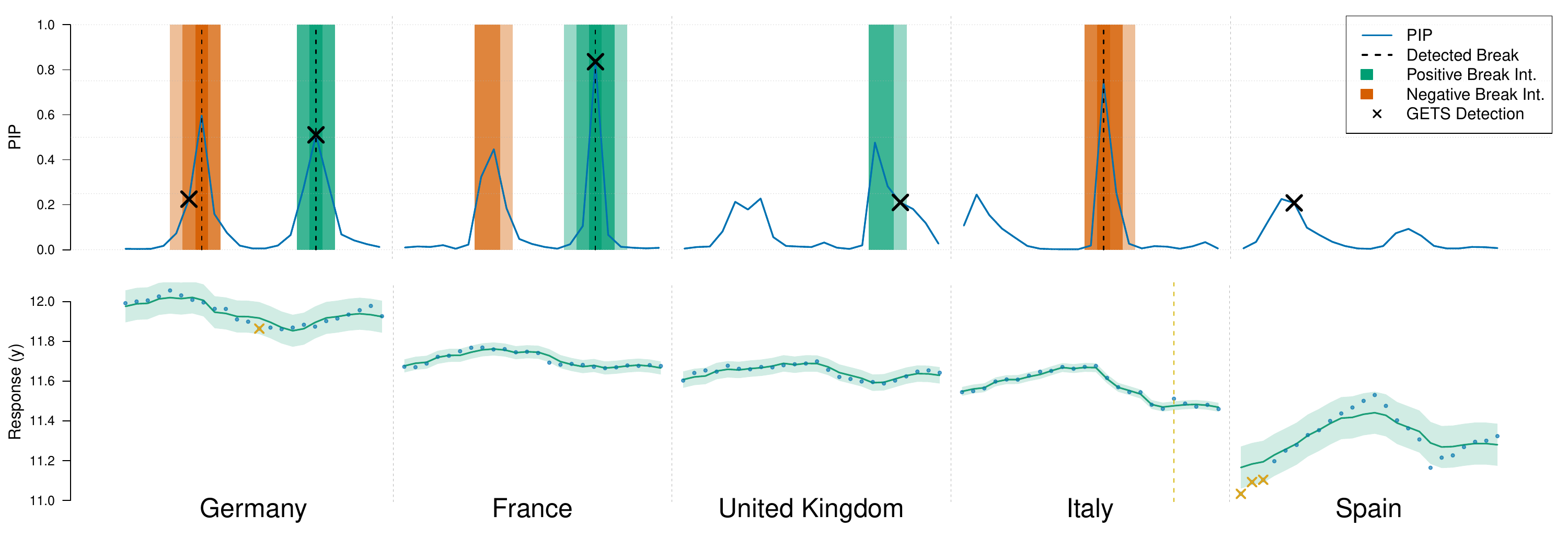}
 \includegraphics[width=\linewidth, height=0.25\textheight, keepaspectratio=false, page = 2]{img/revision/figure_5_outlier-TRUE-10_beta-f-10_sisprior-imom-1.92072941034706_modprior-beta_bern-0.5.pdf}
 \includegraphics[width=\linewidth, height=0.25\textheight, keepaspectratio=false, page = 3]{img/revision/figure_5_outlier-TRUE-10_beta-f-10_sisprior-imom-1.92072941034706_modprior-beta_bern-0.5.pdf}
 \caption{Break detection in EU transport emissions using a Beta-Bernoulli setup for the prior inclusion probabilities $\omega_{i,s}$. Breaks for individual dates detected are marked by vertical dashed lines (for BISAM) and black crosses (for GETS) in the upper panels, detected outliers are marked similarly in the lower panels. Green (orange) shaded areas indicate periods in which BISAM detects positive (negative) breaks for window lengths $\bar{t}_{\{1,2,3\}}$. Transparency of the shading indicates the length of windows under which a break is detected, with longer windows being more transparent.}
 \label{fig:replication_beta-bern}
\end{figure}

\newpage

\end{document}